\documentclass[11pt]{article}
\usepackage{times}
\usepackage{amsmath,amsfonts,amssymb,amsthm}
\usepackage[colorlinks=true,linkcolor=red,citecolor=blue,urlcolor=blue]{hyperref}
\usepackage{paralist}
\usepackage{framed}
\usepackage{verbatim}
\usepackage{cleveref}
\usepackage[dvipsnames]{xcolor}

\usepackage{mathptmx}
\usepackage[text={6.45in,8.95in}]{geometry}

\newenvironment{reminder}[1]{\smallskip
\noindent {\bf Reminder of #1. }\em}{\smallskip}

\theoremstyle{plain}

\newtheorem{theorem}{Theorem}[section]

\newtheorem{lemma}[theorem]{Lemma}
\newtheorem{proposition}[theorem]{Proposition}
\newtheorem{corollary}[theorem]{Corollary}
\newtheorem{claim}[theorem]{Claim}
\theoremstyle{definition}
\newtheorem{definition}[theorem]{Definition}
\newtheorem{remark}[theorem]{Remark}

\newenvironment{proofof}[1]{\smallskip
\noindent {\bf Proof of #1.  }}{\hfill$\Box$
\smallskip}

\def \eps{{\varepsilon}}
\def\poly{\text{poly}}

\def \P {\mathsf{P}}
\def \NP {\mathsf{NP}}

\def \PP {\text{\sf PP}}
\def \sP {\#{\mathsf P}}

\def\ShowAuthNotes{1}
\ifnum\ShowAuthNotes=1
\newcommand{\authnote}[2]{\ \\ \textcolor{red}{\parbox{0.9\linewidth}{[{\footnotesize {\bf #1:} { {#2}}}]}}\newline}
\else
\newcommand{\authnote}[2]{}
\fi

\newcommand{\grp}[1]{\left(#1\right)}

\newcommand{\ssat}[1]{\#\text{SAT}(#1)}

\newcommand{\set}[1]{\left\{#1\right\}}

\newcommand{\thr}[2]{{THR$_{#1}$-{#2}SAT}}

\newcommand{\gtthr}[2]{{GtTHR$_{#1}$-{#2}SAT}}

\DeclareMathOperator{\ord}{ord}

\usepackage{enumitem}
\usepackage[utf8]{inputenc}
\usepackage{framed}
\usepackage{backref}

\def\PP{{\mathsf{PP}}}
\def\BPP {{\mathsf{BPP}}}
\def \PSPACE {{\mathsf{PSPACE}}}

\def \EMAJSAT {{E-MAJ-SAT}}
\newcommand \EMAJkSAT[1] {E-MAJ-{#1}SAT}
\def \MAJSAT {MAJORITY-SAT}
\newcommand \MAJkSAT[1] {MAJORITY-{#1}SAT}
\def \GMAJSAT {GtMAJORITY-SAT}
\newcommand \GMAJkSAT[1] {GtMAJORITY-{#1}SAT}
\def \MAJMAJSAT {MAJ-MAJ-SAT}
\newcommand \MAJMAJkSAT[1] {MAJ-MAJ-{#1}SAT}

\title{MAJORITY-3SAT (and Related Problems) in Polynomial Time
\thanks{Supported by NSF CCF-1909429 and NSF CCF-1741615.}} 
\author{Shyan Akmal \\ MIT \\ \href{mailto:naysh@mit.edu}{naysh@mit.edu} 
	\and Ryan Williams\thanks{This work was done while the author was visiting the Simons Institute for the Theory of Computing, participating in the \emph{Theoretical Foundations of Computer Systems} and \emph{Satisfiability: Theory, Practice, and Beyond} programs.} 
	\\ MIT \\ \href{mailto:rrw@mit.edu}{rrw@mit.edu}}
\date{}

\begin{document}

\maketitle

\begin{abstract}

    Majority-SAT (a.k.a.~MAJ-SAT) is the problem of determining whether an input $n$-variable formula in conjunctive normal form (CNF) has at least $2^{n-1}$ satisfying assignments. Majority-SAT and related problems have been studied extensively in various AI communities interested in the complexity of probabilistic planning and inference. 
    Although Majority-SAT has been known to be {\bf PP}-complete for over 40 years, the complexity of a natural variant has remained open: Majority-$k$SAT, where the input CNF formula is restricted to have clause width at most $k$. 
    
    We prove that for every $k$, Majority-$k$SAT is in {\bf P}; in fact, the problem can be solved in linear time 
    (whereas the previous best-known algorithm ran in exponential time). 
    More generally, for any positive integer $k$ and constant $\rho \in (0,1)$ with bounded denominator, we give an algorithm that can determine whether a given $k$-CNF has at least $\rho \cdot 2^n$ satisfying assignments, in deterministic linear time. We find these results surprising, as many analogous problems which are hard for CNF formulas remain hard when restricted to $3$-CNFs. Our algorithms have interesting positive implications for counting complexity and the complexity of inference, significantly reducing the known complexities of related problems such as E-MAJ-$k$SAT and MAJ-MAJ-$k$SAT. Our results immediately extend to arbitrary Boolean CSPs with constraints of arity $k$. 
    At the heart of our approach is an efficient method for solving threshold counting problems by extracting and analyzing various sunflowers found in the corresponding set system of a $k$-CNF.
    
    Exploring the implications of our results, we find that the tractability of Majority-$k$SAT is somewhat fragile, in intriguing ways. For the closely related GtMajority-SAT problem (where we ask whether a given formula has \emph{greater than} $2^{n-1}$ satisfying assignments) which is also known to be $\PP$-complete, we show that GtMajority-$k$SAT is in {\bf P} for $k\le 3$, but becomes {\bf NP}-complete for $k\geq 4$. We also show that for Majority-SAT on $k$-CNFs with {\bf one} additional clause of arbitrary width, the problem is {\bf PP}-complete for $k \geq 4$, is {\bf NP}-hard for $k=3$, and remains in {\bf P} for $k=2$. 
    These results are counterintuitive, because the ``natural'' classifications of these problems would have been {\bf PP}-completeness, and because there is a stark difference in the complexity of  GtMajority-$k$SAT and Majority-$k$SAT for all $k\ge 4$.
    \end{abstract}

\thispagestyle{empty}
\setcounter{page}{0}
\newpage

\section{Introduction}

The complexity of $\#$SAT, the problem of counting satisfying assignments to propositional formulas (a.k.a. ``model counting'' in the AI and SAT literature), has been intensely studied for decades. 
The pioneering work of Valiant \cite{Valiant_2-SAT} showed that $\#$SAT is $\sP$-complete already for 2-CNF formulas.

Of course, $\#$SAT (and any other $\sP$ problem) is a function problem: up to $n+1$ bits need to be output on a given $n$-variable formula. 
A natural question is: how efficiently can output bits of the $\#$SAT function be computed? Obvious choices are the \emph{low-order bit}, which corresponds to the $\oplus \P$-complete PARITY-SAT problem, and the \emph{higher-order bits}. For CNF formulas, 
the highest-order bit of $\#$SAT corresponds to the case where the $\#$SAT value is $2^n$, which is trivial for CNF formulas.\footnote{The only $n$-variable CNFs with $2^n$ satisfying assignments are those with no clauses. 
Of course, when the formula is DNF, the high-order bit problem is $\mathsf{coNP}$-complete.} 
When the value is less than $2^n$ and $\#$SAT outputs $n$ bits, the high-order bit corresponds to \MAJSAT, the problem of determining whether $\#\text{SAT}(F) \geq 2^{n-1}$. 
It is more common to think of it
as a probability threshold problem: given a formula $F$, is $\Pr_a[F(a)=1] \geq 1/2$?
Sometimes \MAJSAT\ is phrased as determining whether or not $\Pr_a[F(a)=1] > 1/2$; we will call this version \GMAJSAT\ to avoid confusion. Over CNF formulas (and more expressive Boolean representations), there is no essential difference between the two problems.\footnote{See~\Cref{sec:prelim} for a discussion.}

\MAJSAT\ (and \GMAJSAT) are the primary subjects of this paper. Gill \cite{prob-gill} and Simon 
\cite{prob-simon} introduced these problems along with the class $\PP$, which consists of decision problems computing ``high-order bits'' of a $\#P$ function. They proved that \MAJSAT\ on CNF formulas is $\PP$-complete, and that $\sP \subseteq \P^{\PP}$, showing that determining higher-order bits of a general $\sP$ function is as hard as computing the entire function.  

The known proofs of $\PP$-hardness for \MAJSAT\ reduce to CNF formulas having clauses of arbitrarily large width. 
This raises the very natural question of whether \MAJSAT\ remains hard over CNFs with fixed-width clauses. 
Intuition suggests that \MAJkSAT{$k$} should remain $\PP$-hard for $k \geq 3$, by analogy with the $\NP$-hardness of 3SAT, the $\PSPACE$-hardness of Quantified 3SAT \cite{S_PH}, the $\oplus\P$-hardness of PARITY-3SAT\footnote{This follows from the fact that there is a parsimonious reduction from SAT to 3-SAT~\cite{prob-simon}.},
the $\Pi_2 \P$-hardness of $\Pi_2$-$3$SAT \cite{SM_QBF}, and so on. Beyond the SAT problem, it is often true that the hardness of a problem can be preserved for ``bounded width/degree'' versions of the problem: for instance, the $\NP$-hardness of $3$-coloring holds even for graphs of degree at most $4$~\cite{GareyJ79}, and the $\sP$-hardness of counting perfect matchings in graphs holds even for graphs of degree at most $3$ \cite{perm-low-degree}. 
Indeed, the more general problem: \emph{given a $3$-CNF $F$ and an integer $k \geq 0$, determine if $\#SAT(F) \geq 2^k$} can readily be proved $\PP$-complete. However, the same argument cannot be used to show that \MAJSAT\ is $\PP$-complete for 3-CNF $F$.\footnote{The proof of $\PP$-completeness (of the more general problem) follows from two facts: (a) the version of the problem for CNF formulas is $\PP$-complete, by Gill and Simon, and (b) the reduction from CNF-SAT to 3SAT preserves the number of solutions. This proof cannot be used to show that determining $\#SAT(F) \geq 2^{n-1}$ is $\PP$-complete for $3$-CNF $F$, because the Cook-Levin reduction introduces many new variables (the variable $n$ increases when going from CNF-SAT to 3SAT, thereby changing the target number of satisfying assignments).} 

Due to these subtleties, there has been significant confusion in the literature about the complexity of \MAJkSAT{$k$}, with several works asserting intractability for \MAJkSAT{3} and its variants, while others observing that the complexity of the problem remained open at the time ~
\cite{Mundhenk_MAJ-3-SAT, Mundhenk-thesis, BaileyDK01, KG_optimal, BaileyDK07,GHM_MAJ-2-SAT, KG_MAJ-3-SAT, TM_MAJ-2-SAT, PLMZ_open, K_MAJ-3-SAT-cite, FGL_open, KC_MON-2-SAT-cite-talk, K_MON-2-SAT-cite, MCC_E-MAJ3SAT-NP^PP-claim, CDB_MAJ-3-SAT, CM_open,BDPR_E-MAJ3SAT-NP^PP, BDPR_E-MAJ3SAT-NP^PP-full}.\footnote{Even the second author is guilty of being confused: see the first comment at \url{https://cstheory.stackexchange.com/questions/36660/status-of-pp-completeness-of-maj3sat}.}  
This is a critical issue, as \MAJSAT\ and its variants have been at the foundation of many reductions regarding the complexity of probabilistic planning, Bayesian inference, and maximum a~posteriori problems in restricted settings, which are of great interest to various communities within AI. 
The true complexity of \MAJkSAT{$k$} (and related problems) has remained a central open question for these communities. 

\subsection{Our Results}
\label{subsec:overview}

Somewhat surprisingly, we show that \MAJSAT\ over $k$-CNFs is in fact {\bf easy}. In fact, for any constant $\rho \in (0,1)$, we can efficiently determine for a given $k$-CNF $F$ whether $\Pr_{a \in \{0,1\}^n}[F(a) = 1] \geq \rho$ or not.

\begin{theorem}\label{thm:main} For every constant rational $\rho \in (0,1)$ and every constant $k \geq 2$, there is a deterministic linear-time algorithm that given a $k$-CNF $F$ determines whether or not $\#SAT(F) \geq \rho \cdot 2^n$. 
\end{theorem}

To our knowledge, the previous best-known algorithm for this problem ran in $2^{n-\Theta(n/k)}$ time, by running the best-known algorithm for $\#k$SAT~\cite{ImpagliazzoMP12,ChanW21}. Of course, Theorem~\ref{thm:main} does not mean that $\#\P$ functions can be computed in polynomial time; rather, it shows that lower-order bits of $\#k$SAT are the more difficult ones\footnote{Indeed in some sense, ``middle bits'' are $\PP$-hard: Bailey, Dalmau, and Koliatis~\cite{BaileyDK01,BaileyDK07} showed 20 years ago that for all integers $t \geq 2$, determining whether $\#\text{SAT}(F) \geq 2^{n/t}$ (for 3-CNF $F$) is $\PP$-complete.} 
(\Cref{cor:most-sig} in Section~\ref{subsec:high-order} formally describes how to use our algorithm to compute the high-order bits of $\#k$SAT in polynomial time). Even the lowest-order bit of $\#k$SAT is evidently harder: for every $k \geq 2$, PARITY-$k$SAT is known to be $\oplus\P$-complete~\cite{ValiantParity}, so (by Toda's theorem \cite{Toda91}) the low-order bit of $\#k$SAT cannot be computed in $\BPP$ unless $\NP = \mathsf{RP}$.

\paragraph{Implications for Related Inference Problems.} Given that \MAJkSAT{$k$} turns out to be easy, it is worth exploring whether related problems in literature are also easy or hard. In the relevant AI literature on the complexity of Bayesian inference and probabilistic planning, the following two problems are prominent in proving conditional lower bounds: 

\begin{quote}
\EMAJSAT: \emph{Given $n$, $n'$, and a formula $\varphi$ over $n+n'$ variables, is there a setting to the first $n$ variables of $\phi$ such that the majority of assignments to the remaining $n'$ variables are satisfying assignments?} 

\MAJMAJSAT: \emph{Given $n$, $n'$, and a formula $\varphi$ over $n+n'$ variables,
	do a majority of the assignments to the first $n$ variables of $\varphi$ 
	yield a formula where the majority of assignments to the remaining $n'$ variables
	are satisfying?}
\end{quote}

These problems may seem esoteric, but \EMAJSAT\ and related problems are used extensively in the relevant areas of AI, where an environment has inherently ``random'' aspects along with variables one can control, and one wants to ``plan'' the control variables to maximize the chance that a desired property holds~(e.g.~\cite{Littman1998,ParkD04,Darwiche09}). \EMAJSAT\  has also recently been used to study the complexity of verifying differential privacy~\cite{GaboardiNP20}. 

Similarly, \MAJMAJSAT\ applies in the context when one wants to know what is the chance that a random setting of control variables will yield a good chance that a property holds~\cite{choi2012same,OztokCD16}. 

For general CNF formulas, 
\EMAJSAT\ is $\NP^{\PP}$-complete~\cite{Wagner86,Toran91,Littman1998} and \MAJMAJSAT\ is $\PP^{\PP}$-complete~\cite{Wagner86,Toran91,AllenderKRRV01}: roughly speaking, these results imply that both problems are essentially intractable, even assuming oracle access to a $\#$SAT solver (we give a proof of completeness for \EMAJSAT\ for $6$-CNFs with one arbitrary-width clause in Appendix~\ref{appendix:EMAJSAT-hardness}).
There has also been significant confusion about whether \EMAJkSAT{3} (the version restricted to $3$-CNF) is $\NP^{\PP}$-complete or not~\cite{KG_optimal, KG_MAJ-3-SAT, MCC_E-MAJ3SAT-NP^PP-claim, CDB_MAJ-3-SAT, BDPR_E-MAJ3SAT-NP^PP, BDPR_E-MAJ3SAT-NP^PP-full}.
We prove that both \EMAJSAT\ and \MAJMAJSAT\ dramatically decrease in complexity over $k$-CNF formulas. 

\begin{theorem} \label{thm:emajkSAT} \EMAJkSAT{$2$} $\in \P$, and for all $k \geq 3$, \EMAJkSAT{$k$} is $\NP$-complete.
\end{theorem}
	
\begin{theorem} \label{majmaj2SAT-P}\MAJMAJkSAT{2} $\in \P$.
\end{theorem}

Even the $\NP$-completeness of \EMAJkSAT{3} is good news, in some sense: Theorem~\ref{thm:emajkSAT} suggests that such counting problems could in principle be handled by SAT solvers, rather than needing $\#$SAT solvers.

\paragraph{Greater-Than MAJORITY-SAT.} The algorithms behind \Cref{thm:main} can efficiently determine if the $\#$SAT value of a $k$-CNF is \emph{at least} a given fraction of the satisfying assignments. Recall the \GMAJSAT\ problem is to determine if the $\#$SAT value is \emph{greater than} a given fraction, and that over CNFs, there is no essential difference between the two problem variants. Another surprise is that, over $k$-CNFs, there is a difference between these problems for $k \geq 4$: the ``greater than'' version becomes $\NP$-complete!

\begin{theorem} \label{thm:gtmaj-easy} For all $k \leq 3$, \GMAJkSAT{$k$} is in $\P$.
\end{theorem}

\begin{theorem} \label{thm:gtmaj-hard} For all $k \geq 4$, \GMAJkSAT{$k$} is $\NP$-complete.
\end{theorem}

\paragraph{Adding One Long Clause Makes MAJORITY-$\boldsymbol{k}$SAT Hard.} Given the surprisingly low complexity of these threshold counting problems over $k$-CNF formulas, it is natural to investigate what extensions of $k$-CNFs suffice in order for the problems to become difficult. This direction is also important for the considerable collection of results in AI whose complexity hinges on the difficulty of \MAJSAT\ and its variants. We show that adding only \emph{one} extra clause of arbitrary width is already enough to make \MAJkSAT{$k$} difficult, for $k \geq 3$.

\begin{theorem} \label{thm:one-extra-clause} Deciding \MAJSAT\ over $k$-CNFs with one extra clause of arbitrary width is in $\P$ for $k=2$, $\NP$-hard for $k=3$, and $\PP$-complete for $k\geq 4$.
\end{theorem}

This may look preposterous: how could adding only one long clause make \MAJkSAT{$3$} hard? Couldn't we simply try all $O(n)$ choices for picking a literal from the long clause, and reduce the problem to $O(n)$ calls to \MAJkSAT{$3$} with no long clauses? Apparently not! Remember that  \MAJkSAT{$3$} \emph{only} decides whether or not the fraction of satisfying assignments is at least $\rho \in (0,1)$. This information does not help us determine the number of satisfying assignments to $O(n)$ subformulas accurately enough to refute the hardness of \MAJkSAT{$3$} with no long clauses.

\subsection{Intuition}\label{sec:intuition}

The ideas behind our algorithms arose from reconsidering the polynomial-time Turing reduction from $\#$SAT to \MAJSAT~\cite{prob-gill,prob-simon}, in the hopes of proving that \MAJkSAT{$3$} is hard. The key is to reduce the problem
\[\#\text{SATD} := \{(F,s)\mid \#\text{SAT}(F) \geq s\}\] to \MAJSAT. 
From there, one can binary search with $\#$SATD to determine $\ssat{F}$. 
The known reductions from $\#$SATD to \MAJSAT\ require that, given a desired $t \in [0,2^n]$, we can efficiently construct a formula $G_t$ on $n$ variables with exactly $t$ satisfying assignments.\footnote{A standard way to do this is to make a formula $G_t$ which is true if and only if its variable assignment, construed as an integer in $[1,2^n]$, is at most $t$. But constructing such a formula requires arbitrary width CNFs.} Then, introducing a new variable $x_{n+1}$, the formula \[H = (x_{n+1} \vee F) \wedge (\neg x_{n+1} \vee G_t)\] will have $\ssat{H} = \ssat{F} + t$, out of $2^{n+1}$ possible assignments to $H$. Setting $t=2^n-s$, it follows that $\ssat{F} \geq s$ if and only if $\ssat{H} \geq 2^n$, thereby reducing from $\#$SATD to \MAJSAT. 
Observe we can convert $H$ into $k$-CNF, provided that both $F$ and $G_t$ are $(k-1)$-CNF.

However, this reduction fails miserably for $k$-CNF formulas, because for constant $k$ and large $n$, there are many values $t \in [0,2^n]$ for which \emph{no} $k$-CNF formula $G_t$ has exactly $t$ satisfying assignments (observe that every $k$-CNF with at least one clause has at most $(1-1/2^k)\cdot 2^n$ satisfying assignments; therefore no such $k$-CNF formulas $G_t$ exist, for all $t \in [2^n-2^{n-k}-1,2^n-1]$). 
Moreover, every $k$-CNF containing $d$ disjoint clauses ($d$ clauses sharing no variables) has at most $(1-1/2^k)^d \cdot 2^n$ satisfying assignments. But ``most'' $k$-CNF formulas (say, from the typical random $k$-SAT distributions) \emph{will} have \emph{large} disjoint sets of clauses (say, of size $\Omega(n)$). So for ``most'' formulas, we can quickly determine that $\ssat{F} < \rho \cdot 2^n$ for constant $\rho > 0$, by finding a large enough disjoint clause set. 

What remains is a rather structured subset of $k$-CNF formulas. If the maximum possible size of a disjoint clause set is small, then there is a small set of variables that ``hit'' all other clauses (otherwise, the set would not be maximal). That is, there is a small set of variables that have non-empty intersection with every clause. This kind of small hitting set can be very algorithmically useful for solving $\#$SAT. For example, if $k=2$, then every assignment to the variables in a small hitting set simplifies the given formula into a $1$-CNF. In other words, when there is a small hitting set, we can reduce the computation of $\#2$SAT to a small number of calls to $\#1$SAT, each of which can be solved in polynomial time. This is essentially how our algorithm for \MAJkSAT{$2$} works.

The situation quickly becomes more technically complicated, as $k$ increases. When $k=3$, setting all variables in a small hitting set merely simplifies the formula to a $2$-CNF, but $\#$SAT is already $\#\P$-hard for $2$-CNF formulas. To get around this issue, we consider more generally \emph{sunflowers} within the $k$-CNF: collections of sets which all share the same pairwise intersection (called the core). 

Sunflowers in a formula can be useful in bounding the fraction of satisfying assignments. 
To give a simple example, if the entire formula was a sunflower with a single literal $\ell$ in its core, then the fraction of satisfying assignments is at least $1/2$ (because setting $\ell$ true already satisfies the formula). 
Our algorithms seek out large sunflowers on disjoint clauses in $k$-CNF formulas, to get tighter and tighter bounds on the fraction of satisfying assignments.
When a formula does not have many such sunflowers, the formula is structured enough that we can find a small hitting set of variables and use the ideas discussed earlier.

\paragraph{Intuition for Theorem~\ref{thm:main}.}

Here we provide an intuitive idea of how our main algorithm works to determine whether a $k$-CNF has at least a $\rho$-fraction of satisfying assignments. Given a Boolean formula $\Phi$ on $n$ variables, let $\Pr[\Phi]$ denote the probability a uniform random assignment to the variables of $\Phi$ is satisfying.
	For a given $k$-CNF $\varphi$, we want to decide whether the inequality
			\[\Pr[\varphi] \ge \rho\]
	holds or not. We will do this by building up a special $(k-2)$-CNF $\psi$ on the same variable set, where each clause of $\psi$ is contained in a clause of $\varphi$. We split the probability calculation into
	\[\Pr[\varphi] = \Pr[\varphi\land \psi] + \Pr[\varphi \land \lnot\psi]\]
	and use the fact that
		\begin{equation}
		\label{eq:prob-bound}
		\Pr[\varphi\land\psi]\le \Pr[\varphi] = \Pr[\varphi\land\psi] + \Pr[\varphi\land\lnot\psi].
		\end{equation}
	Intuitively, we will construct $\psi$ in such a way that 
	$\Pr[\varphi\land\lnot\psi] < \eps_1$
	for an \emph{extremely} small $\eps_1 > 0$,
	so that it is possible to reduce the problem of determining $\Pr[\varphi] \geq \rho$ to the problem of determining $\Pr[\varphi\land\psi] \geq \rho$. In other words, we can reduce \thr{\rho}{$k$} on $\varphi$  to 
	\thr{\rho}{$k$} on $\varphi\land\psi$.

	This reduction is helpful because the clauses of $\psi$ are subclauses appearing frequently in $\varphi$, so the formula $\varphi\land\psi$ simplifies to a smaller formula than $\varphi$. Additionally, $\varphi \wedge \psi$ has a smaller solution space than $\varphi$, so intuitively it becomes easier to check if the resulting formula has fewer than a $\rho$-fraction of satisfying assignments.	More precisely, it follows from \eqref{eq:prob-bound} that if $\Pr[\varphi\land \psi] \geq \rho$ 
	then $\Pr[\varphi] \geq \rho$ as well.
	The more surprising result is that we can construct $\psi$ so that, 
	if $\Pr[\varphi\land\psi] < \rho$, then we can in fact infer that $\Pr[\varphi\land\psi] < \rho - \eps_2$ for some $\eps_2 > \eps_1$.
	Hence by \eqref{eq:prob-bound} we can deduce that $\Pr[\varphi] < \eps_1 + \rho - \eps_2 < \rho$.
	
	We construct the clauses of $\psi$ by taking cores of large sunflowers in $\varphi$. Defining what counts as ``large'' depends on quite a few parameters, so the analysis becomes rather technical. 

\subsection{Paper Organization}
\label{sec:org}

In \Cref{sec:prelim}, we formally define the problems we are considering, introduce notation, and discuss more related work.

The proof that \MAJkSAT{$k$} is in $\P$ (Theorem~\ref{thm:main}) is spread across multiple sections, to enhance readability. In \Cref{sec:maj2sat} we present a simple algorithm for solving \MAJkSAT{$2$} in linear time, and in \Cref{sec:maj-3sat} we extend this algorithm to solve \MAJkSAT{$3$} in linear time.
Building on these preliminary results, in \Cref{sec:thr3sat} we use more sophisticated arguments to show that for any fixed rational $\rho\in (0,1)$ with denominator bounded above by a constant, we can detect if a 3-CNF on $n$ variables has at least $\rho 2^n$ satisfying assignments in linear time.
In \Cref{sec:thrksat} we extend the results to $k$-CNFs for any fixed integer $k$, finally proving \Cref{thm:main}.
Although \Cref{sec:thrksat} subsumes the main results of the prior sections, we include the proofs of these simpler cases earlier in order to motivate and highlight the key ideas in the final algorithm, and make the overall proof more accessible. 

In \Cref{sec:e-and-maj} we discuss the applications of our algorithmic results to the \EMAJSAT\ and \MAJMAJSAT\ problems.
In \Cref{subsec:big-clause} we discuss hardness results for \MAJkSAT{$k$} with one arbitrary width clause, and in \Cref{subsec:greater} we discuss  \GMAJkSAT{$k$} and how it differs from \MAJkSAT{$k$}. We conclude in \Cref{sec:conclusion} with a discussion of several intriguing open problems.

\section{Preliminaries}
\label{sec:prelim}

We assume basic familiarity with computational complexity, including concepts such as $\PP$ and $\sP$~\cite{AB09-book}. For a formula $F$ on $n$ variables, let $\ssat{F}$ be its number of satisfying assignments as an integer in $[0,2^n]$.

\noindent {\bf CNF Formulas.} A literal is a Boolean variable or its negation, a \emph{clause} is a disjunction of literals, and a \emph{CNF formula} is a conjunction of clauses.
The \emph{width} of a clause is the number of literals it contains.
Given an integer $w$, a \emph{$w$-clause} is just a clause of width $w$.
Given a positive integer $k$, we say a formula is a $k$-CNF if every clause in the formula has width \emph{at most} $k$. We stress that we allow our $k$-CNFs to have clauses of length \emph{up to} $k$: clauses of width $1,\ldots,k$ are allowed.
An empty CNF formula evaluates to $\top$, meaning it is always true.
An empty clause evaluates to $\bot$, meaning it is always false.
Given a CNF formula $\varphi$, we let $|\varphi|$ denote the size of the formula, which is just the sums of the widths of all clauses in $\varphi$.

We remark that all of the results in this paper that hold for $k$-CNF formulas also hold for Boolean constraint satisfaction problems (CSPs), over arbitrary constraints of arity at most $k$. This is because each constraint of such a CSP can be converted into an equivalent $k$-CNF over the same variable set. 

\noindent{\bf \GMAJSAT\ vs MAJORITY-SAT.} Here we briefly describe how to reduce between these two problems. 
To reduce from \GMAJSAT\ to \MAJSAT\ given an $n$-variable formula $F$, introduce $n$ new variables $y_1,\ldots,y_n$ and map $F$ to $F' := (y_1 \vee \cdots \vee y_n) \wedge F$. Then $\ssat{F} \geq 2^{n-1}+1$ 
implies $\ssat{F'} \geq (2^n-1)(2^{n-1}+1) = 2^{2n-1} + 2^n - 2^{n-1}-1 > 2^{2n-1}$ and $\ssat{F} \leq 2^{n-1}$ implies $\ssat{F'} \leq (2^n-1)2^{n-1} = 2^{2n-1} - 2^{n-1}< 2^{2n-1}$.

To reduce from \MAJSAT\ to \GMAJSAT\ given an $n$-variable $F$, introduce one new variable $x_{n+1}$, let $G$ be an $n$-variable formula with precisely $2^n - 2^{n-1}+1$ satisfying assignments, and set $F':=(\neg x_{n+1} \vee F)\wedge(x_{n+1} \vee G)$. Then $\ssat{F'} = \ssat{F} + 2^n - 2^{n-1} + 1$. When $\ssat{F} \geq 2^{n-1}$, we have $\ssat{F'} \geq 2^n+1$, and when 
$\ssat{F} \leq 2^{n-1}-1$ we have 
$\ssat{F'} \leq 2^n$ (we can increase the gap by increasing the number of additional variables). Both reductions need unbounded width CNF formulas.

\noindent{\bf Threshold SAT.} We have already defined the \MAJkSAT{$k$} problem. To discuss problems of detecting fractions of satisfying assignments at other thresholds besides $1/2$, we introduce the following problem.

\begin{definition}[Threshold SAT]
    \label{def:thr}
    For any positive integer $k$ and threshold $\rho\in (0,1)$, the \thr{\rho}{$k$} problem is the following task: given a $k$-CNF formula $\varphi$ on $n$ variables, determine if the inequality $\ssat{\varphi} \ge \rho \cdot 2^n$ holds.
\end{definition}

In our algorithms, we will often make use of the following structures in CNF formulas.

\begin{definition}[Consistent Literal Set]
\label{def:consistent}
    Given a set of literals, we say the set is \emph{consistent} if the set does not simultaneously include $x$ and $\lnot x$ for any variable $x$.
\end{definition}

\begin{definition}[Variable Disjoint Set]
\label{def:disjoint-set}
	Given a set $S$ of clauses, we say $S$ is a (variable) \emph{disjoint set} if for every pair $C,C'$ of distinct clauses of $S$, $C$ and $C'$ share no variables.
\end{definition}

We will also utilize the following simple observations about CNFs formulas.

\begin{proposition}\label{prop:subformula}
Let $F$ be a CNF formula on $n$ variables, construed as a set of clauses. 
Suppose there is a $\rho \in (0,1)$ and a subset $F'$ of the clauses of $F$ such that $F'$ contains $r \leq n$ variables and $\#SAT(F') \leq \rho \cdot 2^r$. 
Then $\#SAT(F) \leq \rho \cdot 2^n$.
\end{proposition}

\begin{proof} 
Note that $F = F' \wedge G$, for some formula $G$. 
Given a fixed $F'$, the number of satisfying assignments to $F$ is maximized when $G$ is a tautology, having $2^n$ satisfying assignments
(since $F$ is over $n$ variables, we can take $G = (x_1 \vee \neg x_1) \wedge \cdots \wedge (x_n \vee \neg x_n)$).
Even in such a case, $\ssat{F} \leq \ssat{F'} \cdot 2^{n-r} \leq \rho \cdot 2^n$.
\end{proof}

\begin{proposition}\label{prop:1cnf} Given a 1-CNF formula $F$ (i.e. $F$ is a conjunction of literals), the number of satisfying assignments to $F$ can be computed in linear time.
\end{proposition}

\begin{proof} 
Let $k$ be the number of 1-clauses (literals) in $F$. If $F$ contains both a variable and its negation, then $F$ is unsatisfiable, and the number of satisfying assignments is $0$. Otherwise, the set of literals in $F$ is consistent, and the number of satisfying assignments is $2^{n-k}$. In either case, we can compute the desired quantity by scanning through the clauses in $F$ once. 
\end{proof}

As two final pieces of notation, we write $A = \poly(B)$ to denote that $A\le B^c$ for some constant $c > 0$, and $A = \exp(B)$ to denote that $A \le 2^{cB}$ for some constant $c > 0$.

\subsection{Comparison With Related Work}

Several works \cite{Luby1996, Hirsch_detsat, Trevisan04, DNF-sparsification, DNF-sparse-sun, DNF-sunflower-regularity} have considered the task of approximately counting satisfying assignments to CNF formulas. In particular, given a constant $\eps \in (0,1)$ and CNF formula $\varphi$, we seek to output an estimate that is within $\eps$ of the true fraction of assignments of $\varphi$ which are satisfying.\footnote{One can also consider \emph{multiplicative} approximations to $\#$SAT, but this task is $\NP$-hard. See for example~\cite{Stockmeyer85,DellL18}.}
In general, the estimates provided by such algorithms may be strictly more or less than the true fraction of satisfying assignments, so such approximation algorithms cannot be used to solve problems like \MAJkSAT{$k$}. 

However, the starting point of our work, the \MAJkSAT{$2$} and \MAJkSAT{$3$} algorithms, uses methods very similar to those of Trevisan~\cite{Trevisan04}, who showed that for any fixed integer $k$ one can \emph{approximately} count the fraction of satisfying assignments in a  $k$-CNF formula efficiently, by working with maximal disjoint sets of clauses.\footnote{In fact, the second author devised an algorithm for \MAJkSAT{$2$} in 2004, inspired by Trevisan's work, but only recently (with the help of the first author) found a way to generalize to \MAJkSAT{$3$} and beyond.}
Given a desired additive approximation error $\eps$, Trevisan's approach shows that every  $k$-CNF can be \emph{approximated} by a special kind of decision tree of $f(\eps,k) \leq O(1)$ size and depth, where the internal nodes are labeled by variables 
and the leaves are labeled with 1-CNFs. 
Computing the exact fraction of satisfying assignments for such a decision tree is simple to do in linear time, and Trevisan uses this count to obtain an $\eps$-additive approximation of the true fraction of satisfying assignments.

In our algorithms, we also implicitly 
(and for \MAJkSAT{$2$}, \EMAJkSAT{$2$}, and \MAJMAJkSAT{$2$},
explicitly) 
construct such decision tree representations, and we also use the fact that one can count satisfying assignments exactly on such decision tree representations. However, for \MAJkSAT{$k$} where $k \geq 3$, 
our algorithms and analysis have to dig further into the problem and take advantage of the structure of the decision tree itself. Informally, we show there are ``gaps'' in the possible $\#$SAT values of such representations. Very roughly speaking, these gaps are part of what allows us to solve the exact threshold counting problem for $k$-CNFs in polynomial time, ``as if'' it were an additive approximation problem. 
Still, many other cases arise in determining the fraction exactly that are irrelevant in approximations. 

More generally, our algorithms rely on extracting sunflowers from various subformulas.
Sunflower lemmas have been used previously for obtaining additive approximations to the fraction of satisfying assignments of disjunctive normal form (DNF) formulas and related problems such as DNF sparsification and compression~\cite{Luby1996, DNF-sparsification, DNF-sparse-sun, DNF-sunflower-regularity}.
These results typically focused on formulas of super-constant width, whereas our work is specialized to CNFs of constant width. Due to our hardness results, one \emph{cannot} extend our algorithms to 3CNFs with even one unbounded width clause, unless $\P = \NP$.

\section{Threshold SAT for 2-CNFs in Linear Time}
\label{sec:maj2sat}

As a warm-up, we begin with a simple linear-time algorithm for \MAJkSAT{$2$} (even \thr{\rho}{$2$}, for every $\rho \geq 1/\poly(n)$) that illustrates a few of the ideas.\footnote{The second author has known of this result since around 2004; see Section 7 of~\cite{Williams04}.}

\begin{theorem} 
	\label{thm:maj2sat} 
For every rational $\alpha \in (0,1)$, there is an $m\cdot \poly(1/\alpha)$-time algorithm that, given any 2-CNF formula $F$ on $n$ variables and $m$ clauses, decides whether $\ssat{F} \geq \alpha \cdot 2^n$ or not. 
Furthermore, when $\ssat{F} \geq \alpha \cdot 2^n$ is true, the algorithm outputs $\ssat{F}$, along with a
a decision tree representation for $F$ of $\poly(1/\alpha)$ size.
The internal nodes are labeled by variables and leaves are labeled by $1$-CNFs.
\end{theorem}

\begin{proof}
For each $\alpha \in (0,1)$, define $c(\alpha) := 1+\lceil\log_{4/3} (1/\alpha)\rceil$. 
Note that $c(\alpha) \leq O(\log \frac{1}{\alpha})$. 

Given a 2-CNF $F$, start by finding a \emph{maximal disjoint set} of clauses $S$. That is, treat the clauses as sets (ignoring literal signs) and find a set $S$ of clauses such that (a) every pair of clauses in $S$ share no variables and (b) all other clauses in $F$ contain at least one variable occurring in $S$. This can be done by greedily choosing the set $S$ (picking disjoint clauses until we cannot) in time $O(m\cdot |S|)$. We argue that we can stop once $|S|$ exceeds $c(\alpha)$.

{\bf Case 1:} Suppose $|S| > c(\alpha)$. Then we claim that $\ssat{F} < \alpha\cdot 2^n$. Note that each of the clauses in $S$ are over disjoint variables, so each clause in $S$ reduces the total number of satisfying assignments by $3/4$. By our choice of $c(\alpha)$, we have $(3/4)^{c(\alpha)} < \alpha$. Therefore, less than an $\alpha$-fraction of the possible assignments satisfy the subformula $S$, and by Proposition~\ref{prop:subformula}, we can return {\bf NO}.

{\bf Case 2:} Otherwise, $S$ is a maximal disjoint set of clauses with $|S| \leq c(\alpha)$. Since every clause in $F$ contains at least one variable occurring in $S$, it follows that, when we plug in any assignment to the variables of $S$, the remaining formula is a $1$-CNF. Therefore, if we try all of the at most 
\[3^{c(\alpha)} \leq O(3^{(\log(\frac{1}{\alpha})/\log(4/3))}) \leq O((1/\alpha)^{\log(3)/\log(4/3)}) \leq O((1/\alpha)^{3.82})\]
satisfying assignments to the clauses of $S$, and solve $\#$SAT on the remaining $1$-CNF formula in $O(m)$ time (Proposition~\ref{prop:1cnf}), we can determine the number of satisfying assignments \emph{exactly} in this case.

Note that {\bf Case 1} of this algorithm only occurs when $\ssat{F} < \alpha \cdot 2^n$.
Consequently, whenever $\ssat{F} \ge \alpha \cdot 2^n$ we fall into {\bf Case 2} and our algorithm reports the exact count of satisfying assignments. 
The overall run time of this algorithm is $m\cdot \poly(1/\alpha)$.
\end{proof}

It is interesting to contrast the above result with the result of Leslie Valiant that $\oplus$2SAT is $\oplus$P complete \cite{ValiantParity}. 
Valiant's result implies that computing the low-order bit of $\#2$SAT in polynomial time would imply that $\NP \subseteq \mathsf{BPP}$. Our result shows that computing the low-order bit of $\#2$SAT looks much more difficult than higher-order bits.

\section{Threshold SAT for 3-CNFs in Linear Time}
\label{sec:maj-3sat}

Recall from \Cref{def:thr} that given a positive integer $k$ and parameter $\rho \in (0,1)$, we define the ``threshold SAT'' problem \thr{\rho}{$k$} to be the task of deciding whether at least a $\rho$-fraction of assignments to a given $k$-CNF are satisfying.
For example, the \MAJkSAT{$3$} problem discussed previously is equivalent to \thr{1/2}{3}.
In this section, we show that for any $\rho$ with constant numerator and denominator, \thr{\rho}{3} can be solved in polynomial time.

More specifically, in \Cref{subsec:gt-1/2} and \Cref{subsec:eq-1/2} we show how to extend the ideas from \Cref{sec:maj2sat} with a subformula detection argument to prove the following result.

\begin{theorem} \label{thm:maj3sat} For every constant $\rho \in [1/2,1]$, we can decide in polynomial time if a given $3$-CNF on $n$ variables has at least $\rho \cdot 2^n$ satisfying assignments.
\end{theorem}

In \Cref{sec:thr3sat}, we generalize these ideas further to solve \thr{\rho}{$3$} for any fixed constant threshold $\rho$.

\begin{theorem}
	\label{thm:thr-3-sat} For every positive integer $M$ and for every rational $\rho\in(0,1)$ whose denominator is bounded above by $M$,
	\thr{\rho}{3} can be solved in $O_M(n)$ time.
\end{theorem}

These results will later be subsumed by \Cref{thm:main} (showing an analogous result for all $k \ge 3$), which is proven in  \Cref{sec:thrksat}. 
This section is included only to present arguments which are less technically challenging, and therefore hopefully more accessible, while still resolving the complexity of \MAJkSAT{$3$}. 
We encourage any reader who becomes bored while reading this section to jump directly to \Cref{sec:thrksat}.
Similarly, the argument in \Cref{sec:thr3sat} is self-contained and can be read without referring to \Cref{subsec:gt-1/2} and \Cref{subsec:eq-1/2}, even though the motivation for the arguments comes from these earlier subsections.

\subsection{Thresholds Greater than One-Half}
\label{subsec:gt-1/2}

Building on the \MAJkSAT{$2$} algorithm of \Cref{thm:maj2sat}, we propose the following natural generalization to 3-CNFs. In the following, a ``disjoint set of clauses'' refers to a variable disjoint set (see \Cref{def:disjoint-set}).

\begin{quote}
{\bf Algorithm A.} \emph{(With two unspecified constants $c_1$ and $c_2$.)}\\
Given a 3-CNF $F$, find a maximal disjoint set $S$ of clauses of $F$. If $|S|$ exceeds a certain constant $c_1$, then output {\bf NO}.\\
For all $7^{|S|}$ SAT assignments $A$ to the clauses in $S$, let $F_{A}$ be the $2$-CNF induced by assignment $A$, and search for a maximal disjoint set $S_{A}$ of 2-clauses in $F_{A}$. If $|S_{A}|$ exceeds a certain constant $c_2$, then output {\bf NO}. Otherwise, try all SAT assignments $A'$ to the clauses in $S_A$. For each $1$-CNF formula induced by an $A'$, count solutions to the $1$-CNF in polynomial time.\\
Return {\bf YES} if and only if the total number of solutions counted (over all assignments $A$) is at least $\rho\cdot 2^{n}$.
\end{quote}

First, we prove that {\bf Algorithm A} correctly decides $\ssat{F} \geq \rho \cdot 2^n$ for all fractions $\rho > 1/2$.

\begin{theorem}  \label{thm:maj3sat-at-least-half} For every $\eps \in (0,1/2]$, we can decide in $\poly(1/\eps,n)$ time if a given $3$-CNF on $n$ variables has at least $(1/2+\eps)\cdot 2^n$ satisfying assignments. Moreover, given any $3$-CNF with at least $(1/2+\eps)\cdot 2^n$ satisfying assignments, we can report the exact number of satisfying assignments.
\end{theorem}

The ability to report the exact number of satisfying assignments will (provably) no longer hold when we consider the case of $\eps = 0$, in the next subsection. This is the major reason why we have treated the two cases separately.

To prove Theorem~\ref{thm:maj3sat-at-least-half}, we show that by setting $c_1,c_2$ appropriately in {\bf Algorithm A}, we can decide if there are at least $\rho \cdot 2^n$ SAT assignments for $\rho > 1/2$. We first prove a lemma regarding the sizes of maximal disjoint sets in formulas obtained by assigning variables.

\begin{lemma} \label{lemma:alpha-bound} Let $\rho > 1/2$, and let $S$ be a maximal disjoint set of $k$-clauses in a $k$-CNF $F$. Suppose $F$ has at least $\rho \cdot 2^n$ satisfying assignments. For all possible assignments $A$ to the variables of $S$, and for every induced $2$-CNF $F_{A}$ obtained by assigning $A$ to $S$, $F_{A}$ must contain a maximal disjoint set of $(k-1)$-clauses of size less than $2^k |S| 
\ln(1/(\rho-1/2))$.
\end{lemma}

\begin{proof} The proof is by contrapositive. Let $\eps > 0$ be such that $\rho := 1/2+\eps$, and let $S$ be a maximal disjoint set of $k$-clauses in a given $k$-CNF $F$. Suppose there is an assignment $A$ to the variables of $S$ such that $F_A$ has a maximal disjoint set of $(k-1)$-clauses of size at least $K := 2^k |S| \ln(1/\eps)$. By the pigeonhole principle, there exists some literal $\ell \in \{x, \neg x\}$, coming from a variable $x$ in the maximal disjoint set $S$, and a set $T_{\ell}$ of  at least $K/(2|S|) = 2^{k-1} \ln(1/\eps)$ clauses in $F$ of the form \[(\ell \vee a_{i,1} \vee \cdots \vee a_{i,k-1})\] where the variables of $a_{i,j}$ are all \emph{distinct} over all $i=1,\ldots,|T_{\ell}|$ 
and $j=1,\ldots,k-1$. That is, the subformula $T_{\ell}$ of $F$ has in total $1+(k-1)r$ distinct variables, where $r := |T_{\ell}|$.
    
Since $r \geq 2^{k-1} \ln(1/\eps)$, the fraction of satisfying assignments in $T_{\ell}$ is at most \[\frac{2^{(k-1)r} + (1-1/2^{k-1})^r\cdot 2^{(k-1)r}}{2^{1+(k-1)r}} = \frac{1}{2} + \frac{1}{2}\cdot(1-1/2^{k-1})^r \leq \frac{1}{2} + \frac{\eps}{2} < \rho,\] where the $2^{(k-1)r}$ term comes from the case where $\ell$ is true, and $(1-1/2^{k-1})^r\cdot 2^{(k-1)r}$ term comes from the case where $\ell$ is false. Therefore, in such a case, $F$ must have less than a $\rho$ fraction of satisfying assignments by Proposition~\ref{prop:subformula}.
\end{proof}

We can apply Lemma~\ref{lemma:alpha-bound} by arguing that, if any subformula $F_A$ of $F$ has a ``large'' maximal disjoint set of 2-clauses, then we can output NO when $\rho > 1/2$. Otherwise, every $F_A$ has a ``small'' maximal disjoint set of 2-clauses, and {\bf Algorithm A} works in that case.

\begin{proofof}{Theorem~\ref{thm:maj3sat-at-least-half}} Let $\eps > 0$. We consider {\bf Algorithm A} with constants $c_1 := 10$ and $c_2 := 72 \ln(1/\eps)$.

If $|S| > 10$, then the fraction of satisfying assignments to the subformula $S$ is less than $1/2$, therefore $F$ has less than a $1/2$ fraction by Proposition~\ref{prop:subformula}. Therefore in step 2 of Algorithm A, we can report {\bf NO}. 

Otherwise, $|S| \leq 9$. Suppose we try all possible satisfying assignments to $S$ (there are at most $3^{|S|}$) and suppose there is some induced formula $F_A$ with a maximal disjoint set $S_A$ of at least $72 \cdot \ln(1/\eps)$ clauses. By Lemma~\ref{lemma:alpha-bound} we can deduce that $F$ has less than an $\rho := 1/2+\eps$ fraction of satisfying assignments, and can report {\bf NO}. 

In the remaining case, every induced formula $F_A$ has a maximal disjoint set $S_A$ of less than $72 \cdot \ln(1/\eps)$ clauses. By trying all possible SAT assignments to each $S_A$ (there are $3^{|S_A|} \leq \poly(1/\eps)$ such assignments) we can count the number of satisfying assignments for each of the remaining $1$-CNF formulas in linear time, and determine the exact number of satisfying assignments by taking the sum of all such counts.
\end{proofof}

\subsection{Threshold of One-Half}
\label{subsec:eq-1/2}

We now we turn to the case of solving \thr{\rho}{$3$} for threshold value $\rho = 1/2$.

When $\rho = 1/2$, {\bf Algorithm A} does not work correctly in all cases (regardless of how its parameters are set). 
Consider a 3-CNF formula $F$ in which every clause contains a common variable $x$ occurring positively. 
This is trivially a YES-instance for \MAJkSAT{$3$}. (Note we cannot efficiently compute the number of satisfying assignments exactly in this case, as it would solve the $\#2$SAT problem in polynomial time!) Running {\bf Algorithm A} on $F$, it will find an $S$ with $|S| = 1$, since $x$ appears in all clauses. 
When we try all satisfying assignments to $S$, and $x$ is set true, the formula becomes a tautology. 
But when $x$ is set false, the formula becomes an arbitrary $2$-CNF, with potentially a very large disjoint clause set. Regardless of the size of that clause set, the original $F$ is still a YES instance, even if \emph{all} of the clauses in the remaining $2$-CNF are disjoint. So, an algorithm for \MAJkSAT{3} needs to be able to account for this sort of behavior, where a single literal appears in many clauses.

To handle this case, we introduce a check for another type of ``bad'' subformula.

\begin{lemma} \label{lemma:subformula1} Let $\ell \in \{x, \neg x\}$ be a literal, and let
\[S=\{(\ell \vee a_1 \vee b_1),\ldots,(\ell \vee a_t \vee b_t),(u \vee v \vee w)\}\] be a set of clauses with the following properties:
\begin{itemize}
    \item For all $i,j \in [t]$, $a_i$ and $b_j$ are literals from $2t$ distinct variables, all of which are different from $x$.
    \item The literal $\ell$ does not appear in $(u \vee v \vee w)$ (however, $\neg \ell$ may appear in $(u \vee v \vee w)$).
\end{itemize}
Then for all $t \geq 8$, $S$ has less than $2^{r-1}$ satisfying assignments, where $r$ is the total number of variables occurring in $S$.
\end{lemma}

\begin{proof} Let $r$ be the total number of variables in $S$; note that $r \geq 2t+1$. When $\ell$ is set to false, the $t$ clauses $(a_i \vee b_i)$ are all disjoint, so the formula $S$ has at most $(3/4)^t \cdot 2^{r-1}$ satisfying assignments over the remaining $r-1$ variables. When $\ell$ is true, the clause $(u \vee v \vee w)$ remains, so (over the remaining $r-1$ variables) the number of satisfying assignments in this case is at most $(7/8)\cdot 2^{r-1}$ (note that if the literal $\neg \ell$ appears in $(u \vee v \vee w)$, then the fraction is $3/4$, which is only better for us).
For $t \geq 8$, the total number of satisfying assignments is therefore $((3/4)^t + 7/8)\cdot 2^{r-1} < 2^{r-1}$.
\end{proof}

For $t$ sufficiently large, Lemma~\ref{lemma:subformula1} can be used to show that $S$ has less than $(7/16+\eps)2^r$ satisfying assignments for any desired $\eps > 0$.

\paragraph{MAJ3SAT in $\P$.} We are now ready to give a polynomial-time algorithm for deciding if a $3$-CNF has at least $2^{n-1}$ satisfying assignments. For ease of reading, here we will describe the algorithm alongside its analysis.

Given a $3$-CNF $F$ on $n$ variables, we start by checking if there is a common literal $\ell$ appearing in every clause of $F$. In this case we output {\bf YES}, as any such formula is satisfied by at least half of its assignments.

After this point, we know:
\begin{quote}
    $(\star)$ For every literal $\ell$ there is at least one clause in $F$ that does not contain $\ell$.  
\end{quote}

Next, we find a maximal disjoint set $S$ among the 3-clauses in $F$. If $|S| \geq 6$ then, since $(7/8)^6 < 0.449 < 1/2$, we can output {\bf NO} by Proposition~\ref{prop:subformula}.

Otherwise, we know that $|S| \leq 5$. For each of the $7^{|S|}$ satisfying assignments $A$ to the clauses of $S$, we do the following:

 For each $2$-CNF formula $F_A$ induced by an assignment $A$ on the variables of $S$ in $F$, find a maximal disjoint set $S_A$ over the 2-clauses in $F_A$. 

\begin{enumerate}
    
    \item We claim that, if there is an assignment $A$ such that $|S_A| \geq 48|S|+2$, then $F$ must contain less than $2^{n-1}$ satisfying assignments. Hence we can output {\bf NO} in this case.

    This paragraph proves the claim. 
    For each 2-clause $(x \vee y)$ in $S_A$, select one clause from $F$ that $(x \vee y)$ arose from: such a clause is either of the form $(\ell \vee x \vee y)$ where $\ell$ is a literal whose variable appears in $S$, or it is simply $(x \vee y)$ ($F$ may contain 2-clauses itself). Put each such clause from $F$ into a new set $S'_A$, so that $|S'_A|=|S_A|$.
    Suppose there are at least three 2-clauses in $S'_A$. Since these 2-clauses are disjoint and appear in $F$, the subformula of $S'_A$ restricted to these 2-clauses is a subformula of $F$ and must have at most a $(3/4)^3 < 0.422 < 1/2$ fraction of satisfying assignments. By Proposition~\ref{prop:subformula}, $F$ has less than $0.422\cdot 2^n$ satisfying assignments in this case.
    Otherwise, there are at most two $2$-clauses in $S'_A$. Removing them from $S'_A$, there are still at least $48|S|$ 3-clauses. As there are $3|S|$ distinct variables appearing in $S$, and hence $6|S|$ literals whose variable appears in $S$, there must be a literal $\ell$ whose variable appears in $S$ such that $\ell$ appears in at least $8$ clauses of $S'_A$. By property $(\star)$ above, it follows that there is a subformula in $F$ satisfying Lemma~\ref{lemma:subformula1}. Therefore $F$ has less than $\rho 2^{n}$ satisfying assignments for a constant $\rho < 1/2$.
        
    \item
        Otherwise, for all assignments $A$, we have $|S_A| < 48|S|$. In this case, we can try all $3^{|S_A|}$ satisfying assignments $A'$ to the 2-clauses in $S_A$. Since $S_A$ is a maximal disjoint set of 2-clauses in $F_A$, every formula obtained by plugging in $A'$ is a $1$-CNF formula. We solve $\#$SAT on the resulting $1$-CNF formula in linear time, and add the number to a running sum (calculated over all choices $A$ and $A'$).
    
\end{enumerate}
Finally, output {\bf YES} if the total sum of satisfying assignments exceeds $2^{n-1}$, otherwise output {\bf NO}. This completes the description of the algorithm, and its analysis.

It is interesting to observe that, no matter what $3$-CNF formula is provided, at least one of the following conditions is true at the end of the algorithm:
\begin{itemize}
    \item[(a)] There are at least $2^{n-1}$ satisfying assignments (an early {\bf YES} case).
    \item[(b)] There are at most $\rho 2^n$ satisfying assignments, for a constant $\rho < 1/2$ (an early {\bf NO} case).
    \item[(c)] The number of satisfying assignments is counted exactly.
\end{itemize}
Therefore, for any $3$-CNF formula in which the $\#$SAT value is strictly between $\rho 2^n$ and $2^{n-1}$, the above algorithm actually computes the $\#$SAT value exactly.

Note that the above algorithm runs in linear time, although the constant factor in the worst case (enumeration over partial assignments) is at least $7^5 \cdot 3^{48\cdot 5 - 1} > 10^{118}$. Of course, in order to give a succinct proof, we have been extremely loose with the analysis; a smaller constant factor is certainly possible. 

\subsection{Generalizing to All Thresholds}
\label{sec:thr3sat}

In this subsection we finish the proof of \Cref{thm:thr-3-sat}, handling the case of fractions $\rho < 1/2$.

\paragraph{Some Lemmas.} The below lemma will be used to show that there must be certain ``gaps'' in the numbers of satisfying assignments to a $k$-CNF formula. This will aid in the detection of {\bf YES} and {\bf NO} instances of Threshold SAT.

\begin{lemma}
	\label{lm:gen-bit-argument}
 	Let $n$ and $m$ be arbitrary positive integers, and let $\rho \in (0,1)$ be rational, of the form \[\rho = \frac{a}{2^vb}\]
	for unique odd integer $b$, nonnegative integer $v$, and integer $a$ with $\gcd(a,2^vb) = 1$. Then for every integer $N$ which is the sum of at most $m$ powers of two, if $N < \rho 2^n$ then $N\le \grp{\rho - \eta}2^n$, for a positive $\eta$ depending only on $a, b, v$, and $m$.
\end{lemma}

The proof of the above result follows from casework on the binary expansion of $\rho$, and is presented in \Cref{sec:bit}.
	\Cref{lm:gen-bit-argument} formalizes the intuition that if we want to maximize the sum of $m$ powers of two with the constraint that the sum is strictly less than $\rho=a/b$,
	the best one can do is to greedily pick distinct powers of $2$ (i.e. binary digits)
	whose total sum is never at least $\rho$.
	The proof in \Cref{sec:bit} shows that when $b$ is a power of $2$, 
	we can take
		\[\eta \ge \frac{\rho}{2^{m + c}}\]
	for some constant $c$.
	For general $\rho$ the decay in terms of $m$ could be much worse,
	but can still bounded below as
		\[\eta \ge \frac{\rho}{2^{(b-1)m+c}}.\]

We will also need the following structural lemma, that lets us extract simple certificates that a $3$-CNF has a small fraction of satisfying assignments.
To state this lemma, recall we defined a \emph{variable-disjoint set} $D$ (\Cref{def:disjoint-set}) to be a collection of clauses such that no two distinct clauses in $D$ have literals corresponding to the same variable.
In other words, if we view the clauses of $D$ as sets of variables, the sets are pairwise disjoint.
The following is a natural generalization of this definition.

\begin{definition}[Sunflowers]
	\label{def:sunflower}
	Let $w$ be a nonnegative integer.
	A set $S$ of clauses is called a {\bf \emph{$w$-sunflower}} if there exists a 
	set $L$ of $w$ distinct literals, such that every clause in $S$ contains the literals of $L$,
	and removing the literals of $L$ from each clause of $S$ produces a variable-disjoint set.
	A set of clauses is a \emph{sunflower} if it is a $w$-sunflower for some nonnegative $w$.
	The \emph{core} of a sunflower is the set $L$,
	the sunflower's \emph{weight} is the size of its core $w$, and the sunflower has \emph{size} $|S|$.
\end{definition}

Note that a $0$-sunflower is just a disjoint set.
In our algorithms, we will always work with sunflowers whose cores are consistent sets (Definition~\ref{def:consistent}), which implies there is an assignment to just the core variables which satisfies all clauses in the sunflower.

Our key lemma describes an algorithm which takes a $3$-CNF $F$ as input, and either finds a large $0$-sunflower in $F$, a large $1$-sunflower in $F$, or it decomposes $F$ into a short list of $1$-CNFs. 
The proof below is just a more sophisticated restatement of arguments from previous subsections, but it will be a useful framework for generalizing to the case of $k \geq 4$.

\begin{lemma}[Sunflower Extraction in a 3-CNF]
	\label{lm:sun-3CNF}
	Let $Z$ and $Q$ be positive integers. There is an algorithm which, given a $3$-CNF $F$, 
	runs in
		\[O\grp{7^Z3^{3ZQ}\cdot |F|}\]
	time, where $|F|$ is the size of the input formula, and either	
	\begin{itemize}
		\item outputs a 0-sunflower of size $Z$ in $F$,
		or
		\item outputs a 1-sunflower of size $Q$ in $F$,
		or
		\item
		outputs a list of $m \le 7^Z3^{3ZQ}$ 1-CNFs such that 
		the number of satisfying assignments in $F$ equals the sum of $\ssat{F'}$
		over all 1-CNFs $F'$ in the list.
	\end{itemize}
\end{lemma}
\begin{proof}
	
	We will prove the result by repeatedly extracting maximal disjoint sets from the formula $F$ in a greedy manner.
	If the sets we find are ``large,'' we will output a sunflower;
	otherwise, we can loop over all satisfying assignments to some ``small'' disjoint sets and recover a list of 1-CNFs.
	
	We begin by scanning through the clauses of $F$ to build a maximal variable disjoint set $S$. 
	If $|S| \geq Z$, {\bf return} the $0$-sunflower $S$ and halt.
	
	Otherwise, $|S| < Z$. We enumerate all $7^{|S|} < 7^Z$ satisfying assignments $\alpha$
	to $S$. For each $\alpha$, we  produce a new formula $F_\alpha$ obtained from $F$ by assigning values to the variables of $S$ according to $\alpha$. Since $S$ is a \emph{maximal} disjoint set and we assigned all variables in $S$, it follows that $F_{\alpha}$ is a $2$-CNF. 
	Next, we process the $F_\alpha$ in the order they are produced.
	For each $F_{\alpha}$, scan through its clauses and build a maximal variable disjoint set $S_{\alpha}$ of width-two clauses.
	We distinguish between two cases, according to whether this set is large or small.
	
	{\bf Case 1: } $|S_{\alpha}| \ge 3Q|S|$.
	
	In this case, we show how to return a large 1-sunflower.
	
	Since the variables in $S_{\alpha}$ are a ``hitting set'' for the clauses of $F$, and since $S_{\alpha}$ has width-two clauses, each clause of $S_{\alpha}$ was formed in $F_{\alpha}$ by removing exactly one literal corresponding to a variable in $S$, from some clause in $F$.
	
	Since $S$ is a variable disjoint set in a $3$-CNF, the number of variables appearing in $S$ is at most $3|S|$.
	Under an assignment $\alpha$ to $S$, we remove a literal $\ell$ from a clause in $F$ to produce $F_\alpha$ if and only if $\alpha$ assigns $\ell$ to be false.
	Thus, at most $3|S|$ literals were removed to produce the clauses in $S_{\alpha}$.
	
	By averaging/pigeonhole, there is some literal $\ell$ and there are at least
		\[\frac{|S_{\alpha}|}{3|S|} \ge Q\]
	clauses in $S_{\alpha}$ whose corresponding clauses in $F$ all contain $\ell$.
	We can find such a literal $\ell$ and set of clauses simply by scanning through the clauses in $F$ and $S_{\alpha}$.
	Because $S_{\alpha}$ was a variable-disjoint set, the resulting set of clauses forms a $1$-sunflower of size at least $Q$ in $F$, which we can then {\bf return}.

	{\bf Case 2: } $|S_\alpha| < 3Q|S|$.
	
	In this scenario, we just loop over all $3^{|S_\alpha|} < 3^{3ZQ}$
	satisfying assignments $\beta$ to $S_\alpha$, noting that each 2-clause has precisely $3$ assignments to the variables it contains which can satisfy it, and no two clauses share variables.
	For each produce a new formula $F_{\alpha,\beta}$ induced from $F_\alpha$ by assigning values to its variables according to $\beta$.
	Every clause of width two in $F_{\alpha}$ has a literal whose value gets assigned in this process (since we chose $S_\alpha$ was chosen to be maximal) so the resulting formula $F_{\alpha,\beta}$ is a 1-CNF.
	
	Thus the set of satisfying assignments in $F$ can be found just by taking the union over all partial assignments $\alpha$ and $\beta$, and counting satisfying assignments for the resulting 1-CNFs $F_{\alpha,\beta}$.
	So in this case we just return the list of all 1-CNFs produced in this way.
	This list has at most
    $7^{Z}3^{3ZQ}$
	formulas, as desired.
\end{proof}

With these two results, we are now ready to prove the main theorem of this section.

\begin{proof}[Proof of \Cref{thm:thr-3-sat}]
	Let $\varphi$ be the input $3$-CNF on $n$ variables.
	
	Set $z := \Theta(\log(1/\rho))$ to be the smallest positive integer such that 
		\[(7/8)^z < \rho.\]
	The variable $z$ will be our cutoff value for a 0-sunflower being ``large.''
	
	Set $t := \lfloor \log(1/\rho) \rfloor$ to be the unique nonnegative integer such that 
		\[(1/2)^t\ge \rho > (1/2)^{t+1}.\]
	Let $q_0 >  \dots > q_{t} > 0$ be integer parameters to be specified later.
	The $q_r$ constants will determine the cutoffs for 1-sunflowers to be considered ``large.''
	For now, we take $q_t$  large enough so that 
		\begin{equation}
		\label{eq:3qt-bound}
		(3/4)^{q_t} < \frac{\rho - (7/8)^z}{t+1}.
		\end{equation}
	Since the denominator of $\rho$ is bounded above by some constant $M$, the above inequality holds provided we take $q_t \ge \Omega(z\log M)$.
	
	With this setup, we present the algorithm for \thr{\rho}{3} below.
	The steps of the algorithm which could return an answer are annotated with case numbers, which correspond to the cases we examine later when arguing correctness. 
	The routine also sets the value of a parameter $r$, which
	corresponds to the number of large 1-sunflowers we end up discovering.
	The symbol $\rhd$ indicates the beginning of a comment which provides context for a given step, and is not part of the algorithm description.
	
	\begin{enumerate}
		\item 
		Initialize $F = \varphi$.
		
		\item 
		Initialize $r = 0$.
		
		\item 
		While $r\le t$:
		
			\begin{enumerate}[label=(\roman*)]
				\item 
				Run the algorithm of \Cref{lm:sun-3CNF} with $Z = z$ and $Q = q_{r}$.
				
				\item
				If step 3(i) returns a 0-sunflower of size at least $Z$, then  return {\bf NO}.
				
				{\bf (Case 1)}

				\item 
				If step 3(i) returns a 1-sunflower of size at least $Q$,
				save this sunflower $S$ and its core $C = \{\ell\}$.
				Then scan through $F$ to find the set of all clauses $F'$ in $F$ containing the literal $\ell$.
				
				Let $G$ be the set of clauses formed by taking $F\setminus F'$ and removing all instances of the literal $\lnot \ell$.
				Update the value of the formula $F\leftarrow G$.
				
				\begin{quote}
				$\rhd $\textit{
				The above step just asserts that $\ell$ is true in $F$.}
				\end{quote}
	
				Increment $r\leftarrow r+1$.

				If the new $F$ is empty (has no clauses) and $r\le t$, then return {\bf YES}.
								
				{\bf (Case 2)}
								
				\begin{quote}
				$\rhd $\textit{
				We check that $r\le t$ since in the case where $r=t+1$ after being incremented, halting at this step 
				means we have a hitting set of
				$t+1$ literals on the clauses of $F$,
				which is not necessarily a YES case.}
				\end{quote}
				
				\item 
				Otherwise, step 3(i) returns a list $L$ of $m\le \exp(z\cdot q_{r})$ 1-CNFs.
				
				Use this list to compute the fraction of satisfying assignments of $F$.
				If this fraction is at least $\rho$ then return {\bf YES},
				otherwise return {\bf NO}.
				
						{\bf (Case 3)}

			\end{enumerate}
		
		\item
		If the while loop completes without early halting, return {\bf NO}.
		
				{\bf (Case 4)}
	\end{enumerate} 

    Recall that running the sunflower extraction procedure from \Cref{lm:sun-3CNF} on a formula $F$ with parameters $Z$ and $Q$ takes $O\left(f(Z,Q)|F|\right)$ time where $f(Z,Q)  := 7^Z3^{3ZQ}$.
    Consequently, the above algorithm takes asymptotically at most
		\begin{equation}
		\label{eq:time-bound}
		\grp{f(z, q_0) + f(z,q_1) + \dots + f(z,q_t)}|\varphi|
		\end{equation}
	time.

	Suppose that before halting, the algorithm completes $r$ iterations of step 3(iii) in the while loop.
	At this point, we have discovered literals 
	$\set{\ell_i}_{1\le i \le r}$ 
	and a sequence of $3$-CNFs $\set{F_i}_{0\le i\le r}$  such that 
	\begin{enumerate}
		\item
		$F_0 = \varphi$,
		\item 
		$F_{i-1}$ contains a 1-sunflower of size $q_{i-1}$ with core $\{\ell_i\}$, and
		\item 
		$F_i$ is induced from $F_{i-1}$ by setting $\ell_i$ to be true
	\end{enumerate}
	for $1\le i\le r$.
	
	The following observation helps us argue that only a negligible fraction of satisfying assignments of $\varphi$ do not set all of the $\ell_i$ to be true.
	
	\begin{claim}
		\label{obs:3sun-small}
		The fraction of assignments which set some $\ell_i$ to be false and satisfy $\varphi$ is less than
			\[r\cdot (3/4)^{q_{r-1}}.\]
			
	    Equivalently, we have 
	       \[\Pr[\varphi\land \grp{\lnot\ell_1 \lor \lnot \ell_2 \lor \dots \lor \lnot \ell_r}] < r\cdot (3/4)^{q_{r-1}} < \rho - (7/8)^z .\]
	\end{claim}
\begin{proof}
	Take any index $1\le i\le r$.
Then if we set $\ell_i$ to be false, and $\ell_j$ to be true for all $j < i$, the resulting formula contains a 0-sunflower of size at least $q_{r-1} \ge q_t$ with clauses of width at most two.
This follows from points 2 and 3 above, which define the $F_j$ formulas.
The fraction of assignments to variables in the resulting formula that can satisfy it is at most 
\[(3/4)^{q_{r-1}} \le (3/4)^{q_t} < \frac{\rho - (7/8)^z}{t+1}\]
by our choice of $q_t$.
Now, by considering the minimum index $i$ with $\ell_i$ set to false,
we see that 
    \[\Pr[\varphi\land \grp{\lnot\ell_1 \lor \lnot \ell_2 \lor \dots \lor \lnot \ell_r}]
 = \Pr[\varphi\land \lnot\ell_1] + \Pr[\varphi\land \ell_1\land \lnot \ell_2] + \dots + \Pr[\varphi\land \ell_1\land  \ell_2\dots \land \lnot \ell_r].\]
 By the above discussion, each term on the right hand side is less than $(3/4)^{q_{r-1}}$.
 Thus we have 
    \[\Pr[\varphi\land \grp{\lnot\ell_1 \lor \lnot \ell_2 \lor \dots \lor \lnot \ell_r}] < r\cdot (3/4)^{q_{r-1}} < \rho - (7/8)^z\]
as claimed.
\end{proof}
	
	We now perform casework on the four distinct ways the above procedure could halt, and in each scenario show that the algorithm returns the correct answer for the \thr{\rho}{3} problem.

	\begin{description}
		
				\item[Case 1: Large 0-Sunflower] \hfill\\
		Suppose we halt in step 3(i).
		Then we found a 0-sunflower $S$ of size at least $z$ in $F_r$.
		
		Consequently, if the literals $\set{\ell_i}_{1\le i\le r}$ are all set true, to satisfy $\varphi$ we still need to satisfy $S$, which is satisfied by at most a 
		\[(7/8)^z < \rho\]
		fraction of assignments to its own variables (which are disjoint from the $\ell_i$). In other words, we have
		\[\Pr[\varphi\land\grp{\ell_1 \land \ell_2 \land \dots \land \ell_r}] < (7/8)^z.\]
		Combining this observation with \Cref{obs:3sun-small}, which shows that the fraction of assignments which satisfy $\varphi$ and set some $\ell_i$ to be false is less than $\rho - (7/8)^z$, we deduce that $\Pr[\varphi] < \rho$, so reporting NO is correct.

			\item[Case 2: Small Covering by Cores] \hfill\\
			Suppose we halt in step 3(iii) because the formula $F$ became empty when $r\le t$.
			
			Then this means that there exists a consistent set of at most $t$ literals $\ell_1, \dots, \ell_{r}$ (corresponding to the cores of 1-sunflowers found in the procedure) 
			corresponding to distinct variables such that every clause in the original formula contains at least one of these literals.
			
			Consequently, any assignment that sets these literals to true is a satisfying assignment.
			It follows that at least a
			\[1/2^t \ge \rho\]
			fraction of assignments to $\varphi$ are satisfying, so returning YES is correct.

		\item[Case 3: Exact Count on Subformula] \hfill\\
		Suppose we halt in step 3(iv) because we obtained a list of 1-CNFs for $F_r$.

		This list contains at most $m \le f(z,q_r)$ formulas.
		Using \Cref{prop:1cnf}, in $O(m|\varphi|)$ time we can 
		count the number of satisfying assignments for each formula in this list.
		Let the sum of these counts be $N$.
		By \Cref{prop:1cnf} we know that $N$ is the sum of at most $m$ powers of two,
		and by \Cref{lm:sun-3CNF} we have
			\[\ssat{F_r} = N\]
		where $F_r$ is viewed as a formula on $n-r$ variables (the original variables of $\varphi$, with the variables corresponding to $\ell_1, \dots, \ell_r$ removed).
		
		If $N \ge \rho 2^n$ then
		we can immediately report YES, because the satisfying assignments where we set all the $\ell_i$ to be true already account for at least a $\rho$-fraction of all assignments.
		
		Otherwise, $N < \rho 2^n$.
		Then by \Cref{lm:gen-bit-argument}, there exists an $\eta > 0$ depending on $\rho$ and $m$ such that
			\[N\le (\rho - \eta)2^n.\]
			
	    In other words, the fraction of assignments which set all of the $\ell_i$ to true and satisfy $\varphi$ is at most 
	        \[\frac{N}{2^n} \le \rho-\eta.\]
	        
	    By \Cref{obs:3sun-small}, the fraction of assignments which set some $\ell_i$ to false and satisfy $\varphi$ is less than
	        \[r \cdot \grp{3/4}^{q_{r-1}}.\]
	        
	    It follows that the fraction of satisfying assignments of $\varphi$ is at most
	        \[\Pr[\varphi] < r \cdot \grp{3/4}^{q_{r-1}} + (\rho - \eta) = \rho + \left[r\cdot \grp{ 3/4}^{q_{r-1}} - \eta\right].\]
	        
	    Thus, if we pick each $q_{r-1}$ large enough that
	        \[r\cdot \grp{\frac 34}^{q_{r-1}} < \eta\]

		we necessarily have $\ssat{\varphi} < \rho 2^n$ and can report NO.

		\item[Case 4: Many Large 1-Sunflowers] \hfill\\
		Suppose	the routine halts because we set $r=t+1$ (terminating the loop).
		
		Then the fraction of assignments which satisfy $\varphi$ and set all the $\ell_i$ to true is at most
		    \[(1/2)^{t+1}\]
		since the literals $\ell_i$ come from different variables.
		
		By \Cref{obs:3sun-small}, the fraction of assignments which satisfy $\varphi$ and set some $\ell_i$ to false is at most
		
		    \[(t+1)\cdot (3/4)^{q_t}.\]
		    
		Consequently, the fraction of satisfying assignments in $\varphi$ is at most 
		
		\[\Pr[\varphi] < \frac{1}{2^{t+1}} + (t+1)\cdot \grp{\frac 34}^{q_t}.\]

		So if we take $q_t$ large enough that
		
		\[(t+1)\cdot \grp{\frac 34}^{q_t} < \rho - (1/2)^{t+1}\]
		
		then the total fraction of satisfying assignments of $\varphi$ is strictly less than $\rho$
		so we can report NO in this case as well.
		Note that our choice of $t$ ensures that the right hand side is positive.		
		Moreover, if $\rho$ has denominator bounded above by a constant $b$, the right hand side is bounded below by $1/(b2^{t+1})$ so there is some constant $q_t$ satisfying the above inequality.
	\end{description}

	This handles all possible ways the routine could terminate.
	The arguments above show that provided we set $q_t$ large enough in terms of $\rho$ and $M$, and then recursively set $q_{r-1}$ large enough in terms of $q_r$ for each possible value $t\le r\le 1$, the algorithm is correct (the precise choice of parameter values is described in more detail in \Cref{sec:eff-parameter}).
	If the upper bound $M$ on the denominator of $\rho$ is a constant,
	then the factor multiplied to $|\varphi|$ in the time bound of \cref{eq:time-bound}
	is also a constant depending on $M$, so the algorithm runs in linear time as desired.
\end{proof}

Getting a precise bound on the runtime of \Cref{thm:thr-3-sat} is difficult, as it depends on $\rho$ in subtle ways. However, we can give a very loose upper bound on the worst case runtime of the algorithm's dependence on $\rho$.
Recall that $A = \poly(B)$ means that $A \le B^{O(1)}$.

\begin{proposition}
	\label{prop:eff}
	Let $\rho\in (0,1)$ be a rational with denominator $b$. 
	Set $t = \lfloor \log(1/\rho)\rfloor$.
	Then there exists $K = \text{poly}(1/\rho)$ such that if we define
		\[c =
	\underbrace{K^{\cdot^{\cdot^{K^{(\log b)}}}}}_{t+2\text{ terms}}\]
	to be a tower of $t+1$ exponentiations of $K$ together with one exponentiation to the $(\log b)^{\text{th}}$ power at the top of the tower, the
	\thr{\rho}{3} algorithm described in the proof of \Cref{thm:main} takes at most 
	\[c |\varphi|\]
	time.
\end{proposition}

We prove the above result in  \Cref{sec:eff-parameter}.

\begin{remark}[Similarity to Regularity Lemmas]
Although the runtime described in \Cref{prop:eff} increases horrendously quickly as $\rho$ gets smaller, this tower-of-exponents dependence on $1/\rho$ is perhaps not too surprising, given the similarities between our approach and other ``regularity lemma'' approaches.
Roughly speaking, a regularity lemma is a structural theorem for some class of combinatorial objects, which states that given some fixed robustness parameter $\eps$, any object from the class has a ``small-sized'' (depending on $\eps$) structured representation, provided we are allowed to make some ``small'' (again depending on $\eps$) number of modifications to the object.

Examples include Szemer\'{e}di's classic graph regularity lemma \cite{sze} and a recent set regularity lemma used for graph coloring algorithms \cite{color}.
In both these examples, the size of the structured representation grows like a tower of exponentials whose height is polynomial in $1/\eps$, and there are proofs showing that any approach using these regularity lemmas requires such a rapid growth \cite{gowers-old, gowers-new, color}.

Our algorithm can be viewed from this same perspective:
given a 3-CNF and fixed threshold $\rho$ we find a constant number of literals whose values we can set, after which point the modified formula has a solution space which can be decomposed into a disjoint union of solutions to 1-CNFs (the structured representation).
It turns out that a similar statement holds more generally for $k$-CNFs.
\end{remark}

\section{Threshold SAT on General Bounded Width CNFs}
\label{sec:thrksat}

In this section, we show how to extend the results of the previous section to detect whether the fraction of satisfying
assignments in a given $k$-CNF is at least $\rho=a/b$ in polynomial time, when the clause width $k$ and denominator of the threshold $b$ are fixed positive integers. 

\begin{reminder}{Theorem~\ref{thm:main}}
 For every constant rational $\rho \in (0,1)$ and constant integer $k$, there is a deterministic linear-time algorithm that given a $k$-CNF $F$ determines whether or not $\ssat{F} \geq \rho \cdot 2^n$. 
\end{reminder}

The proof extends the ideas employed in the proof of \Cref{thm:thr-3-sat} even further. We will need the following generalization of \Cref{lm:sun-3CNF}, to extract sunflowers in CNFs of width greater than three.

\begin{lemma}[Sunflower Extraction Algorithm]
	\label{lm:sun-k-cnf}
	Fix positive integers $Q_0, Q_1, \dots, Q_{k-2}$. There is a computable function $f$ and an algorithm which runs in at most \[f(Q_0, Q_1, \dots, Q_{k-2})\cdot |F|\]
	time on any given $k$-CNF $F$, which either

	\begin{itemize}
		\item 
		produces a $v$-sunflower of size at least $Q_w$ in $F$ for some  $w\in\set{0, 1, \dots, k-2}$ and $v<w$, 
		or
		
		\item 
		produces a collection ${\cal C}$ of $1$-CNF formulas such that $|{\cal C}|\le f(Q_0, Q_1, \dots, Q_{k-2})$ and
		\[\ssat{F} = \sum_{F' \in {\cal C}} \ssat{F'}.\] That is, $\ssat{F}$ equals the sum of $\ssat{F'}$ over all $1$-CNFs $F'$ in ${\cal C}$.
	\end{itemize}
	
\end{lemma}

We prove Lemma~\ref{lm:sun-k-cnf} in \Cref{sec:sun}, using an inductive argument very similar to the proof of the classic sunflower lemma \cite{original-sun} and the proof of \Cref{lm:sun-3CNF}.

\begin{theorem}
	\label{thm:thr-k-sat}
	For every positive integer $k$ and for every rational $\rho\in (0,1)$ whose denominator is bounded above by some constant $M$,
	\thr{\rho}{$k$} can be solved deterministically on input $\varphi$ in  $O_{M,k}(|\varphi|)$ time.
\end{theorem}

\begin{proof} The intuition behind the proof is provided in \Cref{sec:intuition}. Here we briefly recall the idea before proceeding formally. 
For a $k$-CNF $\varphi$, we want to decide whether \[\Pr[\varphi] \ge \rho.\] We will do this by constructing a special $(k-2)$-CNF $\psi$ over the same variable set, splitting the probability calculation into
\[\Pr[\varphi] = \Pr[\varphi\land \psi] + \Pr[\varphi \land \lnot\psi].\]
    Observe that
\begin{equation}
\label{eq:prob-bound2}
		\Pr[\varphi\land\psi]\le \Pr[\varphi]= \Pr[\varphi\land\psi] + \Pr[\varphi\land\lnot\psi].
\end{equation}
	Intuitively, we construct $\psi$ in such a way that 
	$\Pr[\varphi\land\lnot\psi] < \eps_1$
	for an \emph{extremely} small $\eps_1 > 0$,
	so that it is possible to reduce the problem  \thr{\rho}{$k$} on $\varphi$ to 
	\thr{\rho}{$k$} on $\varphi\land\psi$.
	We pick $\psi$ so that $\varphi\land\psi$ can be significantly simplified, making the new formula  easier to work with.

	On a first read, it may be easier to ignore the precise bounds we require for each constant, and instead just check which variables correspond to the values of sunflower sizes and sunflower counts.
	
	\paragraph{\bf Beginning of parameters.}
	
	Let $z$ be the smallest positive integer such that $\grp{1 - \frac{1}{2^k}}^z < \rho$.
		
	As in the proof of \Cref{thm:thr-3-sat}, the parameter $z$ is our cutoff value for a large 0-sunflower.
	
	For convenience we write
	    \begin{equation}
		\label{eq:alpha-def}
		\alpha := \rho - \grp{1 - \frac{1}{2^k}}^z
		\end{equation}
	which will be a useful upper bound on fractions of satisfying assignments.
	Note that because $\rho$ has denominator bounded above by some constant $M$, 
	$\alpha$ is bounded below by some function of $M$ and $k$.
	
	Let $t_1 := \lceil \log 1/\rho\rceil$ be the smallest positive integer such that 
	$(1/2)^{t_1-1} > \rho \ge (1/2)^{t_1}.$
	This is analogous to the parameter $t$ introduced in the proof of \Cref{thm:thr-3-sat}.
		
	The constant $t_1$ represents the maximum number of 1-sunflowers we will search for.
	
	We then take parameters 
		\[q_1(0) > q_1(1) > \dots > q_1(t_1-1) \]
	where $q_1(r_1)$ denotes the size of the $1$-sunflower we look for, assuming we have already found $r_1$ sunflowers of weight $1$
	(recall that we defined the \emph{weight} of a sunflower to be the size of its core).
	We will describe the values of these parameters later in the proof.
	
	We introduce similar parameters for sunflowers of larger weight, but the setup is more involved.
	For each possible index $0\le r_1 \le t_1-1$ we introduce a parameter $t_2(r_1)$ to be determined,
	and then additionally take a sequence of constants 
		\[q_2(r_1, 0) > q_2(r_1, 1) > \dots > q_2(r_1, t_2(r_1)-1).\]
	More generally, for each $2\le w\le k-2$,
	assuming we have defined some $t_v$ and $q_v$ sequences for $v<w$,
		for each choice of arguments $0\le r_v \le t_v(r_1, \dots, r_{v-1}) - 1$,
	we recursively define a parameter
		\[t_w(r_1, r_2, \dots, r_{w-1})\] 
		together with a sequence of constants 
		\[q_w(r_1, r_2, \dots, 0) > q_w(r_1, r_2, \dots, 1) > \dots > 
		q_w(r_1, r_2, \dots, t-1)\]
	where $t = t_w(r_1, r_2, \dots, r_{w-1})$.
	
	Intuitively, $t_w(r_1, \dots, r_{w-1})$ denotes the maximum number of \emph{additional $w$-sunflowers we will search for}, and the constant $q_w(r_1, \dots, r_{w-1}, r_w)$ represents the size of the $w$-sunflower we are looking for, assuming that thus far we have found precisely $r_v$ large $v$-sunflowers for each $v<w$.
	
	For every $1\le w\le k-2$, we also define the constant
	    \[T_w = \prod_{v=1}^{w-1} \max_{r_1, \dots, r_{v-1}} t_v(r_1, \dots, r_{v-1}).\]
	 The maximums taken in each factor on the right hand side above are taken over sequences $(r_1, \dots, r_{v-1})$ which which satisfy $0\le r_{u} \le t_{u}(r_1, \dots, r_{u-1})-1$ for all $u\le v-1$.
	 The constant $T_w$ is just an upper bound on how many different values of $(r_1, \dots, r_{w-1})$ we will ever consider (because there are at most $t_1$ choices of $r_1$, then at most $\max_{r_1} t_2(r_1)$ choices of $r_2$, etc.).
	
	For convenience, we typically abbreviate  $\vec{r} = (r_1, r_2, \dots, r_{k-2})$ 
	and $\vec{r}[v] = (r_1, \dots, r_v)$.
	This allows us to write $t_w(\vec{r}[w-1])$ and $q_w(\vec{r}[w])$ instead of $t_w(r_1, \dots, r_{w-1})$ and $q_w(r_2, \dots, r_{w})$,
	where it is understood that the entries of $\vec{r}$ satisfy the appropriate inequalities of $0\le r_v\le t_v(\vec{r}[v-1])-1$ for each $v\le w$.

	We will describe the values of all these constants near the end of the proof.
	For now, we note some bounds they must satisfy.
Define
		    \begin{equation}
		\label{eq:beta-def}
	\beta = \rho - \frac 1{2^{t_1}} \grp{1 - \frac{1}{2^k}}.
	\end{equation}
	The above quantity is positive by our choice of $t_1$.
	Like $\alpha$, the quantity $\beta$ is another useful upper bound on the fraction of satisfying assignments a formula has, and is bounded below by some function of $M$ and $k$.
	
	We then set the $q_w$ parameters
	large enough that
		\begin{equation}
		\label{eq:w-sun-main-bound}
		 \grp{1 - \frac{1}{2^{k-w}}}^{q_{w}(\vec{r}[w])} < \frac{\min(\alpha, \beta)}{(k-2)T_w\cdot t_w(\vec{r}[w-1])}.
		\end{equation}
	for all $\vec{r}$ and $w\ge 2$.
	We can ensure this holds by making sure $q_w(\vec{r}[w])$ is large enough in terms of $T_w$ and $t_w(\vec{r}[w-1])$ for all $\vec{r}$.
	In other words, we just need to pick each $q_w(\vec{r}[w])$ large enough in terms of $t_w(\vec{r}[w-1])$ and $t_v(\vec{s}[v-1])$ for all $\vec{s}$ and $v<w$.
	\\
	\noindent{\bf End of parameters.}
	
	\paragraph{How the Algorithm Works.} At a high level, the algorithm proceeds by scanning through $\varphi$ for large sunflowers.
	If there is no large sunflower at all, we can exactly count the number of satisfying assignments, using \Cref{lm:sun-k-cnf}.
	If we do find a large sunflower, we take its core, add it as a clause to $\psi$, then repeat the procedure on $\varphi\land\psi$.
	We will prove after our description of the algorithm that the answer to the problem is preserved by replacing $\varphi$ with $\varphi\land\psi$.
	
	These clauses are the analogues of the literals $\ell_i$ from the proof of \Cref{thm:thr-3-sat} (adding the clause $(\ell_i)$ to $\varphi$ is the same as the same as setting $\ell_i$ true).
	
	Our procedure returns NO if we find ``too many'' large sunflowers. Intuitively, this is equivalent to the fraction of satisfying assignments in $\psi$ becoming ``too small''
	(and if $\Pr[\psi]$ is small, $\Pr[\varphi\land\psi]$ will also be small).
	Otherwise, the procedure returns YES if  $\varphi\land\psi$ has a small consistent hitting set, since in this case $\Pr[\varphi\land\psi] \ge \rho$ just by considering assignments which satisfy the literals in the hitting set.

	\paragraph{The Pseudocode.} Below we present the pseudocode for the algorithm.
	We annotate the description of the algorithm with some comments, indicated by the $\rhd$ symbol, for additional context and recording changes to  $\varphi$ (this formula is not explicitly used in the algorithm, but is useful to keep around for bookkeeping).
	After that, we argue correctness and then discuss the runtime.
	
	\begin{enumerate}
		\item 
			Initialize $F\leftarrow \varphi$.
			\begin{quote}
			\emph{$\rhd$
			Initialize $\psi \leftarrow \top$.
			We maintain throughout that $F=\varphi\land\psi$.}
			\end{quote}
	
			Initialize $r_w\leftarrow 0$ for all $w \in \{1,2,\ldots,k-2\}$.
			\begin{quote}
			\emph{$\rhd$
			Here $r_1$ counts the number of $1$-sunflowers found so far, and for $w>1$,
			$r_w$ counts the number of ``relevant'' $w$-sunflowers which have been added since we last added a $v$-sunflower of weight $v<w$.}
			\end{quote}
			
		\item 
		\label{step:while-loop}
			While $r_w < t_w(\vec{r}[w-1])$ holds for all $w \in \{1,2,\ldots,k-2\}$:
			
			\begin{quote}
			\emph{$\rhd$
			Recall that $\vec{r}[w-1]$ is shorthand for $(r_1, \dots, r_{w-1})$. 
				It will turn out that in each pass through the body of this loop, the value of $\vec{r}$ becomes lexicographically larger. 
				Consequently, since all the $t_w(\vec{r}[w-1])$ bounds are constants, the loop is guaranteed to halt within $O(1)$ iterations.}
			\end{quote}
			
				\begin{enumerate}[label=(\roman*)]
					\item
					\label{step:gen-sun}
						Run the sunflower extraction algorithm from \Cref{lm:sun-k-cnf} on $F$ with parameter values
						$Q_0 \leftarrow z$ and
						$Q_w \leftarrow q_w(\vec{r}[w])$
						for $1\le w\le k-2$.
						
						\begin{quote}
			    \emph{$\rhd$ 
						This step either outputs a sunflower or a list of 1-CNFs.}
						\end{quote}

				\item 
				\label{step:counting}
				If step \ref{step:while-loop}\ref{step:gen-sun} outputs a list of 1-CNFs,
				use this list to compute $\Pr[F]$ exactly.
				
				If $\Pr[F]\ge \rho $ return {\bf YES}.
				Otherwise return {\bf NO}.
					
					\item 
					\label{step:0-sun}
						If instead step \ref{step:while-loop}\ref{step:gen-sun} outputs a 0-sunflower of size at least $Q_0$, return {\bf NO}.
		
				\item 
				\label{step:sun-add}
					Otherwise, step \ref{step:while-loop}\ref{step:gen-sun} returns a sunflower $S$ of size at least $Q_w$ and core $C$ of size $|C|=w$.
					
					For every proper subset $D$ of the literals in $C$ in increasing order of size,
					let $|D| = v$ and check if $F$ has as $v$-sunflower $T$ of size at least $Q_v$.
					If any smaller weight sunflower $T$ exists, take such a sunflower-and-core pair $(T,D)$ which minimizes the size $|D|$, and update $S\leftarrow T$ and $C\leftarrow D$.
					
					\begin{quote}
				\emph{$\rhd$
					The above check ensures that whenever we add a sunflower of weight $w$, none of its literals belong to a large $v$-sunflower for some $v<w$. This is useful because it will let us bound the number of times any literal appears in $\psi$, which (looking ahead) will help us argue that if we add too many sunflowers $\Pr[\psi]$ becomes arbitrarily small.
						This is the key observation that lets us argue that we can run the while loop for $O(1)$ iterations and still correctly return NO.}\\ 
				\emph{~We describe in the analysis later how to do the above check in linear time, using existing FPT set-packing algorithms.}
						\\ 
					~	\emph{		 Update $\psi\leftarrow \psi\land C$.}
						\end{quote}

					Increment $r_w\leftarrow r_w+1$ and reset $r_{x} \leftarrow 0$ for all $x\in [w+1, k-2]$.
					
					Let $F'$ be the formula obtained by removing every clause of $F$ which contains $C$ as as subclause.
					Update $F\leftarrow F'\land C$.
					
				\begin{quote}
				\emph{$\rhd$
				This last step corresponds to asserting that $C$ is true in $F$.}
				\end{quote}

				\end{enumerate}
				
		\item 
		\label{step:outside}
			If $\rho = (1/2)^{t_1}$ and the formula $F$ is just a conjunction of a consistent set of $t_1$ literals, return {\bf YES}.
		
			Otherwise, return {\bf NO}.
							
			\begin{quote}
			\emph{$\rhd$
			Note that $\rho \ge (1/2)^{t_1}$ by definition.}
			\end{quote}
			
	\end{enumerate}
		
		This completes the description of our algorithm.
		We now prove correctness, by bounding by the fraction of satisfying assignments $\Pr[\varphi]$ via conditioning on whether $\psi$ is satisfied or not.
		We first consider the case where the algorithm returns an answer during some iteration of the loop in step \ref{step:while-loop}.
		
		In each iteration of the loop that completes without halting, we find some clause $C$.
		We then update $F$ by asserting that $C$ is true, and update $\psi$ by adding it as a clause to $C$.
		Because this is the only way these formulas are modified, we see that in the algorithm above, at the beginning of any loop iteration
		the set of satisfying assignments of $F$ is the same as the set of satisfying assignments of $\varphi\land\psi$.
		By construction, $\psi$ is a $(k-2)$-CNF with
		$r_1$ distinct clauses of width $1$, and at least 
		$r_w$ distinct clauses of width $w$ for each $w\ge 2$
		(there may be more clauses of width greater than $1$, because of the reset procedure at the end of step \ref{step:while-loop}\ref{step:sun-add}).

		Because each clause added to $\psi$ is the core of a large sunflower in $F$, it turns out that among assignments which do not satisfy $\psi$, the fraction which do satisfy $\varphi$ is extremely small.
        We show the following.
        
        \begin{claim}
			\label{eq:conditional-small}
		We have 
\[
			\Pr[\varphi\land\lnot \psi] < \min(\alpha,\beta).
\]
		\end{claim}

		Recall that $\alpha$ and $\beta$ are probability thresholds defined according to \cref{eq:alpha-def} and \cref{eq:beta-def} in the parameters section. 
		This claim is the analogue of \Cref{obs:3sun-small} from the proof of \Cref{thm:thr-3-sat}.
		We actually show the claim follows from a more general result that will be useful later.
		
		  \begin{claim}
			\label{claim:conditional-general}
			
		For each clause of $C$ of $\psi$, let $S(C)$ be the sunflower it came from. 
		Let $w(C) = |C|$ be the weight of this sunflower and $q(C) = |S(C)|$ be the size of this sunflower.
		We also define $\vec{r}(C) = (r_1(C), \dots, r_{k-2}(C))$ to be the particular value taken by variable $\vec{r}$ at the time clause $C$ was added to the algorithm.

		If $\gamma$ is a real number such that
\[
			\left(1 - \frac{1}{2^{k-w(C)}}\right)^{q(C)} < \frac{\gamma}{(k-2)T_w\cdot t_w(r_1(C), \dots, r_{w-1}(C))}
\]
    for every clause $C$ in $\psi$, then we have 
        \[\Pr[\varphi\land\lnot\psi] < \gamma\]
    as well.
		\end{claim}
		
		\begin{proof}
		
	    To prove the above inequality, write $\psi = C_1\land C_2\land \cdots \land C_s$ as a conjunction of clauses $C_i$.
	    Now define the formulas
			\[D_i = \grp{\bigwedge_{j<i} C_j}\land \lnot C_i\]
		for each index $1\le i\le s$.

    Note that no two of the $D_i$ can be satisfied simultaneously.
    Then by the law of total probability we have 
	 \begin{align}
	 \label{eq:majksat-cond-prob}
	    \Pr[\varphi\land \lnot\psi] = \Pr[\varphi\land \lnot\psi\land D_1] + \dots +  \Pr[\varphi\land \lnot\psi\land D_s] +  \Pr[\varphi\land \lnot \psi\land \lnot\grp{D_1 \lor \dots \lor D_s}].
	  \end{align}
	  
	Observe that $\lnot \psi$ is satisfied precisely when one of the clauses $C_i$ is false.
    Consequently $\lnot \psi$ is satisfied if and only if one of the $D_i$ is satisfied.

	 Therefore the final term in \eqref{eq:majksat-cond-prob} vanishes and we can simplify the equation to 
	 \begin{align}
	     \label{eq:majksat-cond-prob2}
	 \Pr[\varphi\land\lnot \psi] = \Pr[\varphi\land D_1] + \Pr[\varphi\land D_2] + \dots + \Pr[\varphi\land D_s].
	 \end{align}
	 
	 If an assignment $A$ satisfies $D_i$, then for $A$ to satisfy $\varphi$ as well, $A$ must satisfy the variable disjoint set of clauses formed by taking the sunflower $S(C_i)$ and removing the core $C_i$ from each clause
	 (this is because for $D_i$ to be satisfied, the cores $C_j$ for $j<i$ must be satisfied, while $C_i$ is not).
	 Because $S(C_i)$ came from a $k$-CNF, the width of its clauses is at most $k-w(C_i)$, and the probability of satisfying this variable disjoint set is bounded above by 
	    \[\left(1 - \frac{1}{2^{k-w(C_i)}}\right)^{q(C_i)} < 		 \frac{\gamma}{(k-2)T_w\cdot t_w(r_1(C_i), \dots, r_{w-1}(C_i))},\] where the inequality holds due to the hypothesis of the claim.
	    So by \Cref{prop:subformula} we have 
	        \begin{equation}
	        \label{eq:extrabound}
	        \Pr[\varphi\land D_i] < \frac{\gamma}{(k-2)T_w\cdot t_w(r_1(C_i), \dots, r_{w-1}(C_i))}
	        \end{equation}
	    for each index $i$.
	    
	    Now that we have upper bounds for each of the terms on the right hand side of \cref{eq:majksat-cond-prob}, we will combine these bounds to get an upper bound on the overall probability $\Pr[\varphi\land\psi].$
	    
	    For each integer $1\le w\le k-2$ and possible value of $\vec{r}[w-1] = (r_1, \dots, r_{w-1})$ which could show up in our algorithm, let $\mathcal{C}_w(\vec{r}[w-1])$ be the set of $w$-clauses $C$ in $\psi$ such that $r_1(C) = r_1$, $r_2(C) = r_2$, \dots, and $r_{w-1}(C) = r_{w-1}$.
	    The loop predicate from step \ref{step:while-loop} guarantees that $|\mathcal{C}_w(\vec{r}[w-1])| \le  t_w(\vec{r}[w-1])$ for each choice of $w$ and $\vec{r}$ (since if we ever get $t_w(\vec{r}[w-1])$ clauses of width $w$ for the given values of $r_1, \dots, r_{w-1}$, the algorithm will halt).
	    Consequently, by grouping terms on the right hand side  of \cref{eq:majksat-cond-prob2} by the width $w = w(C_i)$ of the clause they correspond to, then grouping by their associated values ${r_1}(C_i), \dots, r_{w-1}(C_i)$, and then finally applying \cref{eq:extrabound}, we can bound
	    
	    \begin{align*}\Pr[\varphi\land\lnot \psi] &< \sum_{w=1}^{k-2} \sum_{\vec{r}[w-1]}\sum_{C\in \mathcal{C}_{w}(\vec{r}[w-1])} \frac{\gamma}{(k-2)T_w\cdot t_w(\vec{r}[w-1])} \\
&\le \sum_{w=1}^{k-2} \sum_{\vec{r}[w-1]} t_w(\vec{r}[w-1])\cdot  \frac{\gamma}{(k-2)T_w\cdot t_w(\vec{r}[w-1])} \\
&= \sum_{w=1}^{k-2} \sum_{\vec{r}[w-1]}   \frac{\gamma}{(k-2)T_w} \\
&\le (k-2)T_w \cdot \frac{\gamma}{(k-2)T_w} = \gamma.
	    \end{align*}
	 Here, the second summation is taken over all $\vec{r}[w-1]$ values satisfying $0\le r_v\le t_v(\vec{r}[v-1]) - 1$ for $v\le w-1$.
	 The inequality from the first to second line holds because $|\mathcal{C}_w(\vec{r}[w-1])| \le t_w(\vec{r}[w-1])$ from the discussion in the previous paragraph.
	 The inequality from the third to final line holds because there are $k-2$ choices for $w$ and at most $T_w$ possibilities for $\vec{r}[w-1]$ by definition.
	 
	 Thus $\Pr[\varphi\land\lnot\psi] < \gamma$, which proves the desired result.
		\end{proof}

\begin{proof}[Proof of \Cref{eq:conditional-small}]

    By \cref{eq:w-sun-main-bound}, the hypothesis of \Cref{claim:conditional-general} holds for $\gamma = \min(\alpha,\beta).$
    The result follows. 
\end{proof}
		
		With these  observations, we can prove the correctness of the algorithm.
		We first consider the case where the algorithm halts during some iteration of the loop.

\begin{description}
	
	\item[Halting at Step \ref{step:while-loop}\ref{step:0-sun}]\hfill\\
		In this case, since $F$ contains a 0-sunflower of size at least $z$,
		by \Cref{prop:subformula} we have
			\[\Pr[\varphi\land\psi] \le \grp{1 - \frac 1{2^k}}^z < \rho - \alpha.\]
			
		Then by \cref{eq:prob-bound2} and \Cref{eq:conditional-small}, we deduce that 
		
			\[\Pr[\varphi] \le \Pr[\varphi\land\psi] + \Pr[\varphi\land\lnot\psi] < (\rho-\alpha) + \alpha = \rho\]
			
		so we can return {\bf NO}.
			
	\item[Halting at Step 2\ref{step:counting}]\hfill\\
		
		If $\ssat{F} \ge \rho 2^n$, then by \cref{eq:prob-bound2} we have
		
			\[\Pr[\varphi] \ge \Pr[\varphi\land\psi] = \Pr[F] \ge \rho\]
			
		and we can return {\bf YES}.
		
		Otherwise, by \Cref{lm:sun-k-cnf}, the list of $1$-CNFs produced consists of 
		
			\[m\le f(z, q_1(\vec{r}[1]), q_2(\vec{r}[2]), \dots, q_{k-2}(\vec{r}[k-2]))\]
			
		formulas for some computable function $f$.
		Then by \Cref{prop:1cnf} and the case assumption, the number of satisfying assignments $N = \ssat{F}$ is a sum of at most $m$ powers of two with
		
			\[N < \rho 2^n.\]
			
		Therefore by \Cref{lm:gen-bit-argument} we have 
		
			\begin{equation}
			\label{eq:gapped}
			\Pr[F] = \Pr[\varphi \land \psi] = \frac{N}{2^n} \le \rho - \eta
			\end{equation}
			
		for some positive $\eta$ whose value depends only on $m$ and $\rho$.

		Now, each clause of width $w$ in $\psi$ came from a $w$-sunflower of size $q_w(\vec{s}[w])$ for some 
		\[\vec{s}= (s_1, \dots, s_{k-2}),\] 
		with the property that its prefix
			\[\vec{s}[w] = (s_1, \dots, s_w)\]
		is lexicographically smaller than the prefix $\vec{r}[w] = (r_1, \dots, r_w)$.
		This lexicographical ordering property holds because of how the entries of $\vec{r}$ are updated in step~\ref{step:while-loop}\ref{step:sun-add} (namely, when we increment $r_w$, we reset all $r_x$ to $0$, for all $x > w$).
		So if we set parameters so that for all such $\vec{s}$ we have 
		
			\begin{equation}
				\label{eq:w-sun-growth}
				\grp{1 - \frac{1}{2^{k-w}}}^{q_w(\vec{s}[w])} < \frac{\eta}{(k-2)T_w\cdot t_w(\vec{s}[w-1])}
			\end{equation}
			
		then \Cref{claim:conditional-general} implies that 
		
		    \[\Pr[\varphi\land \lnot\psi] < \eta.\]
		
		We can ensure that \cref{eq:w-sun-growth} holds just by setting the $q_w(\vec{s}[w])$
		to be sufficiently large in terms of the $q_w(\vec{r}[w])$ for all $\vec{r}$ lexicographically after $\vec{s}$.
		This lexicographical order ensures that we can satisfy these inequalities while not getting any cyclic dependencies between constants (in fact, this is exactly why we had the parameters $q_w$ depend on the number of $v$-sunflowers found for $v<w$).

		Combining the above inequality with \cref{eq:prob-bound2} and \cref{eq:gapped} implies that 
		
			\[\Pr[\varphi] \le \Pr[\varphi\land\psi] + \Pr[\varphi\land\lnot\psi] < (\rho - \eta) + \eta = \rho\]
			
		so we can return {\bf NO}.
\end{description}

In the remainder of the proof, we show the algorithm is correct when it halts outside the loop in step \ref{step:outside}.

As noted in the comments for the algorithm, in each pass through the loop, the value of $\vec{r}$  lexicographically increases.
Because each entry of $\vec{r}$ is bounded above by some constant it must be the case that the loop runs for at most $O(1)$ iterations.
So suppose the algorithm completes the loop without halting.
There are two ways this could happen.

First, it could be the case that we set $r_1 = t_1$,
meaning we found a set of $t_1$ large 1-sunflowers.
In this case, either the cores of the sunflowers form a small hitting set for the clauses of $\varphi\land\psi$, which lets us argue that we can return YES, or we can extract a simple subformula whose fraction of satisfying assignments is strictly less than $\rho$ (as in the final case in the proof of \Cref{thm:thr-3-sat}) and we can return NO.

The other possibility is that the loop terminated because we set $r_w = t_w(\vec{r}[w-1])$ for some $w\ge 2$.
In this case, we stopped because we have many large $w$-sunflowers.
The analysis for this case is much trickier, and involves arguing that the structure of $\psi$ and its large number of clauses forces it to satisfy $\Pr[\psi] < \rho$.
Then by \cref{eq:prob-bound2} we can return NO in this case.

We now show that the algorithm behaves correctly in these two cases (i.e. we prove that if the algorithm halts at step \ref{step:outside} it returns the correct answer).

\begin{description}
	\item[Many 1-Sunflowers: $r_1 = t_1$ ]\hfill\\
	
	Suppose we exit the loop because we set $r_1 = t_1$.
	In this case, $F$ contains at least $t_1$ distinct clauses with just one literal each (corresponding to the 1-clauses in $\psi$).
	Let $L$ be the set of literals among these 1-clauses of $\psi$.
	
	Consider first the case that there is a consistent set $S$ of $t_1$ literals hitting every clause.
	Then we claim that $S = L$ is forced, so that the hitting set is made up of precisely the literals appearing as 1-clauses in $\psi$.
	In fact, we must have $F = \psi$ and the formulas are completely equal.
	This is because after a clause consisting of a single literal $\ell$ is added, 
	we assert $\ell$ is true in $F$, so all other clauses with $\ell$ disappear from $F$ and no future clauses with $\ell$ can ever be added.
	
	In this case we have 
		\[\Pr[F] = \Pr[\varphi\land\psi] = (1/2)^{t_1}.\]
	If $\rho = (1/2)^{t_1}$ we can return YES,
	and if $\rho > (1/2)^{t_1}$ we can return NO.
	
	Otherwise, 	if there is no consistent set of $t_1$ literals hitting every clause,
	it means that either the set of literals $L$ appearing in the 1-clauses of $\psi$ are not consistent,
	or it means that $F$ has a clause that does not contain any literal from $L$.
	
	In the first case we just have $\Pr[F] = \Pr[\varphi\land\psi] = 0$,
	so since $\Pr[\varphi\land\lnot\psi] < \alpha < \rho$ by \Cref{eq:conditional-small},
	we have 
	$\Pr[\varphi] < \rho $
	by \cref{eq:prob-bound2} which means we can return NO.
	
	In the second case, consider the subformula formed by taking the 1-clauses in $\psi$
	together with an extra clause of $F$ that shares no literals with them.
	Then the fraction of satisfying assignments in this subformula is at most 
		\[\frac{1}{2^{t_1}}\cdot \grp{1 - \frac{1}{2^k}} = \rho-\beta.\]
		
	Then by \Cref{prop:subformula} we deduce that 
		\[\Pr[\varphi\land\psi] \le \rho - \beta.\]
		
	Combining this with \Cref{eq:conditional-small} and \cref{eq:prob-bound2} as before we deduce that 
	
		\[\Pr[\varphi] \le \Pr[\varphi\land\psi] + \Pr[\varphi\land \lnot\psi] < (\rho - \beta) + \beta = \rho\]
		
	so we can return NO.	
	
	\item[Many Larger Weight Sunflowers  ]\hfill\\
	
	If the algorithm never set $r_1 = t_1$,
	then it must have exited the loop because it set $r_w = t_w(\vec{r}[w-1])$ for some weight $w>1$.
	
	So $\psi$ contains at least $t_w(\vec{r}[w-1])$ distinct clauses of width $w$,
	and we ended with setting values for $\vec{r}$ such that $r_v\le t_v(\vec{r})-1$ for $v<w$.
	
	Define $\psi'$ to be the subformula of $\psi$ formed by the last $r_w = t_w(\vec{r}[w-1])$ clauses of width $w$ which were added to the formula. 
	By the rule for updating the entries of $\vec{r}$ at the end of step \cref{step:while-loop}\cref{step:sun-add}, we know that for all $v<w$, any $v$-clause that ended up in $\psi$ must have been added before any of the clauses of $\psi'$ were added to $\psi$.

	Our goal is to show that $\Pr[\psi'] < \rho - \alpha$. Provided this is true, by \Cref{prop:subformula} we  have 
		\[\Pr[\varphi\land\psi] < \Pr[\psi] \le \Pr[\psi'] < \rho - \alpha\]
	which we can then combine with \Cref{eq:conditional-small} and \cref{eq:prob-bound2} to deduce that 
		\[\Pr[\varphi] \le \Pr[\varphi\land\psi] + \Pr[\varphi\land\lnot\psi]
		< (\rho - \alpha) + \alpha = \rho\]
	so that returning NO is correct.
	
	To show that $\Pr[\psi'] < \rho-\alpha$, we argue in three steps.
	First, we show that $\psi'$ has no large sunflowers of weight less than $w$.
	Second, we use the absence of large sunflowers to argue that each literal appears at most a bounded number of times in $\psi'$.
	Finally, we use this last property to prove that $\psi'$ contains a large variable disjoint set, which forces $\Pr[\psi']$ to be small.
	Intuitively, our approach is a sort of reverse sunflower extraction argument.
	Whereas previously we used \Cref{lm:sun-k-cnf} to argue that if a formula has no small disjoint set it must have a large sunflower, we will now prove that if a formula has no small sunflowers it must have a large disjoint set.
	
	Throughout the remainder of the proof, we let $\vec{r}$ be the vector storing the final settings of $r_v$ after exiting the loop from \ref{step:while-loop} of the algorithm,
	and for convenience write $Q_v = q_v(\vec{r}[v])$.

	\begin{claim}
		\label{claim:pulling-back-sun}
		The formula $\psi'$ has no $v$-sunflower of size $Q_v$ for any $v\le w-1$.
	\end{claim}
	\begin{proof}
	    Intuitively, this result holds because in step
	    \ref{step:while-loop}\ref{step:sun-add}, before we add the core of a $w$-sunflower to $\psi$, we first check that none o the literals of the core belong to a large $v$-sunflower for any $v<w$.
	
		Suppose to the contrary that $\psi'$ has a $v$-sunflower of size $Q_v$.
		Let $C_1, C_2, \dots, C_{Q_v}$ be the clauses of this $v$-sunflower,
		where $C_{1}$ is the clause that was most recently added to $\psi'$.
		Just before $C_1$ was added, the clauses $C_2, \dots, C_{Q_v}$ were all clauses of $F$.
		Moreover, $F$ had a $w$-sunflower, call it $S_1$, of size $q_w(r_1, r_2, \dots, r_{w-1}, s)$ for some $s\le r_{w}$, with core $C_1$.
		
		By definition, each clause in $S_1$ is of the form $C_1\lor D$ where $D$ is a clause of width at most $k-w$.
		Removing $C_1$ from each clause of $S_1$, the result is a variable disjoint set.
		In particular, ignoring the literals appearing in the core $C_1$, each variable of the formula appears at most once among the literals in $S_1$.
		
		There at most $2w(Q_v-1)$ literals which share a variable with some  clause from $C_2, \dots, C_{Q_v}$, since each clause has width $w$.
		Consequently, by the above discussion, if $S_1$ has at least $2w(Q_v-1) + 1$ clauses, it will contain a clause of the form $C' = C_1\lor D$ such that $D$ shares no variables with any of $C_2, \dots, C_{Q_v}$.
		Since $C_1$ forms a $v$-sunflower with $C_2, \dots, C_{Q_v}$, we deduce that $C', C_2, \dots, C_{Q_v}$ also form a $v$-sunflower.
		
		Consequently, if $S_1$ has size at least $2w(Q_v-1) + 1$, then there will be some clause $C_1'$ in $S_1$ with the property that $C_1'$ together with $C_2, \dots, C_{Q_v}$ forms a $v$-sunflower (this is because the latter $Q_v-1$ clauses together account for at most $2w(Q_v-1)$ ``bad literals'' that a clause of $S_1$ would need to avoid to be able to add on to the $v$-sunflower).
		
		We can ensure that $S_1$ has size at least $2w(Q_v-1) + 1$ by taking $q_w(r_1, r_2, \dots, r_{w-1}, s)$ for any choice of $s$ to be large enough in terms of $q_v(\vec{r}[v])$ for all $v<w$.
		This can be done, for example, by taking 
			\[q_w(r_1, r_2, \dots, r_{w-1}, t_w(\vec{r}[w-1])) >  2w\grp{q_v(r_1, \dots, r_v) - 1} + 1\]
			
		for all $v<1$ and make sure that 
		\[q_w(r_1, r_2, \dots, r_{w-1}, 0) > q_w(r_1, r_2, \dots, r_{w-1}, 1) > \dots\] 
		is a decreasing sequence in the final argument.
		
		Given this condition, we see that some of the literals of $C_1$ belong to a $v$-sunflower of size at least $Q_v$.
		However, if this were true, the check in step \ref{step:while-loop}\ref{step:sun-add} would have found this $v$-sunflower, and then we would not have added $C_1$ (instead, we would have added some smaller weight sunflower).
		So this contradicts the behavior of the algorithm, and $\psi'$ has no $v$-sunflower of size $Q_v$ as claimed.
	\end{proof}
	
	\begin{claim}
		\label{claim:bounded-repetition}
		No literal appears more than 
			 \[(w-1)!\cdot 2^{w-1}\prod_{k=1}^{w-1}\grp{Q_k - 1} \]
		times in $\psi'$.
	\end{claim}

	\begin{proof}
		We show that more generally, this claim holds for a $w$-CNF $\psi'$ with $2\le w\le k-1$ provided $\psi'$ satisfies the conclusion of \Cref{claim:pulling-back-sun}.
	We induct on $w$, using similar arguments to the proof of the classical sunflower lemma \cite{original-sun}.
		
		First consider the base case of $w=2$.
		Take any literal $\ell$ in $\psi'$.
		Let $S$ be the set of clauses which contain $\ell$.
		If we remove $\ell$ from each clause of $S$,
		we get a 1-CNF $S_\ell$ of distinct clauses (if two of the clauses were equal, then $\psi'$ would have duplicate clauses, but this cannot occur because the algorithm never adds the same sunflower core twice).
		
		Since each variable can only be represented twice in $S_\ell$ (as a literal and its negation), if this set of clauses has size at least $|S_\ell| \ge 2(Q_1 - 1) + 1$ then $S_\ell$ has a variable disjoint set of size at least $Q_1$.
		But the size of $S_\ell$ is just the number of appearances of $\ell$ in $\psi'$.
		If we add $\ell$ back into these clauses, we recover a 1-sunflower of size $Q_1$ in $S$, and thus $\psi'$.
		But this contradicts the condition from \Cref{claim:bounded-repetition} that the formula has no 1-sunflower of size at least $Q_1$.
		Thus $\ell$ appears at most $2(Q_1 - 1)$ times, which proves the base case.
		
		Now, take some integer $w\le k-2$ with $w > 2$.
		For the inductive step, we assume that the result holds for $(w-1)$-CNF formulas, and prove it must hold for $w$-CNF formulas $\varphi'$ as well.
		
		As before, take an arbitrary literal $\ell$, let $S$ be the set of clauses containing $\ell$, and let $S_\ell$ be the set formed by removing $\ell$ from each clause of $S$.
		The formula $S_\ell$ is a $(w-1)$-CNF on distinct clauses.
		Moreover, a $v$-sunflower in $S_\ell$ corresponds  to a $(v+1)$-sunflower in $S$ (by adding the literal $\ell$ to the core).
		Consequently, since we assume that $\psi'$ satisfies the conclusion of \Cref{claim:pulling-back-sun}, we know that $S_\ell$ has no 1-sunflower of size $Q_2$, no 2-sunflower of size $Q_3$, etc.
		Thus by the inductive hypothesis we get that any literal appears at most $P = (w-2)!\cdot 2^{w-2}\prod_{k=2}^{w-1}\grp{Q_k-1}$ times in $S_\ell$.
		
		Take a maximal variable disjoint set $D$ of clauses from $S_\ell$.
		By maximality, every clause in $S_\ell$ shares some variable with a clause from $D$.
		Since each clause in $D$ has width at most $w-1$ and each literal appears in $S_\ell$ at most $P$ times, each clause of $D$ shares a variable with at most $2(w-1)P$ clauses in $S_\ell$.
		This means that
		    \[|S_\ell| \le 2(w-1)P\cdot |D|.\]
		
		Now, the clauses in $D$ form a 0-sunflower in $S_\ell$.
		Adding the literal $\ell$ back into these clauses forms a 1-sunflower of size $|D|$ in $S$.
		By assumption, $S$ has no 1-sunflower of size $Q_1$.
		Thus $|D|\le Q_1 - 1$.
		Substituting this bound into the above inequality, we deduce that
		    \[|S_\ell| \le 2(w-1) P\cdot (Q_1 - 1) =(w-1)!\cdot 2^{w-1}\prod_{k=1}^{w-1}\grp{Q_k-1}.\]
		    However, the size of $S_\ell$ is just the number of times $\ell$ appears in $\psi'$.
		So $\ell$ appears at most 
			\[(w-1)!\cdot 2^{w-1}\prod_{k=1}^{w-1}\grp{Q_k-1} \]
		times in $\psi'$, which completes the induction and proves the desired result.
		\end{proof}
	
		Using \Cref{claim:bounded-repetition}, we now show that for the appropriate choice of parameters, $\psi'$ will have a large variable disjoint set.
		Set $t_w(\vec{r}[w-1])$ large enough that
		
			\begin{equation}
			\label{eq:tw-bound}
			\grp{1 - \frac{1}{2^{w}}}^{\lceil t_w(\vec{r}[w-1])/\grp{2w\cdot (w-1)!2^{w-1}\prod_{k=1}^{w-1}\grp{Q_k-1}}\rceil} < \rho - \alpha
			\end{equation}
			
		for all possible values for the $Q_k$ as described above.
		We can do this by setting $t_w(\vec{r}[w-1])$ large enough in terms of all $q_v(r_1, \dots, r_{v})$ values for $v<w$
		(and then once we set this $t_w(\vec{r})$ value we can set values for the corresponding $q_w$ constants).
		This is valid, because $t_w$ depends on the numbers of $v$-sunflowers found for $v < w$ (so we do not have any cyclic dependencies among parameters in this step).
		
		By applying \Cref{claim:bounded-repetition} and picking clauses greedily, we see that the formula $\psi'$ must have a variable disjoint set of size at least 
		
			\[\left\lceil \frac{t_w(\vec{r}[w-1])}{\grp{2w\cdot (w-1)!\cdot 2^{w-1}\prod_{k=1}^{w-1}\grp{Q_k-1}}}\right\rceil.\]
			
		Then by our choice in \cref{eq:tw-bound} and \cref{prop:subformula} we get 
		
			\[\Pr[\psi'] < \rho - \alpha.\]
			
		The previous discussion then implies that we can return NO.
	
\end{description}

This completes the proof of correctness for the entire algorithm.
It remains to check the runtime.

As mentioned before, the loop runs for at most a constant number of iterations.
In each iteration, we may use the sunflower extraction algorithm of \Cref{lm:sun-k-cnf}, scan through the clauses of a formula to remove clauses which have a particular subclause, and check whether a particular clause has literals belonging to a $v$-sunflower of some size $Q$.

The proof of \Cref{lm:sun-k-cnf} shows that calling the sunflower extraction algorithm takes linear time for constant $k$ and constant sized parameters (and all our parameters are constant, since $k$ is constant and $\rho$ has bounded denominator).
We can certainly remove all clauses which have a particular subclause in linear time, since this just takes a constant amount of work for each clause in the formula.
Finally, to check whether a particular clause has literals belonging to a $v$-sunflower of some size $Q$, we check at most $2^k$ subsets of literals.
For each such subset, we can perform the check by considering the set of clauses containing these literals, removing this subset of literals from each clause, and then checking if there exists a variable disjoint set of size at least $Q$ among the resulting clauses.

This last problem, of finding some number of pairwise disjoint sets from a family of subsets, is referred to as the \emph{(Disjoint) Set Packing} problem in the literature. 
Our algorithm needs to solve a Set Packing instance where each set has size at most $k$ and we are looking for at least $Q$ mutually disjoint sets.
Directly applying known fixed-parameter tractable algorithms for Disjoint Set Packing, such as the algorithm from \cite{fpt-mset-packing} for example, this can be done in linear time for constant $k$ and $Q$.

Thus the algorithm runs in linear time as claimed.
\end{proof}

\subsection{Computing the Higher-Order Bits of \#\texorpdfstring{$k$}{k}SAT}
\label{subsec:high-order}

As a simple consequence of our threshold SAT algorithm, we observe that for any $k$-CNF we can compute the ``higher order bits'' of the number of satisfying assignments in polynomial time.
More precisely, given an integer $N\in [0,2^n]$, it has a unique binary representation of the form
    \[N = \sum_{j=0}^{n} b_j 2^{n-j}\]
for some binary digits $b_j\in\set{0,1}$.
Given any integer $t$, the first $t+1$ of these digits $b_0, \dots, b_t$ are said to be the $t+1$ \emph{most significant bits} of $N$.

\begin{corollary}
    \label{cor:most-sig}
    For any positive integers $k$ and $t$, and a $k$-CNF $\varphi$, we can compute the $t+1$ most significant bits of $\ssat{\varphi}$ in $O_{k,t}(|\varphi|)$ time.
\end{corollary}

\begin{proof}
    Let the number of satisfying assignments of $\varphi$ have binary representation
        \begin{equation}
        \label{eq:ssat-binary}
        \ssat{\varphi} = \sum_{j=0}^{n} b_j 2^{n-j}.
        \end{equation}

    First, we check if $\varphi$ is empty or not.
    If it is empty, then $\varphi = \top$ is a tautology, so $\ssat{\varphi} = 2^n$ and we can just return $b_0 = 1$ and $b_1 = \dots = b_t = 0$ to get all $t+1$ most significant bits of $\ssat{\varphi}$.
    
    Otherwise $\varphi$ is not always true, so $\ssat{\varphi} < 2^n$ and $b_0 = 0$.
    
    We now run the following procedure for stages $i = 1, \dots, t$. Each stage determines the $i$-th higher-order bit. At the beginning of stage $i$, we assume we have already computed $b_0,\dots, b_{i-1}$ and that $b_0 = 0$.
    Set 
        \[\rho_i = \frac 1{2^i} + \sum_{j=1}^{i-1} \frac{b_j}{2^j}.\]
    Then $\rho_i$ can be written as a rational with denominator at most $2^i \le 2^t$.
    So by \Cref{thm:thr-k-sat}, we can solve \thr{\rho_i}{$k$} on $\varphi$ in $O_{k,t}(|\varphi|)$ time.
    
    If we find that $\ssat{\varphi} \ge \rho 2^n$, then we set $b_i = 1$.
    Otherwise we set $b_i = 0$.
    We then proceed to the next stage.
    After stage $t$ terminates, we return the values $b_0, \dots, b_t$ computed.
    
    This procedure is correct by a simple induction on the stage number $i$.
    We already computed $b_0$ at the beginning.
    
    Suppose at the beginning of stage $i$ we have correctly found the $i$ most significant digits $b_0, \dots, b_{i-1}$ of the total count satisfying assignments $\ssat{\varphi}$.
    Then if $b_i = 1$, by \cref{eq:ssat-binary} we have 
        \[\ssat{\varphi} \ge \sum_{j=0}^{i} b_j 2^{n-j} = \grp{\frac{1}{2^i} + \sum_{j=0}^{i-1} \frac{b_j}{2^j}}2^n = \rho_i 2^n\]
    so $\varphi$ is as YES instance for \thr{\rho_i}{$k$}.
    Similarly, if $b_i = 0$, by \cref{eq:ssat-binary} we have 
            \[\ssat{\varphi} = \grp{\sum_{j=0}^{i-1} \frac{b_j}{2^j} + \sum_{j=i+1}^{n} \frac{b_j}{2^j}}2^n < \grp{ \sum_{j=0}^{i-1}\frac{b_j}{2^j} + \sum_{j= i+1}^{\infty}\frac{1}{2^j}  }2^n = 
            \grp{ \sum_{j=0}^{i-1}\frac{b_j}{2^j} + \frac 1{2^i}  }2^n = \rho_i 2^n\]
    so $\varphi$ is a NO instance for \thr{\rho_i}{$k$}.
    
    Thus at each stage $i$, our algorithm correctly computes the next most significant bit of $\ssat{\varphi}$.
    So by induction at the end of $t$ stage we will have correctly computed the $t+1$ most significant bits of the number of satisfying assignments in $\varphi$ as claimed.
    The claimed runtime bound holds because we make $t+1$ calls to routines taking $O_{k,t}(|\varphi|)$ time per call.
\end{proof}

\section{The Complexity of Inference Problems over k-CNFs}
\label{sec:e-and-maj}

In this section we explain how our algorithms for \MAJkSAT{$k$} have interesting implications for the more general inference problems discussed in \Cref{subsec:overview}. 

\subsection{E-MAJ-\texorpdfstring{$k$}{k}SAT} 

For reference, we recall the definition of \EMAJSAT, stated slightly differently for this section:

\begin{quote}
\EMAJSAT: Given a CNF $F$ over $n+n'$ variables $\vec{x}=x_1,\ldots,x_n$ and $\vec{y}=y_1,\ldots,y_{n'}$, determine if there is an assignment $A$ to the $\vec{x}$-variables such that $F(A,\vec{y})$ has at least $2^{n'-1}$ satisfying assignments.
\end{quote}

In the context of \EMAJSAT, we say the $\vec{x}$ are \emph{existential variables} and the $\vec{y}$ are \emph{probabilistic variables}. Observe that when $n=0$, \EMAJSAT\ is equivalent to \MAJSAT. 

Our first result generalizes the Theorem~\ref{thm:maj2sat} showing \MAJkSAT{$2$} $\in \P$, to prove that \EMAJkSAT{$2$} is in $\P$. In fact a slightly more general result holds. 

\begin{theorem}
	\label{thm:emaj-2sat}
	\EMAJkSAT{$2$} $\in \P$. In particular, for every $\rho \in (0,1)$, we can determine in $n^{O(\log(1/\rho))}$ time whether there is an assignment to the existential variables such that at least a $\rho$-fraction of the assignments to probabilistic variables are satisfying.
\end{theorem}

\begin{proof} Let $\rho \in (0,1)$. We are given a $2$-CNF $F(\vec{x},\vec{y})$ with existential variables $\vec{x}=x_1,\ldots,x_n$ and probabilistic variables $\vec{y}=y_1,\ldots,y_{n'}$, and wish to know if there is an assignment $A$ to $\vec{x}$ such that at least a $\rho$-fraction of the possible assignments to $\vec{y}$ satisfies $F(A,\vec{y})$. Call an assignment $A$ {\bf good} if it satisfies this property.

For brevity, in the following we will refer to a ``$\vec{y}$-variable'' as a variable $y_i$ from $\vec{y}$, and a ``$\vec{y}$-literal'' as a literal ($y_i$ or $\neg y_i$) from $\vec{y}$.

We partition the $2$-clauses of $F$ into three sets. $P_x$ contains those clauses with two $x_i$ variables, $P_{x,y}$ contains those clauses with one $x_i$ variable and one $y_j$ variable, and $P_y$ contains those clauses with two $y_j$ variables. 

Over the clauses in $P_y$, we find a maximal disjoint set $S$ of clauses, analogously to Theorem~\ref{thm:maj2sat}. If $|S| > 1+\log_{4/3}(1/\rho)$, then we can answer NO, as the subformula $P_y$ of $F(A,\vec{y})$ is already satisfied by less than a $\rho$-fraction of the possible assignments to $\vec{y}$. Otherwise, $|S| \leq 1+\log_{4/3}(1/\rho)$. By trying all possible satisfying assignments to the variables of $S$ and arguing as in Theorem~\ref{thm:maj2sat}, we can express $P_y$ as a decision tree $T_y$ of size at most $2^{|S|} \leq \poly(1/\rho)$, whose inner nodes are labeled with variables from $S$, and whose leaves are labeled with $1$-CNFs over $\vec{y}$ (the decision tree represents all assignments over the ``hitting set'' $S$, and the $1$-CNFs represent the reduced $2$-CNF formula after the variables in the hitting set are assigned).

Now we consider the $P_{x,y}$ subformula. For any assignment $A$ to the variables $\vec{x}$, let $L_A$ be the set of literals of $\vec{y}$ that are implied by assigning the $\vec{x}$ variables according to $A$ in $P_{x,y}$ (for example, if $(x_i \vee y_j)$ is a clause of $P_{x,y}$ and $x_i$ is set false in $A$, then $y_j$ is put in $L_A$). 
Our task of finding a good assignment $A$ is thus equivalent to finding an $A$ to the $\vec{x}$ variables such that the 2-CNF $P_x$ is satisfied, and $T_y$ conjoined with the set of $\vec{y}$-literals $L_A$ has at least $\rho \cdot 2^{n'}$ satisfying assignments. 

Observe that for any good assignment $A$ we must have $|L_A| \leq \log_2(1/\rho)$, since otherwise the fraction of satisfying assignments in $F(A,\vec{y})$ is already at most $1/2^{|L_A|} < \rho$. Thus it must be that at most $\log_2(1/\rho)$ distinct $\vec{y}$-literals are implied in $P_{x,y}$ by a good assignment $A$. 

Let us guess the set of $\vec{y}$-literals $L^{\star}$ that are implied in $P_{x,y}$ by a good assignment $A$ to $\vec{x}$. If $F$ is a YES-instance, then $|L_A| \leq \log_2(1/\rho)$ and there are only $n^{O(\log(1/\rho))}$ possible guesses for the literals in $L^{\star}$, which we can enumerate one-by-one. We can verify $L^{\star}$ by doing two checks which connect $L^{\star}$ to $P_x$ and to $P_y$, respectively.
\begin{itemize}
    \item[(a)] First, we check that some assignment $A$ to $\vec{x}$ implies \emph{exactly} the literals in $L^{\star}$ to be true in $P_{x,y}$ (no other literals from $\vec{y}$ are forced in $P_{x,y}$), and $A$ satisfies $P_x$. 
    
    For every $(\ell_i \vee \ell'_j)$ in $P_{x,y}$ where $\ell_i$ is a literal over $\vec{x}$ and $\ell'_j$ is a  $\vec{y}$-literal {\bf not} in $L^{\star}$, we must have that $\ell_i$ is true according to our guess (otherwise, $\ell'_j$ would be forced true, but $\ell'_j \notin L^{\star}$). Therefore, we can set true all literals over $\vec{x}$ variables that appear in clauses with literals of $\vec{y}$ that are not in $L^{\star}$.
    
    After doing so, the clauses $(\ell_i \vee \ell'_j)$ remaining in $P_{x,y}$ are such that $\ell'_j \in L^{\star}$. As we are guessing that $\ell'_j$ is implied true by $A$ in $P_{x,y}$, it must be that $\ell_i$ is false for at least one clause that $\ell'_j$ appears in. Thus, for the set of clauses of $P_{x,y}$ containing $\ell'_j \in L^{\star}$
    \[S=\{(\ell_1 \vee \ell'_j),\ldots,(\ell_t \vee \ell'_j),\}\]
    we replace $S$ in $P_{x,y}$ with the single long clause
    \[(\neg \ell_1 \vee \cdots \vee \neg \ell_t).\]
    This replacement is valid, as we are guessing that $\ell'_j$ is implied true in $P_{x,y}$, so at least one of the $\ell_i$ must be false (and each clause in $S$ will be satisfied if $\ell'_j$ is true). As $|L^{\star}| \leq \log_2(1/\rho)$, at most $\log_2(1/\rho)$ such long clauses are added to $P_{x,y}$. 
    
    Finally, enumerating all $n^{O(\log(1/\rho))}$ possible ways to choose one literal from each of the $(\neg \ell_1 \vee \cdots \vee \neg \ell_t)$ clauses 
    (which every assignment that forces $L^{\star}$ must satisfy), we then check (in polynomial time)
    that $P_x$ is satisfiable on the remaining $\vec{x}$ variables. If this is true for some assignment to the long clauses, we say the check \emph{passes} (otherwise, the check fails, and we try a different $L^{\star}$).

    \item[(b)] Second, we check that the number of satisfying assignments to the remaining subformula $F(A,\vec{y})$ is at least $\rho \cdot 2^{n'}$, assuming $L^{\star}$ is the set of $\vec{y}$-literals implied by $A$ in $F_{x,y}$.
    
    As argued above, this equals the number of satisfying assignments to $T_y$ conjoined with $L^{\star}$. To compute this, we conjoin the set of literals $L^{\star}$ with each of the $1$-CNFs on the leaves of $T_y$,  then compute the sum over all leaves $\ell$ in $T_y$ of the number of SAT assignments for the $1$-CNF at leaf $\ell$. Therefore we can compute the $\#$SAT value in $O(n)$ time, and verify whether it is at least $\rho \cdot 2^{n'}$. If so, we say that this check \emph{passes}.
\end{itemize}
Finally, our algorithm outputs YES if and only if both checks pass for some guessed set $L^{\star}$. This completes the algorithm, and the proof.
\end{proof}

Although \EMAJkSAT{$2$} turns out to be solvable in polynomial time, we can show that \EMAJkSAT{$3$} is $\NP$-complete. 

\begin{theorem} For all $k \geq 3$, 
\EMAJkSAT{$k$} is $\NP$-complete.
\end{theorem}

\begin{proof} First, the problem is in $\NP$ for every $k$, because we could guess a satisfying assignment to the existential variables, plug it in, then determine whether the remaining formula has a majority (or a $\rho$-fraction, for constant $\rho$) of satisfying assignments in polynomial time, by \Cref{thm:maj3sat}.

Now we prove $\NP$-hardness. Given a $k$-CNF $F$ on variables $x_1,\ldots,x_n$, make an \EMAJkSAT{$k$} instance $F' = F \wedge (x_{n+1} \vee x_{n+2})$ with $n$ existential variables and the two probabilistic variables $x_{n+1}$ and $x_{n+2}$. If $F$ is satisfiable, then there is an assignment to the first $n$ variables of $F'$ such that the probability a random assignment satisfies $(x_{n+1} \vee x_{n+2})$ is at least $3/4 > 1/2$. Otherwise, if $F$ is unsatisfiable, then every assignment to the first $n$ variables leads to a probability of zero on the remaining formula (the remaining formula is false).
\end{proof}

\subsection{MAJ-MAJ-\texorpdfstring{$k$}{k}SAT}

Recall the \MAJMAJSAT\ problem, as defined in the introduction:

\begin{quote}
\MAJMAJSAT: \emph{Given $n$, $n'$, and a formula $F$ over $n+n'$ variables,
	do a majority of the assignments to the first $n$ variables of $F$ 
	yield a formula where the majority of assignments to the remaining $n'$ variables
	are satisfying?}
\end{quote}

Since \MAJkSAT{$k$} is in $\P$ for all $k$ (Theorem~\ref{thm:main}), we can already conclude that \MAJMAJkSAT{$k$} is in $\PP$ (a significant improvement over the obvious $\PP^{\PP}$ upper bound). Just as with \EMAJkSAT{$2$}, we can say more in the case of $k=2$.

\begin{theorem}
	\label{thm:majmaj-2sat}
	\MAJMAJkSAT{$2$} $\in \P$. Moreover, given a $2$-CNF $F$ and parameters $\rho,\sigma \in (0,1)$, we can determine if
	\[\Pr_{a \in \{0,1\}^n}\left[\Pr_{b \in \{0,1\}^{n'}}[F(a,b)=1] \geq \sigma \right] \geq \rho\] in $\poly(1/\rho) \cdot n^{O(\log(1/\sigma))}$ time. In fact, when the algorithm answers YES, it returns the exact number of $a$ such that $\Pr_b[F(a,b)=1]\geq \sigma$.
\end{theorem}

\begin{proof} We proceed very similarly as the algorithm for \EMAJkSAT{$2$} (Theorem~\ref{thm:emaj-2sat}), but need to make some changes to avoid having to solve the $\#\P$-complete $\#2$SAT problem.

Let $F(\vec{x},\vec{y})$ be a $2$-CNF with variables $\vec{x}=x_1,\ldots,x_n$ and variables $\vec{y}=y_1,\ldots,y_{n'}$. Given $\rho, \sigma \in (0,1)$, we wish to know whether or not \[\Pr_{a \in \{0,1\}^n}\left[\Pr_{b \in \{0,1\}^n}[F(a,b)=1] \geq \sigma \right] \geq \rho.\]
As in \Cref{thm:emaj-2sat}, we partition the clauses of $F$ into three sets: $P_x$ contains those clauses with two $x_i$ variables, $P_{x,y}$ contains those clauses with one $x_i$ variable and one $y_j$ variable, and $P_y$ contains those clauses with two $y_j$ variables. 

Over the formula $P_x$, we find a maximal disjoint set $S_x$ of clauses (analogously to Theorem~\ref{thm:emaj-2sat} and Theorem~\ref{sec:maj2sat}), and over $P_y$, we find a maximal disjoint set $S_y$ of clauses. If either $|S_x| > 1+\log_{4/3}(1/\rho))$ or $|S_y| > 1+\log_{4/3}(1/\sigma)$, then we can answer NO. 
In more detail, this is because in the first case, the subformula $P_x$ of $F(\vec{x},\vec{y})$ is only satisfied by less than a $\rho$-fraction of the possible assignments $a$ to $\vec{x}$, regardless of how $\vec{y}$ is set, so the fraction of $a \in \{0,1\}^{n'}$ such that $(\exists b)[F(a,b)=1]$ holds is less than $\rho$. Therefore the fraction of $a$ such that $\Pr_b[F(a,b)=1] \geq \sigma$ holds is also less than $\rho$. In the second case, $P_y$ is only satisfied by less than a $\sigma$-faction of the possible assignments to $\vec{y}$ regardless of how $\vec{x}$ is set.

Otherwise, we have $|S_x| \leq 1+\log_{4/3}(1/\rho)$ and $|S_y| \leq 1+\log_{4/3}(1/\sigma)$. Enumerating all satisfying assignments to the variables of $S_x$ and $S_y$, and arguing as in Theorem~\ref{thm:maj2sat}, both $P_x$ and $P_y$ can be expressed as decision trees $T_x$ and $T_y$ of a particular form. $T_x$ has size at most $2^{|S_x|} \leq \poly(1/\rho)$ with inner nodes labeled by variables from $S_x$, and leaves labeled by $1$-CNFs over the variables $\vec{x}$. $T_y$ has size at most $2^{|S_y|} \leq \poly(1/\sigma)$, inner nodes labeled with variables from $S_y$, and leaves labeled by $1$-CNFs over $\vec{y}$.

We have processed the $P_x$ and $P_y$; now we turn to handling $P_{x,y}$. As in Theorem~\ref{thm:emaj-2sat}, for any assignment $A$ to the variables $\vec{x}$, let $L_A$ be the set of literals of $\vec{y}$ that are implied by assigning the $\vec{x}$ variables according to $A$ in $P_{x,y}$. We say that an assignment $A$ to $\vec{x}$ is {\bf good} if $\Pr_b[F(A,b)=1] \geq \sigma$. We want to determine whether the fraction of good assignments is at least $\rho$ or not.

For any good assignment $A$, it must be that $|L_A| \leq \log_2(1/\sigma)$, since otherwise the fraction of satisfying assignments in $F(A,\vec{y})$ is at most $1/2^{|L_A|} < \sigma$ (so $A$ is not good). Thus there are $n^{O(\log(1/\sigma))}$ possible choices for the set $L_A$
(note that every good assignment $A$ is associated with exactly one set $L_A$: this is trivial, but important to note for the count of SAT assignments).

For each of these choices $L^{\star}$, we process $P_y$ similarly as in Theorem~\ref{thm:emaj-2sat}. First of all, assuming $L^{\star}$ is the set of literals forced in $P_{x,y}$, we can determine the exact number of assignments to the remaining variables of $\vec{y}$ that satisfy $P_y$ in $\poly(n)$ time, by asserting the literals of $L^{\star}$ at each leaf of the decision tree $T_y$ and solving $\#$SAT on the resulting decision tree in $\poly(1/\sigma)\cdot \poly(n)$ time (this is completely analogous to what happens in Theorem~\ref{thm:emaj-2sat}).
Let $N_y(L^{\star})$ be the number of satisfying assignments obtained.

Given $L^{\star}$, we also set all $\vec{x}$-variables that appear in clauses of $P_{x,y}$ with $\vec{y}$-literals that are \emph{not} in $L^{\star}$. For the literals $\ell \in L^{\star}$, in Theorem~\ref{thm:emaj-2sat}, we derived $|L^{\star}|$ long clauses over $\vec{x}$ that must to be true in order for $L^{\star}$ to be correct, and we simply picked one literal from each of these clauses. To get a proper count of the assignments $a \in \{0,1\}^n$ such that $\Pr_b[F(a,b)]\geq \sigma$, we have to work a little harder. We will use an inclusion-exclusion approach, standard in exponential-time algorithmics~\cite{FominK10}. 
In particular, for long clauses $C_1,\ldots,C_t$, we have
\begin{align} \label{eqn:maj2sat-ie}
    \sum_{a \in \{0,1\}^n} (T_x(a) \wedge C_1(a),\ldots,C_t(a))
    = \sum_{S \subseteq [t]}(-1)^{|S|}\cdot \sum_{a \in \{0,1\}^n} \left(T_x(a) \wedge \bigwedge_{i \in S} (\neg C_i(a))\right).
\end{align}
Note the sum on the LHS is only over assignments to the existential variables $\vec{x}$. In our case, $t = |L^{\star}| \leq O(\log(1/\sigma))$, so the number of terms on the RHS is only $2^{O(\log(1/\sigma)}\leq \poly(1/\sigma)$. Each such term asserts a $1$-CNF formula $\bigwedge_{i \in S} (\neg C_i(a))$ over the variables $\vec{x}$. As in the case of $T_y$ above, we can enumerate all the assignments to $\vec{x}$ satisfying $(T_x(a) \wedge \bigwedge_{i \in S} (\neg C_i(a))$ in $\poly(1/\rho)\cdot \poly(n)$ time. Let $N_x(L^{\star})$ be the number of assignments obtained by evaluating \eqref{eqn:maj2sat-ie}.

Finally, we use these $N_x(L^{\star})$ and $N_y(L^{\star})$ values to determine the answer. Let $S_{L^{\star}}$ be the set of good assignments $A$ to $\vec{x}$ such that $P_x(A) \wedge P_{x,y}(A)$ evaluates to precisely the conjunction of all $\vec{y}$-literals in $L^{\star}$. Observe that for different $\vec{y}$-literal sets $L^{\star}$ and $(L')^{\star}$, we must have $S_{L^{\star}} \cap S_{(L')^{\star}} = \emptyset$, and the union over all $S_{L^{\star}}$ contains all good assignments to $\vec{x}$. 
This is because the ones which are not good are those for which $L^{\star}$ is the set of implied $\vec{y}$-literals, but those implied literals force $F$ to be false for less than a $\sigma$-fraction of the $\vec{y}$-assignments.
We observe that the condition
	\[\Pr_{a \in \{0,1\}^n}\left[\Pr_{b \in \{0,1\}^n}[F(a,b)=1] \geq \sigma \right] \geq \rho\]
is equivalent to the condition
\begin{align}
    \label{eq:majmaj2sat-final}
\sum_{\text{set~}L^{\star} : |L^{\star}| \leq O(\log(1/\sigma)} N_x(L^{\star})\cdot [N_y(L^{\star}) \geq \sigma \cdot 2^{n'}] \geq \rho \cdot 2^n,
\end{align}
where $[P] := 1$ if the condition $P$ is true, and is $0$ otherwise. In particular, the sum in \eqref{eq:majmaj2sat-final} is equivalent to enumerating over all assignments $A$ to $\vec{x}$ in the union of all $S_{L^{\star}}$, but only counting those $A$ which result in at least $\sigma \cdot 2^{n'}$ $\vec{y}$-assignments being true in $F(A,\vec{y})$. These are precisely the good assignments of $F$.

Therefore, we can determine YES or NO for our given instance in $n^{O(\log(1/\rho))}\cdot \poly(1/\rho,n)$ time, by computing the sum on the LHS of \eqref{eq:majmaj2sat-final}.
\end{proof}

What is the complexity of \MAJMAJkSAT{$k$} for $k > 2$? Our inclination is to believe that \MAJMAJkSAT{$k$} is in $\P$ for all $k \geq 3$, but have not yet extended our other algorithms to show this.

\section{Algorithms and Hardness For Two Variants of MAJORITY-SAT}

We have shown surprising positive results for solving \MAJSAT\ on bounded-width CNFs. In this section, we prove results on the complexity of natural variations of \MAJkSAT{$k$}: algorithms and hardness for \MAJkSAT{$k$} with one arbitrary-width clause, and \GMAJkSAT{$k$} which asks whether the number of satisfying assignments to a $k$-CNF is strictly greater than $\rho \cdot 2^n$ for a fraction $\rho \in (0,1)$.

\subsection{MAJORITY-SAT With One Long Clause}
\label{subsec:big-clause}

Here, we prove the following theorem from the introduction.

\begin{reminder}{\Cref{thm:one-extra-clause}} Deciding \MAJSAT\ over $k$-CNFs with one extra clause of arbitrary width is in $\P$ for $k=2$, $\NP$-hard for $k=3$, and $\PP$-complete for $k\geq 4$.
\end{reminder}

We begin with the case of $k \geq 4$.

\begin{theorem} \label{thm:maj4sat-one-extra-clause} It is $\PP$-complete to determine if a given $k$-CNF $F$ with \emph{one} arbitrary width clause and $n$ variables has $\ssat{F} \geq 2^{n-1}$, for all $k \geq 4$.
\end{theorem}

\begin{proof} Verifying that the problem is in $\PP$ is straightforward. To prove $\PP$-hardness, suppose we are given a parameter $t \in \{1,\ldots,n\}$ and a $3$-CNF formula $F=C_1 \wedge \cdots \wedge C_m$ on $n$ variables (where $n$ is even) and want to know if $\ssat{F} \geq 2^{n-t}$ or not. Bailey-Dalmau-Koliatis~\cite{BaileyDK07} show this problem is $\PP$-complete.
	
Introduce new variables $x_{n+1},y_1,\ldots,y_t$, and consider the formula 
\[F' = (x_{n+1} \vee C_1) \wedge \cdots \wedge (x_{n+1} \vee C_m) \wedge  (\neg x_{n+1} \vee y_1 \vee \cdots \vee y_t).\] Observe that $F'$ is $4$-CNF with the exception of one clause of length $t+1$.
Considering the cases where $x_{n+1}$ is false and $x_{n+1}$ is true, we have
\[\ssat{F'} = 2^t \cdot \ssat{F} + 2^{n+t}\cdot(1-1/2^t).\] As a fraction of all possible assignments, this is 
\[\rho = \ssat{F}/2^{n+1} + 1/2 - 1/2^{t+1}.\]
Observe that $\ssat{F} \geq 2^{n-t}$ if and only if
\[\rho \geq 1/2 + 2^{n-t}/2^{n+1} - 1/2^{t+1} = 1/2.\] Therefore $F'$ is an instance of \MAJSAT\ if and only if $\ssat{F} \geq 2^{n-t}$.
\end{proof}

Next, we show that \MAJSAT\ over $3$-CNFs with one  long clause is $\NP$-hard. 

\begin{theorem} If we can determine whether a given $3$-CNF $F$ with one arbitrary-width clause and $n$ variables satisfies $\ssat{F} \geq 2^{n-1}$ in polynomial time, then $\P = \NP$.
\end{theorem}

\begin{proof} Using the same reduction as in Theorem~\ref{thm:maj4sat-one-extra-clause}, given any $2$-CNF $F$ and integer $t$, we can reduce the problem of determining whether $\ssat{F} \geq 2^t$ to determining whether a given $3$-CNF $F'$ with one long clause has $\ssat{F'} \geq 2^{n-1}$. 
Therefore, if there were a polynomial-time algorithm for determining $\ssat{F'} \geq 2^{n-1}$, there would also be a polynomial-time algorithm for determining whether $\ssat{F'} \geq 2^{t}$ or not, for $2$-CNF formulas. 

Making $O(\log n)$ calls to such a  polynomial-time algorithm (until we find a $t$ such that $\ssat{F'} \geq 2^t$ but $\ssat{F'} < 2^{t+1}$), we can approximate the number of satisfying assignments to any $2$-CNF, within a factor of $2$. However, this problem is $\NP$-hard~(see for example Theorem 4.1 in \cite{Zuckerman96}).

\end{proof}

It seems the following problem is probably  $\PP$-complete: 
\begin{quote}
Given a $2$-CNF $F'$ and integer $t$, decide whether $\ssat{F} \geq 2^t$.
\end{quote}
If so, the above hardness result can be improved to $\PP$-hardness.\footnote{The difficulty with directly using the fact that $\#2$SAT is $\#\P$-hard~\cite{Valiant_2-SAT} is that all proofs we know that reduce from $\#3$SAT to $\#2$SAT significantly alter the number of satisfying assignments. So, a reduction from the 3-CNF version to the 2-CNF version is not immediate.}

Nevertheless, \MAJkSAT{$2$} remains in $\P$, even with $O(\log n)$ extra clauses.

\begin{theorem}
    \MAJSAT\ over $2$-CNFs with $O(\log n)$ additional arbitrary width-clauses is decidable in polynomial time.
\end{theorem}

\begin{proof} 
Let $F$ be a given $2$-CNF conjoined with $O(\log n)$ arbitrary-width clauses, and let $\rho \in (0,1)$ be given. We wish to determine if $\ssat{F} \geq \rho \cdot 2^n$. First, we imagine executing the algorithm of Theorem~\ref{thm:maj2sat} on just the $2$-CNF part of $F$, ignoring the arbitary-width clauses for now. The algorithm either reports NO (in which case we should also report NO) or it reports the exact number of SAT assignments to the $2$-CNF part. In the latter case, the algorithm computes this number by enumerating over all assignments to an $O(1)$-size set of variables, and solving $\#$SAT on $1$-CNFs obtained by each assignment. We can think of this algorithm as in fact converting the given $2$-CNF into a decision tree of depth $O(1)$, whose leaves are labeled by $1$-CNF formulas. 

To solve $\#$SAT for the entire instance (including the $O(\log n)$ long clauses), we apply the Inclusion-Exclusion Principle. A well-known algorithm for solving $\#$CNF-SAT on $m$ clauses~\cite{FominK10} works by repeatedly applying the equation
\begin{align}
    \label{eqn:ie}
    \ssat{F \wedge C} = \ssat{F} - \ssat{F \wedge \neg C}
\end{align}
to the clauses $C$ of a given formula. Letting $F'$ be the $2$-CNF part of the instance, and letting $C_1,\ldots,C_k$ be the $k \leq O(\log n)$ long clauses, applying \eqref{eqn:ie} repeatedly yields
\begin{align}
\label{eqn:ie2}
\ssat{F' \wedge C_1 \wedge \cdots \wedge C_k} = \sum_{S \subseteq [k]} (-1)^{|S|} \cdot \#\text{SAT}\left(F' \wedge \bigwedge_{i \in S} (\neg C_i)\right).
\end{align}
On the RHS of \eqref{eqn:ie2}, only the negations of long clauses appear; these are simply $1$-CNF formulas. For each $F' \wedge \bigwedge_{i \in S} (\neg C_i)$, we can assert the $1$-CNF $\bigwedge_{i \in S} (\neg C_i)$ at each leaf of our decision tree for $F'$, and count the satisfying assignments for each $F' \wedge \bigwedge_{i \in S} (\neg C_i)$ in polynomial time.
\end{proof}

\subsection{Greater-Than MAJORITY-SAT}
\label{subsec:greater}

In this section, we consider the ``greater than'' version of \MAJSAT:
\begin{quote}
\GMAJkSAT{$k$}: \emph{Given a $k$-CNF $F$ on $n$ variables, does it have greater than $2^{n-1}$ satisfying assignments?} 
\end{quote}

Although both \GMAJSAT\ and \MAJSAT\ are both $\PP$-complete for \emph{general} CNF formulas, it is not at all obvious that the two problems should have the same complexity for $k$-CNF formulas. Their complexities turn out to be quite different (assuming $\P \neq \NP$). The following theorems summarize this section.

\begin{reminder}{\Cref{thm:gtmaj-easy}} For all $k \leq 3$, \GMAJkSAT{k} is in $\P$.
\end{reminder}

\begin{reminder}{\Cref{thm:gtmaj-hard}} For all $k \geq 4$, \GMAJkSAT{k} is $\NP$-complete.
\end{reminder}

We begin with the case of \GMAJkSAT{$2$}, as it is the simplest.

\begin{proposition} \GMAJkSAT{$2$} is in $\P$. Moreover, the problem of determining if the fraction of satisfying assignments is greater than $\rho$ is in $\P$, for every constant $\rho \in (0,1)$. 
\end{proposition}

\begin{proof} The algorithm given for \MAJkSAT{$2$} (Theorem~\ref{thm:maj2sat}) can be used directly to solve \GMAJkSAT{$2$}: when the fraction of satisfying assignments is at least $\rho$, the algorithm actually counts the number of assignments exactly.
\end{proof}

The algorithm for \MAJkSAT{$3$} (Theorem~\ref{thm:maj3sat}) can also be modified to solve the greater-than version. For simplicity, we will fix the fraction $\rho$ to be $1/2$ in the following, but the results below hold for every constant $\rho \in [1/2,1]$.

\begin{theorem} \GMAJkSAT{$3$} is in $\P$. 
	\label{thm:gt-maj3sat}
\end{theorem}

\begin{proof} To solve \GMAJkSAT{$3$}, we slightly modify the algorithm for \MAJkSAT{$3$} in Theorem~\ref{thm:maj3sat}. Observe that every NO-instance of \MAJkSAT{$k$} is also a NO-instance of \GMAJkSAT{$k$}, so we only have to potentially modify the algorithm when it answers YES.

When a literal $\ell$ occurs in every clause of the $3$-CNF, we cannot yet conclude YES in the case of \GMAJSAT\ because such a variable only guarantees that \emph{at least} half of the assignments are satisfying, and not necessarily \emph{greater than} half. However, if we decide satisfiability on the $2$-CNF obtained by removing $\ell$ from every clause, that will decide \GMAJkSAT{$3$} for the overall formula. The rest of the algorithm works with no serious modification: in the other cases, we either conclude the instance has less than half of the satisfying assignments and stop, or we can count the number of assignments exactly and use that to determine the YES or NO answer.
\end{proof}

\begin{corollary}
	There is a polynomial-time algorithm for deciding whether or not a given $3$-CNF formula on $n$ variables has exactly $2^{n-1}$ satisfying assignments.
\end{corollary}

\begin{proof} For a given $3$-CNF $F$, observe that $\#\text{SAT}(F) = 2^{n-1}$ if and only if $F$ is a NO-instance of \GMAJkSAT{$3$} and a YES-instance of \MAJkSAT{$3$}.
\end{proof}

Similarly, we can extend the \thr{\rho}{$3$} algorithm from \Cref{sec:thr3sat} to show that that for every constant threshold $\rho$, \gtthr{\rho}{$3$} is in $\P$.
Hence the above result extends for fractions beyond $1/2$.

\begin{theorem} For every threshold $\rho\in (0,1)$ whose denominator is bounded above by a constant, the  \gtthr{\rho}{3} problem
is in $\P$. 
\end{theorem}
\begin{proof}
    Let $\varphi$ be an arbitrary input 3-CNF on $n$ variables.
    We begin by running the \thr{\rho}{3} algorithm of \Cref{thm:thr-3-sat} on $\varphi$.
    If the routines reports NO we can return NO for the  \gtthr{\rho}{3} as well, because fewer than a $\rho$ fraction of assignments are satisfying.
    
    Otherwise, $\varphi$ is a YES instance for \thr{\rho}{$k$}.
    Inspecting the proof of \Cref{thm:thr-3-sat} from \Cref{sec:thr3sat} shows that in this case,
    the routine will have actually returned a consistent set $S$ of $t\le \log(1/\rho)$ distinct literals such that the clauses of $\varphi$ can be partitioned into sets of clauses $\tilde{\varphi}$ and $\varphi'$ with the following properties:

    \begin{itemize}
        \item[(a)] Every clause in $\tilde{\varphi}$ contains a literal from $S$ and every literal from $S$ occurs in some clause of $\tilde{\varphi}$.
        \item[(b)] The formula $\varphi'$ has a constant size maximal disjoint set $S$.
    Moreover, for any assignment $A$ to $S$, the formula $\varphi'_A$ induced from $\varphi$ by setting values according to $A$ also has a constant size maximal disjoint set.
    This is equivalent to saying that $\varphi'$ can be represented as a constant size decision tree with 1-CNFs at the leaves.
    This decision tree representation is already computed by the \thr{\rho}{3} algorithm.
    \item[(c)]
    The clauses in $\varphi'$ do not contain any of the literals from $S$.
    \item[(d)]
    Finally, if we view $\varphi'$ as a formula on $n-t$ variables, we have $\ssat{\varphi'} \ge \rho 2^n$.
    \end{itemize}
        
    Using the decision tree representation from (b) together with \Cref{prop:1cnf}, we can compute the number of satisfying assignments in $\varphi'$ exactly.
    By (c) the satisfying assignments of $\varphi'$ correspond bijectively to the satisfying assignments of $\varphi$ which set all of the literals to be false.
    
    So, if $\ssat{\varphi'} > \rho 2^n$ already, then we can return YES, since $\ssat{\varphi} \ge \ssat{\varphi'}$.
    
    Otherwise, we have $\ssat{\varphi'} = \rho 2^n$.
    
    Now, for each proper subset $T\subset S$, let $N_T$ denote the number of satisfying assignments of $\varphi$ which set the literals in $T$ to true and the literals in $S\setminus T$ to false.
    Then by the above discussion we have 
    
        \[\ssat{\varphi} = \ssat{\varphi'} + \sum_{T\subset S} N_T = \rho 2^n + \sum_{T\subset S} N_T.\]
        
    Thus we should return YES if and only if $N_T > 0$ for some proper subset $T$ of $S$.
    
    For each subset $T$, let $\tilde{\varphi}_T$ be the formula on $n-t$ variables obtained from $\tilde{\varphi}$ by setting the literals of $T$ to be true and those of $S\setminus T$ to be false.
    Then by definition $N_T$ is just the number of satisfying assignments to the formula $\varphi'\land \tilde{\varphi}_T.$
    
    However, we can decide satisfiability of these formulas just by solving SAT on a constant number of formulas 2-CNFs.
    This is because $\tilde{\varphi}_T$ is a 2-CNF for each $T$ and $\varphi'$ has a small decision tree representation.
    This means we can detect if $\varphi'\land \tilde{\varphi}_T$ is satisfiable by looping over assignments to the constant size disjoint sets in subformulas of $\varphi'$, and then for each such partial assignment, checking if the conjunction of the corresponding 1-CNF leaf of $\varphi'$ and $\tilde{\varphi}_T$ with the the additional partial assignment is satisfiable.
    
    Doing this for all $2^t - 1 \le O(1)$ subsets $T$ lets us check if any $N_T > 0$.
    If there is some positive $N_T$ we return YES.
    Otherwise we return NO.
\end{proof}

Finally, we turn to \GMAJkSAT{$k$} for $k \geq 4$.
It is practically trivial to prove that \GMAJkSAT{$4$} is $\NP$-hard, which makes the algorithmic results of this paper all the more surprising! 

\begin{proposition}
	\label{prop:gt-maj4sat-hard}
	\GMAJkSAT{$k$} is $\NP$-hard for all $k \geq 4$.
\end{proposition}

\begin{proof} We reduce from 3SAT. Take any $n$-variable $3$-CNF $F$ that we wish to determine SAT for, and introduce a new variable $x_{n+1}$ which is then included in every clause of $F$. This $4$-CNF on $n+1$ variables has greater than $2^n$ satisfying assignments if and only if the original $3$-CNF has at least one satisfying assignment.
\end{proof}

Finally, using the tools developed in the proof of Theorem~\ref{thm:main}, we show that \GMAJkSAT{$k$} is actually contained in $\NP$ for all $k$. Thus the problem is $\NP$-complete for all $k \geq 4$. This is an intriguing development, as it shows that SAT solvers can potentially attack such threshold counting problems.

\begin{theorem} For every integer $k$, \GMAJkSAT{$k$} is in $\NP$. 
\end{theorem}

\begin{proof} We give a nondeterministic linear-time algorithm that can determine whether strictly greater than $1/2$ of all assignments to $k$-CNF formula are satisfying. In what follows, let $\varphi$ be an arbitrary input $k$-CNF.

First, we run the algorithm from \Cref{thm:main} from \Cref{sec:thrksat} with a threshold of $\rho = 1/2$ to solve \MAJkSAT{$k$}. If the algorithm reports NO on $\varphi$, then we can also return NO for $\varphi$, since in that case strictly fewer than $1/2$ of its assignments satisfy $\varphi$.

Otherwise, $\varphi$ is a YES instance for \MAJkSAT{$k$}, and we wish to determine if it is also a YES instance of \GMAJkSAT{$k$}. Inspecting the proof of \Cref{thm:main}, we see that in this case the algorithm constructs a $(k-2)$-CNF $\psi$ and a formula $\varphi'$ such that
\begin{itemize}
    \item[(a)] $\Pr[\varphi'] \ge 1/2$, and
    \item[(b)] for every assignment $a$, $\varphi'(a) = \varphi(a) \land\psi(a)$. Note this implies $\Pr[\varphi'] = \Pr[\varphi\land\psi]$.
\end{itemize}
Furthermore, at least one of the following cases holds.
\begin{itemize}
    \item {\bf Case 1: Exact Count} \\
    We can determine the exact number of satisfying assignments of $\varphi'$ in polynomial time, because the algorithm decomposed $\varphi\land\psi$ into a decision tree with $1$-CNFs at the leaves.
    \item {\bf Case 2: Covered by a Literal} \\ There is a single literal that appears in every clause of $\varphi'$. 
    \end{itemize}
We show how these conditions are sufficient for proving that $\varphi$ has greater than $1/2$ satisfying assignments, using the equation
\begin{align}\label{eq:gtmaj-ksat-np}
    \Pr[\varphi] = \Pr[\varphi\land\psi] + \Pr[\varphi\land\lnot \psi] =  \Pr[\varphi'] + \Pr[\varphi\land\lnot \psi]. 
\end{align}
{\bf Case 1.} Suppose we fall in the first case.
Then we can exactly determine $\Pr[\varphi']=\Pr[\varphi\land\psi]$.
If this probability is greater than $1/2$, then we can return YES by \eqref{eq:gtmaj-ksat-np}. 
Otherwise, $\Pr[\varphi'] = 1/2$ by condition (a) above. Therefore, $\Pr[\varphi] > 1/2$ if and only if $\Pr[\varphi\land\lnot \psi] > 0$. 
The latter can be checked by nondeterministically guessing a variable assignment $A$, and returning YES if and only if $A$ satisfies $\varphi\land\lnot \psi$.

{\bf Case 2.} Otherwise, we fall in the second case, where every clause of $\varphi'$ contains a common literal $\ell$.

Let $\tilde{\varphi}$ be the formula formed by removing all occurrences of $\ell$ from $\varphi'$ (equivalently, we are setting $\ell$ to be false in $\varphi'$).

Nondeterministically guess an assignment $A$. If $A$ satisfies $\tilde{\varphi}$, then we return YES, because in \[\Pr[\varphi] \geq \Pr[\varphi'] \ge \Pr[\varphi'\land \ell ] + \Pr[\tilde{\varphi}] > 1/2.\] 
The last inequality follows from the observations that (i) $\Pr[\varphi'\land \ell ] = 1/2$ since setting $\ell$ true automatically satisfies $\varphi'$, and (ii) the existence of the satisfying assignment $A$ proves that  $\Pr[\tilde{\varphi}] > 0$. 

If $A$ does not satisfy $\tilde{\varphi}$, then we guess an assignment $A'$ and return YES if $A'$ satisfies $\varphi\land\lnot \psi$. This is valid because in this case,
\[\Pr[\varphi] = \Pr[\varphi'] + \Pr[\varphi\land\lnot \psi] > 1/2,\] because $\Pr[\varphi'] \geq 1/2$ by item (a) above. If $A'$ does not satisfy $\varphi\land\lnot \psi$ then we return NO.

Suppose that every nondeterministic branch of the above procedure returns NO, i.e., $\tilde{\varphi}$ and $\varphi\land\lnot \psi$ are both unsatisfiable. In this case, we have 
\begin{align*}
\Pr[\varphi] &= \Pr[\varphi'] + \Pr[\varphi\land\lnot \psi] = \Pr[\varphi'\land \ell ] + \Pr[\tilde{\varphi}] + 0 = \frac 12,
\end{align*}
because $\Pr[\varphi'\land \ell ] = 1/2$ and $\Pr[\tilde{\varphi}] = \Pr[\varphi\land\lnot \psi] = 0$. 
Thus $\varphi$ is a NO instance of \GMAJkSAT{$k$}.
\end{proof}

Essentially the same proof lets us extend the above result to other thresholds.

\begin{theorem} 
\label{thm:gt-thr-np}
For every integer $k$ and threshold $\rho\in (0,1)$ whose denominator is bounded above by a constant, the  \gtthr{\rho}{$k$} problem
is in $\NP$. 
\end{theorem}

\begin{proof} 
    Let $\varphi$ be an arbitrary input $k$-CNF.
    
    First, solve \thr{\rho}{$k$} on $\varphi$ using the algorithm of \Cref{thm:main}.
    If $\varphi$ is a NO instance for \thr{\rho}{$k$} we can immediately return NO.
    
    Otherwise, $\varphi$ is a YES instance for \thr{\rho}{$k$}.
    
    In this case, examining the proof of \Cref{thm:main} from \Cref{sec:thrksat} shows that the algorithm will have constructed a $(k-2)$-CNF $\psi$ and formula $\varphi'$ whose set of satisfying assignments equals the set of satisfying assignments of $\varphi\land\psi$, such that either
        \begin{itemize}
        \item we know the exact number of satisfying assignments of $\varphi'$ (because we were able to decompose the solution space of $\varphi\land\psi$ as a sum of solutions to 1-CNFs), or
        \item for some positive integer $t\le \log(1/\rho)$, we have a consistent set of $t$ literals such that every clause of $\varphi'$ contains some literal from the set.
    \end{itemize}
    
    Moreover, $\varphi'$ has the property that $\Pr[\varphi'] \ge \rho$.
    
    We handle the above two cases separately.

\begin{description}
		\item[Case 1: Exact Count] \hfill\\
		Suppose we have the exact fraction of satisfying assignments in clause of $\varphi'$.
		If this fraction is greater than $\rho$ we can return YES since $\Pr[\varphi] \ge \Pr[\varphi']$.
		
		Otherwise, $\Pr[\varphi'] = \rho$.
		So we have 
				    \[\Pr[\varphi]  = \Pr[\varphi'] + \Pr[\varphi\land\lnot\psi] = \rho + \Pr[\varphi\land\lnot\psi] .\]
		In this case, nondeterministically guess an assignment to $\varphi\land\lnot\psi$.
		If the assignment is satisfying we return YES, and otherwise we return NO.
		This is correct because of the above equation, which shows that we should return YES if and only if $\varphi\land\lnot\psi$ is satisfiable.
		
		\item[Case 2: Small Hitting Set of Literals] \hfill\\
		If we do not fall into the first case, then we must have a consistent set $S$ of $t$ distinct literals with $t\ge \log(1/\rho)$, such that every clause of the formula  $\varphi'$ contains some literal from $S$.
		We assume this set $S$ is a minimal set with this property.
		Then any assignment with sets all the literals to $S$ true automatically satisfies $\varphi'$.
		It follows that 
		    \[\Pr[\varphi'] \ge (1/2)^t.\]
		    
		Thus if $(1/2)^t > \rho$ we return YES.
		
		Otherwise, $\rho = (1/2)^t$ exactly.
		In this scenario, define the formula 
		    \[\tilde{\varphi} = \varphi'\land \lnot\grp{\bigwedge_{\ell\in S} \ell}\]
		which correspond to satisfying assignments of $\varphi'$ which set some literal in $S$ to false.
		
		Since a $\rho$ fraction of assignments to $\varphi'$ simultaneously satisfy $\varphi$ and all the literals in $S$ we must have 
		    \[\Pr[\varphi'] = \rho + \Pr[\tilde{\varphi}].\]
		Consequently we get that 
		    \begin{align*}
		    \Pr[\varphi] &= \Pr[\varphi\land\psi] + \Pr[\varphi\land\lnot\psi] \\
		    &= \Pr[\varphi'] + \Pr[\varphi\land\lnot\psi] \\
		    &= \rho + \Pr[\tilde{\varphi}] + \Pr[\varphi\land\lnot\psi].
		    \end{align*}
		   
		This final equation shows that we should return YES if and only if at least one $\tilde{\psi}$ or $\varphi\land\lnot\psi$ is satisfiable.
		Hence, we nondeterministically guess an assignment to $\tilde{\varphi}$ and to $\varphi\land\lnot\psi$.
		If either of these assignments are satisfying for their respective formulas, we return YES by the above equation.
		Otherwise we return NO, since $\Pr[\varphi] = \rho$.
\end{description}
This completes the proof. \end{proof}

\section{Concluding Thoughts \& Open Problems}\label{sec:conclusion}

There are many interesting open issues left to pursue; here are a few.

\begin{itemize}
	\item {\bf Determine the complexity of \MAJMAJkSAT{$k$} for $k\geq 3$.} For any fixed integer $k\ge 3$, is the \MAJMAJkSAT{$k$} problem $\PP$-complete, in $\PP$, or somewhere in between? 
	We conjecture the problem is in $\P$ for all constant $k \geq 3$, but have not yet extended our methods to prove this result.

	\item {\bf Parameter Dependence.} 
	Although our algorithms for \thr{\rho}{$k$} run in linear time for fixed $k$ and $\rho$, these runtimes grow \emph{extremely} quickly as a function of $\rho$, even for $k=3$ (as noted in \Cref{prop:eff}).
	Is a better dependence on $\rho$ possible, or can we prove that a significantly better dependence is unlikely to exist? Could there be a $\poly(1/\rho)$ dependence, as in the \MAJkSAT{$2$} algorithm?
	
	\item{\bf Threshold Counting Beyond Satisfiability.}
	Are there other natural problems where the counting problem is known to be hard, but the threshold counting problem turns out to admit a polynomial time algorithm? 
	Our results show this phenomenon holds for the counting and threshold counting versions of $k$SAT for constant $k$, but perhaps similar behavior occurs for other problems, such as counting perfect matchings or counting proper $k$-colorings of graphs.

    \item {\bf Variants of Weighted Model Counting.} A natural ``weighted'' extension of the \MAJkSAT{$k$} problem would be: given $\rho \in (0,1)$ and $m$ degree-$k$ polynomials $p_1(x),\ldots,p_m(x) \in {\mathbb Q}[x_1,\ldots,x_n]$, determine if
	\[\sum_{a \in \{0,1\}^n} \prod_{j=1}^m p_j(x) \geq \rho\cdot 2^n.\]
	What does the complexity of this problem look like? In the special case solved in this paper ($k$-CNF), our polynomials have the form $1-C_j$ where $C_j$ is a product of $k$ literals ($x_i$ or $1-x_i$). 
	
	To specialize further (and still fall within the $k$-CNF case), suppose each $p_i$ takes values in $[0,1]$ over all $a \in \{0,1\}^n$, so their product $\prod_j p_j(a)$ is always in $[0,1]$. For constant $\rho \in (0,1)$, can the above sum-product problem be solvable in polynomial time? 
	
	\item {\bf Bayesian inference with $k$-CNFs.} Given two $k$-CNF formulas $F$ and $G$ over a common variable set, and given $p \in (0,1)$, the inference problem is to determine whether \[\Pr_{x}[F(x) = 1 \mid G(x) = 1] \geq p.\] 
	By definition, this is equivalent to determining whether
	\[\frac{\Pr_{x}[(F(x) \wedge G(x)) = 1]}{\Pr_x[G(x)=1]} \geq p.\] 
	Since determining if the denominator is nonzero is already $\NP$-hard for $k=3$, the best we can hope for is to put this problem in $\NP$. To sidestep the division-by-zero issue, we can rephrase the inference problem as determining whether \[\Pr_{x}[(F(x) \wedge G(x)) = 1] \geq p \cdot \Pr_x[G(x)=1].\] Already this problem is interesting for the case where $F$ and $G$ are $2$-CNF.
	
	\paragraph{Algorithms.} The results of this paper imply that the inference problem is in polynomial time when $F$ is $3$-CNF and $G$ is a $1$-CNF. When $\ssat{F \wedge G} \geq 2^n/\poly(n)$, \Cref{thm:maj2sat} implies that the inference problem is in $\P$ for $2$-CNFs regardless of $p$ (because  both sides of the inequality can be counted exactly). Also, if $\ssat{G} \geq 2^n/\poly(n)$ and $p \geq 1/\poly(n)$, then we can solve the inference problem for $2$-CNFs using \Cref{thm:maj2sat}.

	\paragraph{Hardness.} If $G$ is an arbitrary $3$-CNF, and $F$ is a $1$-CNF, then the inference problem is already $\NP$-hard. Deciding \[\Pr_{x}[(F(x) \wedge G(x)) = 1] \geq p \cdot \Pr_x[G(x)=1]\] lets us construct a satisfying assignment for $G$: try both $F=x_1$ and $F=\neg x_1$ with $p=1/2$. Observe that 
	\[\Pr_x[G(x)=1] = \Pr_{x}[(x_1 \wedge G(x)) = 1] + \Pr_{x}[(\neg x_1 \wedge G(x)) = 1],\] so either $\Pr_{x}[(x_1 \wedge G(x)) = 1]  \geq 1/2 \cdot  \Pr_x[G(x)=1]$ or $\Pr_{x}[(\neg x_1 \wedge G(x)) = 1]  \geq 1/2 \cdot \Pr_x[G(x)=1]$. By choosing the larger of the two, we can construct a satisfying assignment for $G$ for each variable one at a time.
	
	The above discussion still does not yet settle the case where $F$ is a $2$-CNF and $G$ is a 2-CNF, and the fractions involved are smaller than $1/\poly(n)$.

\end{itemize}


\bibliographystyle{alpha}
\bibliography{main,extra}

\appendix

\section{Sunflower Extraction Algorithm}
\label{sec:sun}

\begin{reminder}{\Cref{lm:sun-k-cnf}}
	Fix positive integers $Q_0, Q_1, \dots, Q_{k-2}$. There is a computable function and $f$ and an algorithm which runs in at most \[f(Q_0, Q_1, \dots, Q_{k-2})\cdot |F|\]
	time on any given $k$-CNF $F$, which either
	
	\begin{itemize}
		\item 
		produces a $v$-sunflower of size at least $Q_w$ in $F$ for some  $w\in\set{0, 1, \dots, k-2}$ and $v<w$, 
		or
		
		\item 
		produces a collection ${\cal C}$ of $1$-CNF formulas such that $|{\cal C}|\le f(Q_0, Q_1, \dots, Q_{k-2})$ and
		\[\ssat{F} = \sum_{F' \in {\cal C}} \ssat{F'}.\] That is, $\ssat{F}$ equals the sum of $\ssat{F'}$ over all $1$-CNFs $F'$ in ${\cal C}$.
	\end{itemize}
	
\end{reminder}

\begin{proof}[Proof of \Cref{lm:sun-k-cnf}]
	
	The proof proceeds by repeatedly extracting maximal disjoint sets from subformulas of $F$.
	If we ever fail to get a ``small'' maximal disjoint set, it means we have a ``large'' sunflower.
	
	Formally, we begin at stage $a = 0$.
	Write $\Phi_0 =  F$.
	By linearly scanning through the clauses of $F$,
	we 
	take a maximal variable disjoint set $S_0$ of clauses of width $k$ from $\varphi_0$.
	If $|S_0| \ge Q_0$ we return $S_0$ and halt.
	Otherwise, $|S_0| < Q_0$.
	At this point $F = \Phi_0$ is a single $k$-CNF.
	
	We describe an inductive procedure that repeatedly replaces parts of $F$ with collections of CNFs which have smaller width.
	In general, at stage $a$ for $0\le a\le k-1$, we build a tree $\Phi_a$ of depth $a$, where nodes at level $j$ for $0\le j\le a$ are labeled by $(k-j)$-CNFs, obtained from $F$ by assigning values to some of its variables.
	In particular, the leaves of $\Phi_a$ are $(k-a)$-CNFs.
	Each node at level $j$ will also carry the data of a maximal variable disjoint set of clauses of width $(k-j)$ for the associated formula.
	We may identify nodes with their formulas.
	
	Inductively, we maintain the invariant that
	\begin{enumerate}[label=(\roman*)]
		\item 
		for any $j<a$, $\Phi_j$ is identical to the tree formed by the first $j$ levels of $\Phi_a$,
		
		\item
		if a node $\varphi'$ is a child of a node $\varphi$ in $\Phi_a$, then the formula $\varphi'$ was induced from
		$\varphi$ by assigning values to the variables in the maximal variable disjoint set of $\varphi$,
		
		\item 
		the formulas at level $j$ of $\Phi_a$ for any $j\le a$ are $(k-j)$-CNFs,
		
		\item 
		the sum of the number of satisfying assignments to the leaves of $\Phi_a$  equals $\ssat{F}$, and
		
		\item
		the size of a maximal disjoint set at any leaf node of $\Phi_a$ is bounded above by a constant depending only on the values of $a$, $k$, and $Q_j$ for $j\le a$.
	\end{enumerate}
	For the base case of $a=0$ we can verify that these properties hold.
	
	Suppose now that we are at stage $a$ for some $a < k-1$, and we have the data described above.
	Take the tree $\Phi_a$, and consider any leaf node.
	Let the $(k-a)$-CNF formula at this node be $\varphi_a$ and corresponding maximal disjoint set be $S_a$.
	We loop over all at most $\grp{2^{k-a} - 1}^{|S_a|}$ satisfying assignments to variables in $S_a$, and for each assignment produce a new induced formula by setting those values in $\varphi_a$.
	Each of these formulas is added as a child of the node with $\varphi_a$.
	
	Since $S_a$ was maximal on clauses of width $(k-a)$,
	the variables in $S_a$ are a hitting set for all clauses of this width in $\varphi_a$.
	Consequently, every new formula produced in the above way is a $(k-a-1)$-CNF.
	Take an arbitrary such formula $\varphi_{a+1}$ produced in this way.
	Now, extract a maximal variable disjoint set $S_{a+1}$ of clauses with
	width $(k-a-1)$ from $\varphi_{a+1}$.
	We claim that either
	\begin{equation}
	\label{eq:max-dset-growth}
	|S_{a+1}|<  Q_{a+1}\prod_{j=0}^a \grp{(k-j)|S_j|+1}
	\end{equation}
	holds regardless of the assignment to $S_a$ which produced $\varphi_{a+1}$, or we can find a literal disjoint set in $F$ of size at least $Q_{a+1}$.
	
	Indeed, suppose to the contrary that the inequality in \cref{eq:max-dset-growth} does not hold for the set $S_{a+1}$ produced by some assignment.
	Consider the unique path in the tree $\Phi_a$ from the root (node corresponding to $\Phi_0$) to $\varphi_a$.
	This path passes through nodes at levels $j=0, \dots, a$.
	Let $\varphi_j$ denote the formula on this path at level $j$, and let $S_j$ be the associated maximal disjoint set.
	By properties (i) and (ii), we know that $\varphi_{j+1}$ was induced from the formula $\varphi_j$ by assigning values to the variables in $S_j$.
	
	By maximality of $S_j$, the variables in $S_j$ are a hitting set for the clauses of
	width $(k-j)$ in $\varphi_j$.
	Consequently each clause in $S_{j+1}$ (which has width $(k-j-1)$ by construction)
	either is the same as some $(k-j-1)$-clause in $\varphi_j$ 
	or corresponds to a $(k-j)$-clause of $\varphi_j$ with exactly one literal removed.
	Under an assignment to the variables in $S_j$, a literal $\ell$ is removed from a clause in $\varphi_j$ to produce $\varphi_{j+1}$ if and only if $\ell$ was set false in the assignment.
	
	Then since $\varphi_j$ is a $(k-j)$-CNF, there are at most $(k-j)|S_j|$
	literals which can be removed to produce the clauses of $S_{j+1}$.
	There is one additional possibility -- that we produced a clause by removing no literals.
	Repeatedly applying this argument, we see that there are at most
	\[\prod_{j=0}^a\grp{(k-j)|S_j| + 1}\]
	choices of a set of at most $a$ literals which could be removed from a clause of $F$ to produce a clause in $S_{a+1}$.
	Then by averaging we deduce there must exist a single set $L$ of at most $a$ distinct literals such that 
	
	\[\frac{|S_{a+1}|}{\prod_{j=0}^a\grp{(k-j)|S_j| + 1}} \ge Q_{a+1}\]
	clauses in $S_{a+1}$ which pull back to clauses containing all the literals of $L$ in $F$
	(we are using the assumption that the inequality in \cref{eq:max-dset-growth} does not hold).
	This corresponds to a $w$-sunflower of size at least $Q_{a+1}$ for some weight $0\le w\le a+1$, so we can return this set of clauses and halt as claimed.
	We can find and return this sunflower in linear time by scanning through the clauses in $\varphi_j$ which correspond to the clauses in $S_{j+1}$ for $j=a, a-1, \dots, 0$.
	
	Otherwise, \cref{eq:max-dset-growth} holds for all assignments to $S_{a}$.
	
	By repeating this procedure for all leaves of $\Phi_a$ we produce the tree $\Phi_{a+1}$.
	Since we are merely inserting new leaves by partially assigning variables to the leaves of $\Phi_a$, properties (i) and (ii) still hold.
	The above construction also shows that (iii) continues to hold.
	
	By induction, $\ssat{F}$ equalled the sum of the numbers of satisfying assignments to the leaves of $\Phi_a$.
	In the above process, each leaf of $\Phi_a$ produced children by looping over assignments to some subset of its variables, so certainly the number of satisfying assignments at a leaf of $\Phi_a$ is equal to the total count of satisfying assignments among its children in $\Phi_{a+1}$.
	This means that property (iv) still holds.
	
	Finally, the inductive hypothesis and \cref{eq:max-dset-growth} show that property (v) holds.
	
	If we never produce a large sunflower,
	this process terminates once we get to $\Phi_{k-1}$.
	At that point, we can just return formulas at the leaves of the tree,
	which are all 1-CNFs by property (iii).
	This step is correct by property (iv).
	
	The maximum number of leaves in the final tree can be obtained by multiplying upper bounds on the number of children of nodes at each level.
	This is bounded above by 
	
	\[\prod_{a=0}^{k-2} \grp{2^{k-a} - 1}^{|S_a|}\]
	where the sizes of the sets $S_a$ are bounded recursively by $|S_0| < Q_0$ and \cref{eq:max-dset-growth}.
	Because building maximal disjoint sets can be done just by scanning once through the clauses of $F$, this calculation shows that if the parameters $Q_0, \dots, Q_{k-2}$ are all bounded above by some constant, the algorithm takes linear time.
\end{proof}

\begin{remark}
	In the above proof, at each stage we selected maximal variable disjoint sets consisting of the clauses with some fixed length.
	If instead we took maximal variable disjoint sets with no restriction on the width of the clauses picked, the proof would still work.
	All that would change is the averaging argument would be a little different, because between stages it is possible that a clause would lose strictly more than one literal.
	It is  possible this alternate version could me more useful for certain applications.
\end{remark}

\section{An Arithmetic Lemma}
\label{sec:bit}

In this section we prove \Cref{lm:gen-bit-argument}, which we use to argue for ``gaps'' in the number of satisfying assignments of a $k$-CNF, in certain cases.

\begin{reminder}{\Cref{lm:gen-bit-argument}}
 	Let $n$ and $m$ be arbitrary positive integers, and let $\rho \in (0,1)$ be rational, of the form \[\rho = \frac{a}{2^vb}\]
	for unique odd integer $b$, nonnegative integer $v$, and integer $a$ with $\gcd(a,2^vb) = 1$. 	Then for every integer $N$ which is the sum of at most $m$ powers of two, if $N < \rho 2^n$ then $N\le \grp{\rho - \eta}2^n$, for a positive $\eta$ depending only on $a, b, v$, and $m$.
\end{reminder}

As noted in \Cref{rem:greedy} below, the proof simply formalizes the intuition that if we want $m$ powers of $2$ whose sum is as large as possible while still being strictly less than $\rho$, we can just use a naive greedy algorithm.

In the discussion below, it will be useful to have the notion of the order of an element: for any odd positive integer $b$, we define
the order $d = \ord(b)$ modulo 2 to be the smallest positive integer such that $b\mid (2^d-1)$. Euler's totient theorem shows that $d$ is well-defined and satisfies $d\mid \varphi(b)$ for some integer $\varphi(b) \le b-1$.

\begin{proof}[Proof of \Cref{lm:gen-bit-argument}] 

Let $\rho = a / (2^v b)$ be given. We first handle the case where $b=1$, so that $\rho = a/2^v$ is a dyadic rational.
	
	{\bf Case 1: $b=1$.~ \\} For a positive integer $a$, let $a_i$ be such that \[a = \sum_{i=1}^s 2^{a_i}\]
	consisting of $s$ terms with exponents $a_1 > a_2 > \dots > a_s$.
	Then we have
	
		\begin{equation}
		\label{eq:binary-bound}
		N < \rho 2^n = \sum_{i=1}^s 2^{n-v+a_i}.
		\end{equation}
	
	By assumption, $N$ is the sum of at most $m$ powers of two.	
	Without loss of generality, suppose $m \ge s$ (this is OK, since setting $m$ to be larger only weakens the hypothesis).
Set
	
		\[N' = \sum_{i=1}^{s-1} 2^{n-v+a_i} + \sum_{j=1}^{m-s+1}2^{n-v+a_s-j}.\]
		
	We claim that $N\le N'$.
	
	Indeed, the inequality \cref{eq:binary-bound} shows that the powers of two showing up
in the	binary representation of $N$ are all at most $2^{n-v+a_1}$.
Also, not all the terms $2^{n-v+a_i}$ can show up in the binary representation,
since then we would have $N \ge \rho 2^n$.

So let $j$ be the minimum index such that the term $2^{n-v+a_j}$ does not appear in the binary expansion of $N$.
Then we claim that besides powers $2^{n-v+a_i}$ for $i<j$,
no term of the form $2^{n-v+d}$ with $d > a_j$ can appear in the binary expansion of $N$. This is because if $d > a_j$, we have

\[2^d \ge 2^{a_j+1} > \sum_{i=j}^s 2^{a_i}.\]
So, if $N$ had such a term in its binary expansion, we would have

\[N\ge \grp{\sum_{1\le i < j}2^{n-v+a_i}} + 2^{n-v+d} > \sum_{i=1}^s 2^{n-v+d}\]
which contradicts \cref{eq:binary-bound}.

This already shows that if $j < s$, then $N < N'$.
Otherwise $j = s$, and the assumption on $N$ implies that 
	\[N \le \sum_{i=1}^{s-1} 2^{n-v+a_i} + \sum_{x\in X} 2^x\]
for some set $X$ of $m-r+1$ distinct integers each smaller than $n-v+a_r$.
But any such set satisfies
	\[\sum_{x\in X} 2^x \le \sum_{j=1}^{m-s+1} 2^{n-v+a_s-j}\]
which means that $N\le N'$ as claimed.
	
	Thus in this subcase we have 
	\[N\le N' = (\rho - \eta)2^n \]
	for 
	\begin{equation}
	\label{eq:eta-1b}
	\eta = \frac{2^{a_s- m + s- 1}}{2^v}.
	\end{equation}
	The above arguments prove the result when $b = 1$.
	It remains to handle the case where $b > 1$ is an odd integer.
	
	{\bf Case 2: $b>1$.~ \\}
	
	Let $d = \ord(b)$ and set $c = \grp{2^d - 1}/b$.
	Then we have
	
		\[\rho = \frac{a}{2^vb} = \frac{ac}{2^v\grp{2^d-1}}.\]
		
	Write $ac = \grp{2^d-1}q + r$ for some unique choice of nonnegative integers $q$ and $r$ with $r < 2^d-1$.
	Then we can further simplify
		\begin{equation}
		\label{eq:rho-bin-split}
		\rho = \frac{1}{2^v}\cdot \grp{q + \frac{r}{2^d-1}}.
		\end{equation}
	It may be helpful to think of $q$ and $r$ respectively as the integer and fractional parts of $a/b$.
		
	Now take the binary representations
		\[q = \sum_{1\le i \le s} 2^{q_i}\quad
		\text{and}
		\quad
		r = \sum_{1\le i\le t} 2^{r_i}\]
	for nonnegative integers $q_1 > \dots > q_s$ and $d > r_1 > \dots > r_t \ge 0$.
	Note that if $q=0$ we have $s=0$, so that the first sum above can be empty.
	Using the formula for an infinite geometric series, we can expand
		\[\frac{r}{2^{d}-1} = \grp{\sum_{1\le i\le t} 2^{r_i}}\grp{\sum_{j\ge 1} 2^{-jd} }
		= \sum_{j=1}^\infty \sum_{i=1}^t 2^{r_i - jd}.\]
	So if we define an infinite sequence $a_1 > a_2 > \dots$ of integers by taking
	$a_i = q_i$ for $i\le s$ and 
		\[a_{s+tj+i} = r_i-jd\]
	for all nonnegative integers $j$ and $1\le i\le t$,	
	substituting these binary expansions into \cref{eq:rho-bin-split}
	yields
		\[\rho = \sum_{i=1}^{\infty} 2^{a_i-v}.\]
	For this case, set 
		\[N' = \sum_{i=1}^m 2^{n-v+a_i}.\]
	As before, we claim that $N\le N'$.
	
	Since $N < \rho 2^n$,
	the binary expansion for $\rho$ shows that the terms in the binary expansion of $N$ are each at most $2^{n-v+a_1}$.
	If for all $1\le i\le m$ the term $2^{n-v+a_i}$ shows up in the binary expansion of $N$, then $N=N'$ and the claim holds.
	Otherwise there is some minimum index $j$ such that $2^{n-v+a_j}$ does not occur in the binary expansion of $N$.
	We cannot have a term of the form $2^{n-v+d}$ with $d > a_j$ in the binary representation of $N$, since then we'd immediately get $N > \rho 2^n$.
	So in this case
	        \[N\le \sum_{1\le i < j} 2^{n-v+a_i} + \sum_{i=1}^{m-j+1}2^{n+a_j - i} < N'.\]
	Thus in this case too we have
		\[N\le N' = (\rho - \eta)2^n\]
	for
	\begin{equation}
	\label{eq:eta-2}
	\eta = \sum_{i=m+1}^\infty 2^{a_i - v}.
	\end{equation}
%
%
%
\end{proof}

\begin{remark}[Effective Bounds for the Lemma]
    \label{rem:greedy}
    Although \Cref{lm:gen-bit-argument} is stated as an existence result (for any $\rho$ and $m$ there exists some positive $\eta$), inspecting the above proof allows one to get a concrete lower bound for $\eta$ in terms of the binary expansion of $\rho$.
    In particular, the proof shows that for any $\rho$, if we want to find a sum of $m$ powers of $2$ whose sum is strictly less than $\rho$, then the greedy algorithm which keeps selecting the largest power of $2$ while ensuring that the running sum of powers taken so far is less than $\rho$ produces a distinct set of powers of two which maximizes the sum that can be obtained.
    If $\rho = a/b$ with $\gcd(a,b) = 1$, the period of $\rho$ is bounded above by $b-1$.
    Consequently, if we fix $\rho$ but allow $m$ to vary, the above discussion implies that 
		\[\eta \ge \frac{\rho}{2^{(b-1)m}}.\]
\end{remark}

Below we provide some concrete examples of how $\eta$ depends on $\rho$ and $m$, given the discussion in \Cref{rem:greedy}, for a few specific thresholds $\rho$.

\begin{remark}[Examples of Effective Dependence]
    \label{rem:eff} 
	The value of $\eta$ obtained in \Cref{lm:gen-bit-argument} depends 
	in a somewhat complicated fashion	on the binary representations of the numerator and denominator of $\rho$, as
	detailed in equation \cref{eq:eta-2}.
	To get some intuition for how big these gaps are, here are the precise bounds obtained for
	$\eta$ for some specific choices of $\rho$:
	
		\begin{enumerate}
			\item 
			When $\rho = 1/2^v$ is the inverse of a power of two, \cref{eq:eta-1b} shows that 
				\[\eta = \frac{\rho}{2^m}.\]
				
			\item 
			When $\rho = 3/7$, simplifying the expression in \cref{eq:eta-2} shows that
				\[\eta = \frac{\rho}{2^{3m/2}}\]
			when $m$ is even and 
				\[\eta = \frac{(5/3)\rho }{2^{\lceil 3m/2\rceil }}\]
			when $m$ is odd.
			
			The decay is worse in terms of $m$ than in the previous case 
			because the denominator has nontrivial order $\ord(7) = 3 > 1$.
				
			\item 
				When $\rho = 1/2 - 1/2^{v}$, and $m \ge v-1$  we can compute from \cref{eq:eta-1b} that
					\[\eta = \frac{1}{2^{m+2}}.\]
					
				So the behavior in this case is similar to
				what happens when $\rho = 1/2$.
						
			\item 
		When $\rho = 1/3^v$ is the inverse of a power of three, \cref{eq:eta-2} implies that 
		\[\eta \ge  \frac{\rho }{2^{2\cdot 3^{v-1}\lceil m/h\rceil }}\]
		
		where $1\le h < 2\cdot3^{v-1}$ is the number of ones in the binary expansion of $(4^{3^{v-1}}-1)/3^v$.
		
		Like the previous examples, the decay is exponential in $m$.
		However, here the decay is much more rapid in terms of $v$.
		The seems to be representative of a worst case scenario where the denominator of $\rho$ has two as a ``primitive root.''
		
		\item 
		When $\rho = 1/(2^v-1)$, \cref{eq:eta-2} implies that 
			\[\eta = \frac{\rho}{2^{mv}}.\]
		Note that this decay is much more rapid than in the first example
		where $\rho = 1/2^v$, even though the threshold values themselves are close.
		\end{enumerate}
		
\end{remark}
\section{Runtime Dependence of THR3SAT Parameterized by Threshold Values}
\label{sec:eff-parameter}

In this section, we give an upper bound for the dependence of our \thr{\rho}{3} algorithm on $\rho$.

\begin{reminder}{\Cref{prop:eff}}
	Let $\rho\in (0,1)$ be a rational with denominator $b$. 
	Set $t = \lfloor \log(1/\rho)\rfloor$.
	Then there exists $K = \poly(1/\rho)$ such that if we define
		\[c =
	\underbrace{K^{\cdot^{\cdot^{K^{(\log b)}}}}}_{t+2\text{ terms}}\]
	to be a tower of $t+1$ exponentiations of $K$ together with one exponentiation to the $(\log b)^{\text{th}}$ power at the top of the tower, the
	\thr{\rho}{3} algorithm described in the proof of \Cref{thm:main} takes at most 
	\[c |\varphi|\]
	time.
\end{reminder}
\begin{proof}
We carry over the notation and parameters defined in the proof of \Cref{thm:thr-3-sat}.
As noted in \cref{eq:time-bound} from the proof of \Cref{thm:thr-3-sat}, the \thr{\rho}{3} algorithm runs in time asymptotically at most
    \begin{equation}
    \label{appc-runtime}
    \grp{\exp(zq_0) + \exp(zq_1) + \dots + \exp(zq_t) }|\varphi|.
    \end{equation}

So, it suffices to explain what values we should set for each of the $q_r$ parameters.

First we set the value of $q_t$.
From \cref{eq:3qt-bound} and the analysis of {\textbf{case 4}} from the proof of \Cref{thm:thr-3-sat},
the algorithm requires
    \begin{equation}
    \label{appc-q}
    (t+1)\cdot (3/4)^{q_t} \le \min\left(\rho - (1/2)^{t+1}, \rho - (7/8)^z\right).
    \end{equation}

Since $\rho$ is a rational with denominator $b$ and is greater than both $(1/2)^{t+1}$ and $(7/8)^z$, we have \[\rho - (1/2)^{t+1} \ge 1/(b2^{t+1})\] and \[\rho - (7/8)^z \ge 1/(b8^z).\]

Since $t = \lfloor \log 1/\rho\rfloor$ and $z = O(\log 1/\rho)$, we can ensure that \cref{appc-q} holds by taking
    \begin{equation}
    \label{appc-qtqt}
    q_t = \Theta\grp{\log t + \max(t,z) + \log b} = \Theta(\log b)
    \end{equation}
    sufficiently large.
In the above calculation, we used the definitions of $t$ and $z$, and the fact that $1/\rho < b$.

Now, suppose that we have already set the value of $q_r$ for some positive integer $r$.
We show how to determine the value of $q_{r-1}$.

Let $\eta_r$ be the constant $\eta$ obtained from \Cref{lm:gen-bit-argument}
with parameters $\rho$ and $m = \exp(zq_r)$.
By the analysis of {\bf case 3} from the proof of \Cref{thm:thr-3-sat}, the algorithm requires  $q_{r-1}> q_r$ and 
    \[r\cdot (3/4)^{q_{r-1}} < \eta_r.\]

Consequently it suffices to take 
\begin{equation}
\label{appc-qrqr}
q_{r-1} = \Theta(\log t + \log(1/\eta_r))
\end{equation}
large enough, where we used the fact that $r\le t$.
From the effective bound of $\eta_r$ in terms of $\rho$ and $m$ discussed in \Cref{rem:greedy},
we know that 
\[\eta_r \ge \rho / 2^{(b-1)m}.\]

Thus we meet the requirements necessary for the algorithm by setting $q_r$ to satisfy
	\[q_{r-1} = \Theta(\log t + \log(1/\eta_r)) = \Theta(\log t + \log(1/\rho) + b\exp(zq_r)) = \Theta(b\exp(zq_r))\]
or making it larger.

Now, set $q_t$ according to \cref{appc-qtqt} large enough so that $b < \exp(zq_t)$.
Then the algorithm will work according to the above equation if for each of $r = t, t-1, \dots, 1$ we recursively set
    \[q_{r-1} = \Theta(\exp(zq_r)) = \grp{\poly(1/\rho)}^{q_r}.\]

Unfolding this recursion with $q_t = \Theta(\log b)$, we get that for some $K = \poly(1/\rho)$ we have 
    \[q_r \le \underbrace{K^{\cdot^{\cdot^{K^{(\log b)}}}}}_{t-r+1\text{ terms}}\]
for every $r$.
In particular, we have 
\[q_0 \le \underbrace{K^{\cdot^{\cdot^{K^{(\log b)}}}}}_{t+1\text{ terms}}.\]
Using these upper bounds in \cref{appc-runtime}, we see that the expression for the asymptotic runtime of the algorithm is dominated by the first term
	\[\exp(z \cdot q_0)|\varphi|.\]
Consequently, we can bound the overall asymptotic runtime of the algorithm by 
    \[c|\varphi|\]
for \[c  = \exp(z\cdot q_0)\le \underbrace{K^{\cdot^{\cdot^{K^{(\log b)}}}}}_{t+2\text{ terms}}\]
for some sufficiently large $K = \poly(1/\rho)$ as claimed.
\end{proof}

\section{The Complexity of E-MAJ-SAT}

\label{appendix:EMAJSAT-hardness}

Recall that \EMAJSAT\ is defined as follows:

\begin{quote}
\EMAJSAT: \emph{Given $k$, $n$, and a CNF formula $\varphi$ over $k+n$ variables, is there a setting to the first $k$ variables of $\phi$ such that the majority of assignments to the remaining $n$ variables are satisfying assignments?} 
\end{quote}

The first $k$ variables are called ``existential'' while the latter $n$ variables are called ``probabilistic''. Here we prove that that \EMAJSAT\ is $\NP^{\PP}$-complete over $6$-CNFs with one arbitrary width clause, in the hope that researchers in the complexity of probabilistic inference may find this reduction to be a useful replacement. 

\begin{theorem} \EMAJkSAT{$6$} with one arbitrary width clause is $\NP^{\PP}$-complete.
\end{theorem}

\begin{proof} Containment in $\NP^{\PP}$ is straightforward. 

To show completeness, we start with Tor{\'{a}}n's characterization of $\NP^{\PP}$~\cite{Toran91}: every $\NP^{\PP}$ language can be simulated with a deterministic polynomial time machine $M(x,y,z)$ which is said to accept an input $x$ if and only if there is a $y \in \{0,1\}^{p(|x|)}$ such that at least half of the choices of $z \in \{0,1\}^{q(|x|)}$ make $M(x,y,z)$ accept, where $p(n)$, $q(n)$ are polynomials. In the following we let $n=|x|$.

Since \MAJkSAT{$4$} with one long clause is $\PP$-complete (Theorem~\ref{thm:maj4sat-one-extra-clause}) we can, without loss of generality, think of $M$ as guessing a $p(n)$-bit $y$, then constructing a $4$-CNF formula $\phi_{x,y}(z_1,\ldots,z_{r(n)})$ on polynomially-many variables with one long (arbitrary width) clause $(z_{i_1} \vee \cdots \vee z_{i_t})$, accepting if and only if $\phi_{x,y}(z_1,\ldots,z_{r(n)})$ has at least $2^{r(n)-1}$ satisfying assignments. 
In particular, note that in the proof of Theorem~\ref{thm:maj4sat-one-extra-clause}, the arbitrary-width clause of $\phi_{x,y}$, does \emph{not} depend on the values of $x$ and $y$, but only the length $n$. 
We will make a \EMAJSAT\ instance which is a $6$-CNF (and one long clause) with $\poly(n)$ more ``existential'' variables, while keeping the number and semantics of the ``probabilistic'' variables $z_1,\ldots,z_{r(n)}$ to be exactly the same.

First, we consider the encoding of the query $\phi_{x,y}$. We construe the $4$-CNF part of $\phi_{x,y}$ as encoded in variables $a_1,\ldots,a_{c n^4}$ for a constant $c > 1$, where each bit indicates which of the $O(n^4)$ possible $4$-CNF clauses is in $\phi_{x,y}$. In our output formula, we will directly include the long clause $(z_{i_1} \vee \cdots \vee z_{i_t})$ in the formula (it does not depend on $x$ or $y$).

The nondeterministic part of $M$, which guesses a $y$ and constructs the query $\phi_{x,y}(z)$, can be written as a $3$-CNF in the variables $y_1,\ldots,y_{p(n)}$, $a_1,\ldots,a_{c n^4}$, and auxiliary variables $b_1,\ldots,b_{r(n)}$ that arise from the Cook-Levin reduction. These clauses will ``set'' the bits of the \MAJSAT\ query $a_1,\ldots,a_{c n^4}$ based on the nondeterminism of $y$ and the computation. Note that the variables $z_1,\ldots,z_{q(n)}$ of the \MAJSAT\ query do not appear in this $3$-CNF part.

Next, we form clauses to relate the bits $a_i$ encoding the \MAJSAT\ query to the variables $z_i$ of the query formula. For each bit $a_i$ of the query, suppose $a_i$ corresponds to the possible clause $C=(z_{i_1} \vee \cdots \vee z_{i_4})$. Then in our formula we include the DNF formula
\[D_i = ((a_i \wedge z_{i_1}) \vee \cdots (a_i \wedge z_{i_4}) \vee (\neg a_i \wedge z_{i_1}) \vee (\neg a_i \wedge \neg z_{i_1})).\]
Observe that if $a_i = 0$ then  $D_i$ is trivially true, and if $a_1 = 1$, then $D_i = C$. Thus, the bits of the query $a_1 \cdots a_{c n^4}$ exactly determine the set of clauses over the $z_i$ variables of the \MAJSAT\ query. 
Finally, we observe that each DNF formula $D_i$ can be converted into a $6$-CNF $E_i$, by applying distributivity. 

Our \EMAJSAT\ instance sets the variables $z_1,\ldots,z_{q(n)}$ to be ``probabilistic'' and all other variables will be ``existential''.
Then, for every string $y$ such that exactly $K$ of the possible choices of $z$ make $M(x,y,z)$ accept, there is a setting to $b_1,\ldots,b_{r(n)}$ and $a_1,\ldots,a_{cn^4}$ which together with $y$ satisfies the $3$-CNF part of the formula, and the sets of $6$-CNF clauses $E_1,\ldots,E_{cn^4}$ over $a_i$ and $z_i$ variables encode the \MAJSAT\ query correctly, so they are all satisfied by exactly $K$ possible choices of $z$ as well. 
\end{proof}

\end{document}